\documentclass[preprint]{elsarticle}

\usepackage{graphicx}
\usepackage{amsfonts}
\usepackage{amssymb}
\usepackage{amsmath}
\usepackage{amsthm}
\usepackage{algorithm, algpseudocode}
\usepackage{float}

\newfloat{algorithm}{t}{lop}
\floatname{algorithm}{Algorithm}

\begin{document}
\title{Two Particle-in-Grid Realisations on Spacetrees}

\author[dur]{T.\ Weinzierl\corref{cor1}}
\ead{tobias.weinzierl@durham.ac.uk}
\ead[url]{www.dur.ac.uk/tobias.weinzierl}
\cortext[cor1]{Corresponding author}
\author[leu,auc]{B.\ Verleye}
\ead{bverleye@vub.ac.be}
\author[pierre]{P.\ Henri}
\ead{pierre.henri@cnrs-orleans.fr}
\author[leu]{D.\ Roose}
\ead{dirk.roose@cs.kuleuven.be}

\address[dur]{School of Engineering and Computing Sciences, Durham University\\
Stockton Road, Durham DH1 3LE, Great Britain}
\address[leu]{Department of Computer Science, KU Leuven\\
 Celestijnenlaan 200A, B-3001 Leuven, Belgium}
\address[auc]{Vrije Universiteit Brussel\\
 Pleinlaan 2, 1050 Elsene, Belgium}
\address[pierre]{Laboratoire de Physique et Chimie de l'Environnement et de
l'Espace (LPC2E)\\  
CNRS, Universit\'{e} d'Orl\'{e}ans\\
45071 Orl\'{e}ans Cedex 2, France}

\newcommand{\peano}{Peano}
\newcommand{\pic}{PIC}

\newcommand{\pit}{PIT}
\newcommand{\pidt}{PIDT}
\newcommand{\rapidt}{raPIDT}

\newcommand{\spacetree}{spacetree}

\newcommand{\EventEnterCell}{{\footnotesize \tt enterCell}}
\newcommand{\EventLeaveCell}{{\footnotesize \tt leaveCell}}

\newcommand{\EventTouchVertexFirstTime}{{\footnotesize \tt
touchVertexFirstTime}}
\newcommand{\EventTouchVertexLastTime}{{\footnotesize \tt touchVertexLastTime}}
\newcommand{\EventMergeWithMaster}{{\footnotesize \tt mergeWithMaster}}
\newcommand{\EventMergeWithNeighbour}{{\footnotesize \tt mergeWithNeighbour}}
\newcommand{\EventMergeWithWorker}{{\footnotesize \tt mergeWithWorker}}
\newcommand{\EventPrepareSendToMaster}{{\footnotesize \tt prepareSendToMaster}}
\newcommand{\EventPrepareSendToWorker}{{\footnotesize \tt prepareSendToWorker}}

\newtheorem{theorem}{Theorem}[section]

\begin{abstract}
The present paper studies two particle management strategies for
dynamically adaptive Cartesian grids at hands of a particle-in-cell code.
One holds the particles within the grid cells, the other within the grid
vertices.
The fundamental challenge for the algorithmic strategies results from the fact
that particles may run through the grid without velocity constraints.
To facilitate this, we rely on multiscale grid representations.
They allow us to lift and drop particles between different spatial resolutions.
We call this cell-based strategy particle in tree (PIT).
Our second approach assigns particles to vertices describing a dual grid
(PIDT) and augments the lifts and drops with multiscale linked cells.

Our experiments validate the two schemes at hands of an 
electrostatic particle-in-cell code by retrieving the dispersion
relation of Langmuir waves in a thermal plasma.
They reveal that different particle and grid characteristics favour
different realisations.
The possibility that particles can
tunnel through an arbitrary number of grid cells implies that most data is 
exchanged between neighbouring ranks, while very few data is transferred 
non-locally. 
This constraints the scalability as the code potentially has to  
realise global communication.
We show that the merger of an analysed tree grammar with PIDT allows us to
predict particle movements among several levels and to skip parts of this
global communication a priori. 
It is capable to outperform several established implementations based upon
trees and/or space-filling curves.
\end{abstract}

%


\maketitle

\section{Introduction}
\label{introduction}


%
%
Lagrangian-Eulerian descriptions of physical phenomena are used by a wide range
of applications.
They combine the short-range aptitude of particle-based with the
long-range capabilities of grid-based approaches.
Besides their popularity driven by application needs, 
particle-grid methods are also popular in supercomputing.
They are among
the best scaling algorithms today (cf.~for example 
\cite{Hamada:09:GordonBell,Rahimian:2010:GordonBell,Sampath:08:Dendro}).
This scaling mainly relies on two ingredients.
On the one hand, particle-particle interactions often are computationally
expensive in terms of floating point operations with moderate memory footprint.
On the other hand, the particle-grid interaction requires a mapping 
of particles to the grid.
This is a spatial sorting problem.
It either is
performed infrequently, is computationally cheap as the particles move at most
one grid cell at a time, or the grid can be 
constructed efficiently starting from the particles \cite{Sampath:08:Dendro}.

The present paper is driven by a plain electrostatic \pic\ simulation
\cite{Hockney:88:PIC,Verleye:13:PIC} of an unmagnetised plasma.
Its computational profile differs from the previous characteristics as it
solves a partial differential equation (PDE) on an adaptive Cartesian grid which
has to be stored persistently in-between two time steps.
At the same time, the particles do
not interact with each other---they induce very low computational workload---but
may move at very high speed through the grid.
Our work focuses on well-suited data structures and algorithms required for such
a code within a dynamic adaptive mesh (AMR) environment where the simultion is
ran on a distributed memory machine.

%
%
The setup raises an algorithmic challenge. 
While dynamic AMR for PDEs as well as algorithms based upon particle-particle
interactions are exhaustively studied, our algorithm requires a fast 
mapping of particle effects onto the grid and the other way round per time step.
This assignment of particles to the grid changes incrementally.
Yet, some particles might {\em tunnel} several cells per time step:
no particle is constrained to move at most into a neighbouring cell.
In general, the time step size is chosen such that the majority of particles travel
at most one cell per time step.
This avoids the
finite grid instability \cite{Birdsall:05:PlasmaPhysics}, i.e.~numerical, non-physical,
heating of the experiment.
However, 
in our application as well as most non-relativistic plasma and gravitation
applications, suprathermal particles do exist. 
Their velocity is not bounded.

%
%
The present paper proposes the parallel, locally refined, dynamically adaptive
grid to result from a \spacetree\
\cite{Weinzierl:2009:Diss,Weinzierl:11:Peano} 
yielding a Cartesian tessellation.
Particles are embedded into the finest tree level.
The latter is the adaptive grid hosting the PDE. 
A multiscale grid traversal with particle-grid updates then can be realised
via a simple recursive code mirroring a depth-first search
\cite{Weinzierl:15:Peano}.

Two particle realisation variants are studied: 
we either store the particles within the spacetree leaves (particle in tree; \pit) or within the dual
tree (particle in dual tree; \pidt) induced by the spacetree vertices.
The latter is similar to a linked-list approach with links on each
resolution level.
This multiscale linked-list can be deduced on-the-fly, i.e.~is not stored
explicitly but encoded within the tree's adjacency information.
Both particle storage variants render the
evaluation of classical compact particle-grid operators straightforward as each
particle is assigned to its spatially nearest grid entity anytime. 
Tunneling is enabled as particles are allowed to move up and down in 
the whole spacetree.
\pidt\ furthermore can move particles between neighbours. 
Besides neighbours and parent-child relations, no global adjacency data
structure is required.
While \pidt\ induces a runtime overhead and induces a more complex code compared
to \pit, the multiscale linked-list nature of \pidt\ reduces the particle
movements up (lift) and down (drop) in the tree, and particles moving along the links can
be exchanged asynchronously in-between grid traversals.
If we combine \pidt\ with a simple analysed tree grammar
\cite{Knuth:90:AttributeGrammar} for the particle velocities, we can predict lifts and anticipate drops in whole grid
regions.
This helps to overcome a fundamental problem.
Since tunneling is always possible, we need all-to-all communication in every 
time step: each rank has to check whether data is to be received
from any other rank. 
This synchronisation introduces inverse weak scaling. The more particles 
the higher the number of tunnels.
The more ranks the more tunnel checks. 
With a lift prediction, we can locally avoid the all-to-all that we map onto
a reduction (reduction-avoiding \pidt; \rapidt).
We weaken the rank synchronisation.
Rank and process are used here as synonyms. 
For reasonably big parallel architectures, \pidt\ and its extension \rapidt\
thus outperform \pit.
Though all three flavours of particle storage support tunneling, their
performance is not significantly slower than a classic
linked-list approach not allowing particles to tunnel.
The multiscale nature of both \pit\ and \pidt\ as well as its variant \rapidt\
further makes the communication pattern, i.e.~the sequence and choice of
particles communicated via MPI, comprise spatial information. 
Ranks receiving particles that have to be sorted into a local grid anticipate
this presorting and thus can outperform other classic approaches
\cite{Dubey:11:ParticlesInMesh,Shamoto:14:GPUSorting,Sundar:13:HykSort}.

%
%
The remainder is organised as follows. 
We refer to related work before we introduce the algorithmic
steps of electrostatic \pic\ motivating our algorithmic research
(Section \ref{section:pic}).
In Section \ref{section:spacetree}, we describe our spacetree grid paradigm.
Two particle realisation variants storing particles either within the grid cells
or within the vertices form the present work's focus
(Section \ref{section:cell_and_vertex_based}).
Our experimental evaluation starts from a review of the particle movement
characteristics before we study the runtime behaviour of the particle storage
and sorting schemes as well as the scaling.
Finally, we compare our results to three other approaches.
This comparison (Section
\ref{section:comparison}) highlights our contribution with respect to the
particular application from Section \ref{section:pic} as well as the inevitable
cost introduced by the tunneling particles.
A summary and outlook in
Section \ref{section:conclusion_and_outlook} close the discussion.
The appendix comprises a real-world validation run of our code as well as
additional experimental data.

\section{Related and used work}
\label{section:related-work}

Requiring tunneling in combination with a persistent grid holding a PDE solution 
renders many established grid data structures inappropriate or unsuited.
We distinguish three classes of alternatives: on-the-fly (re-)construction,
sorting and linked-list.
Reconstructing the assignment for the prescribed grid from scratch per
time step \cite{Sampath:08:Dendro} is not an option, as the grid is given, holds
data and is distributed.
Local sorting followed by scattering is not an option (see e.g.~\cite{Shamoto:14:GPUSorting,Sundar:13:HykSort} and references therein) as the
decomposition, i.e.~the adaptive mesh, is not known on each individual rank.
The two possible modifications are to rely on replicated grids per rank and then
to map these grids onto the real PDE grid---this approach is not followed-up
here---or to sort locally within the grid and then to scatter those particles
that cannot be sorted locally.
Such a two-stage approach enriches local meta data information with additional
decomposition data (ranks have to know how to scatter the particles) held on
each rank, but may run into network congestion (cf.~results from
\cite{Dubey:11:ParticlesInMesh}) even if we avoid all-to-all collectives \cite{Sundar:13:HykSort} and rely on elegant decomposition paradigms
such as space-filling curves (SFCs).
Finally, standard linked-cell algorithms, cell-based Verlet lists
or hypergraph-based approaches
\cite{Gonnet:07:VerletLists,Mattson:99:ModifiedLinkedList} are not well-suited.
They impose constraints on the velocities of particles, i.e.~they can not handle particles that may travel arbitrarily fast.
Alternative approaches such as overlapping domain decompositions where the grid
holding the particles is replicated on all ranks, schemes where particle and
domain decomposition may differ, schemes where the associativity is re-emended only every $k$ steps or approaches tailored to regular grids only are not considered here 
\cite{Bender:12:PackedMemoryArray,Ihmsen:11:Sorting,Mattson:99:ModifiedLinkedList,Plimpton:03:LBForPIC}.

For the present challenge, we require an algorithm that updates the grid-particle correlation
incrementally.
It has to be fast, i.e.~anticipate the incremental character, but nevertheless
has to support particle escapees.
Hereby, the particle distribution follows the grid decomposition:
particles contained within a grid cell should reside on the compute node
holding the respective grid cell in a non-overlapping domain decomposition
sense.
For this, we rely on a tree data structure where grid cells `point' to
particles.
Our approaches materialise ideas of bucket sorts \cite{Sundar:13:HykSort}.
A tree's multiscale nature allows us to realise the sorting incrementally,
locally and fast without constraints on the particle velocity or the dynamically
adaptive grid structure.
The omnipresent multiscale nature and the lack of grid constraints make it
differ from the \texttt{up\_down\_tree} approach in
\cite{Dubey:11:ParticlesInMesh}.
Our \pidt\ variant holds particles within vertices and thus allows cells to push
particles into their neighbouring cells directly.
No multiscale data is involved.
As such it materialises the linked-cells idea where each cells holds a list of
its neighbours.
As these links are available on all discretisation levels, our \pidt\ variant
can be read as a multiscale linked-cell approach.
Finally, our code holds all tree data non-overlappingly and does not require any
rank to hold the subdomain dimensions of each and every other rank.
This makes the present algorithms differ from classic codes that rely on SFCs to
obtain a fine grid partition and then construct (overlapping) coarse grid
regions bottom-up.
All adjacency data of the domain decomposition is localised.
%
%

Comparisions to several of these aforementioned approaches with tunneling
highlight the present algorithms' strenghts and weaknesses.
Notably, they highlight that our local decomposition and particle handling
is, for many applications, superior to classic SFC
decompositions where the domain decomposition is known globally as well as to
approaches where particles are handed around cyclically or sorted from a coarse
regular mesh into the fine AMR structure \cite{Dubey:11:ParticlesInMesh}.
A comparison to a classic linked-cell code that does not support
tunneling finally reveals the price we have to pay to enable tunneling. 

%
%
Throughout the present paper, we rely on our open source AMR software framework
Peano \cite{Weinzierl:15:Peano,Software:Peano}.
It allows us to rapidely implement and evaluate the present algorithmic ideas.  
\rapidt\ was fed back into this framework and now is
available as black-box particle handling scheme.
Yet, none of the present algorithms is tied to that particular framework.
They work for any spactree-based code---relying on classic quadtree or octree
discretisation, e.g.---and for any spatial dimension as long as two criteria are
met.
First, our algorithms need all resolution levels of the multiple AMR
grids embedded into each other explicitly.
The meshing software has to provide all the mesh resolution levels
to the solver and the solver has to be able to embed data into each and every
level. 
Providing solely the finest mesh, e.g., is not sufficient.
Second, the algorithms have to switch from coarser to
finer levels and the other way round throughout the mesh traversal. 
Notably, a depth-first, a breadth-first or any hybrid traversal on the tree
that is implemented recursively work.  

\section{Use case: a particle-in-cell code}
\label{section:pic}

 \begin{figure}[!ht]
 \centering
 \includegraphics[width=0.32\linewidth]{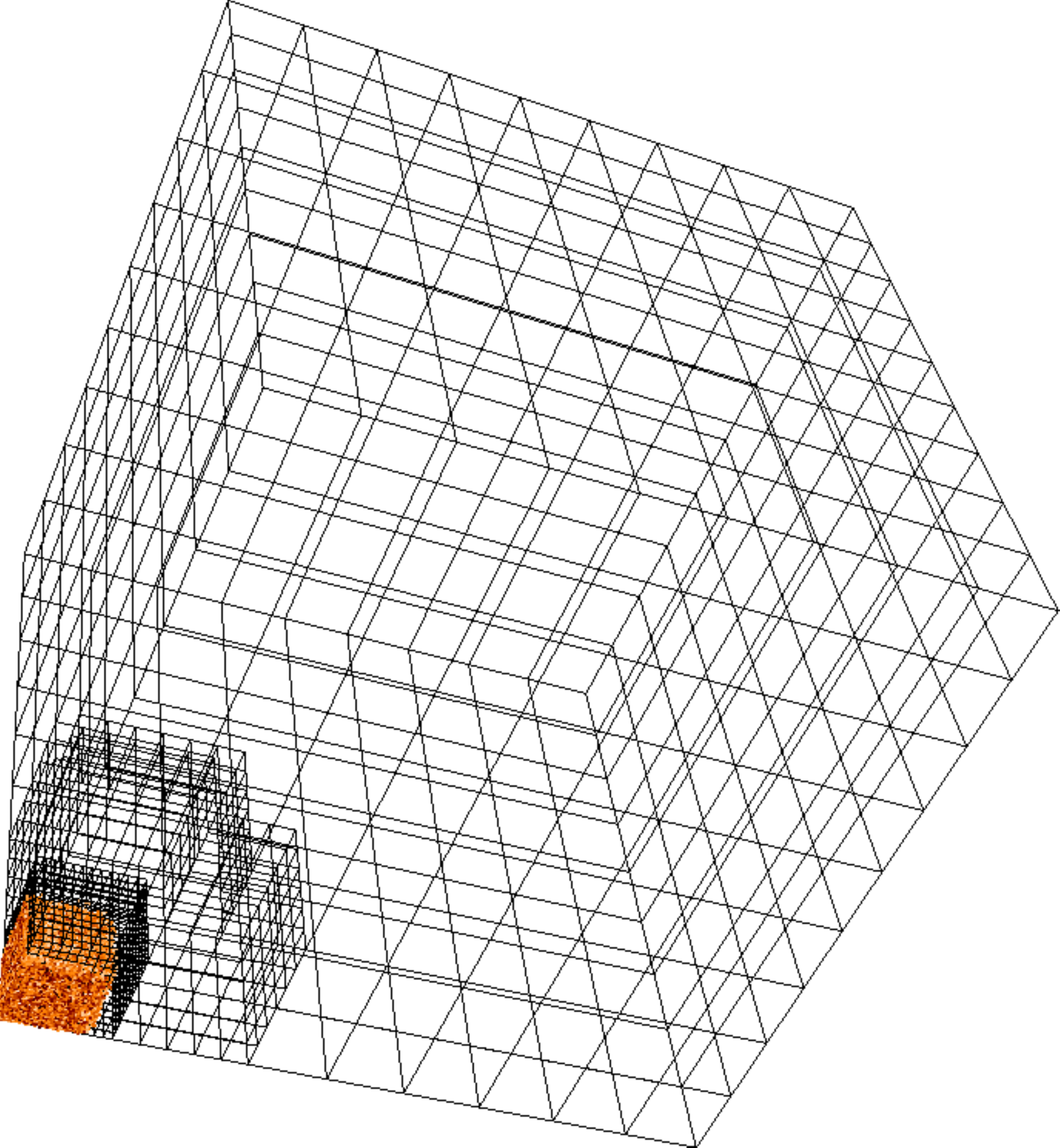}
 \includegraphics[width=0.32\linewidth]{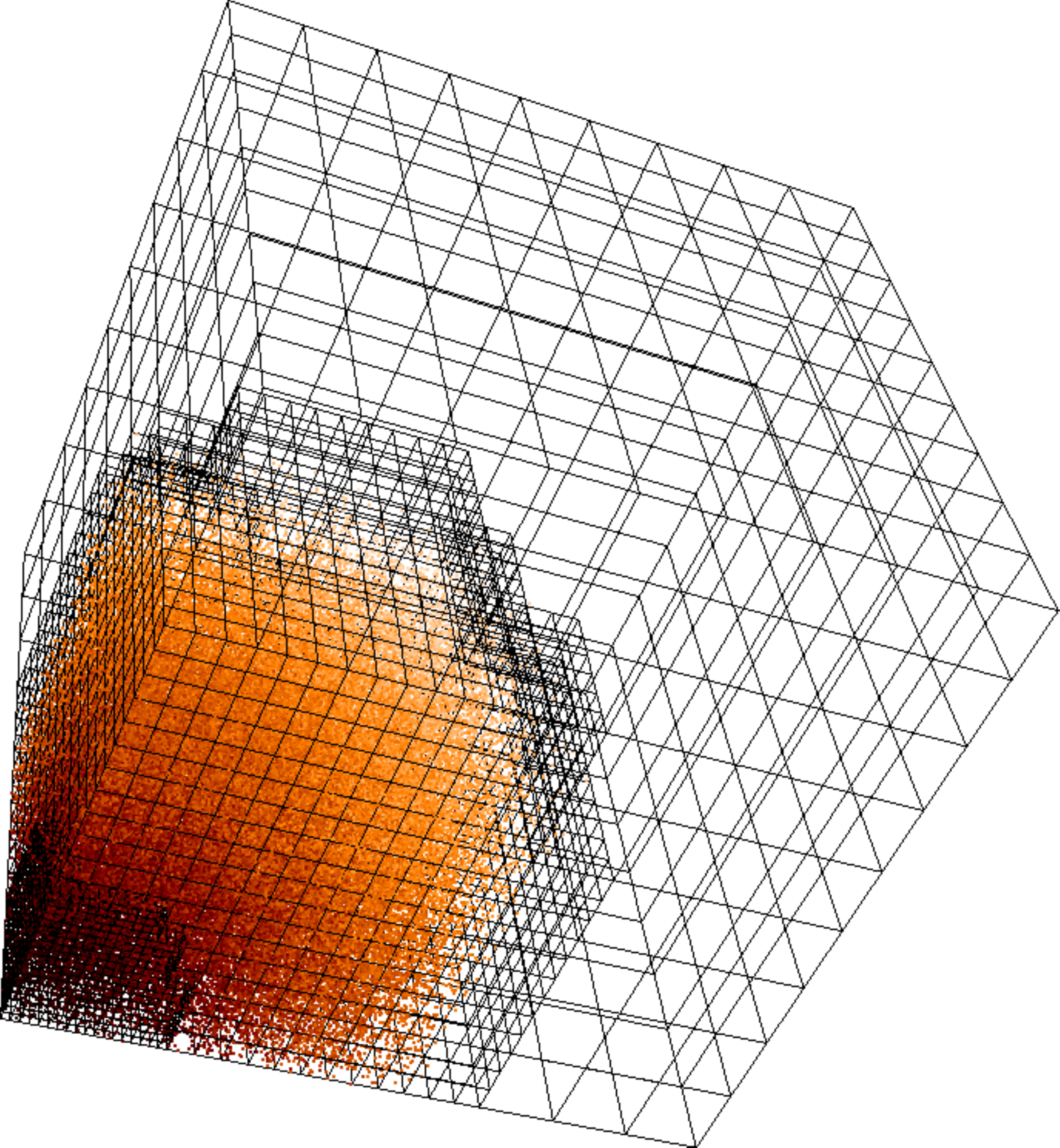}
 \includegraphics[width=0.32\linewidth]{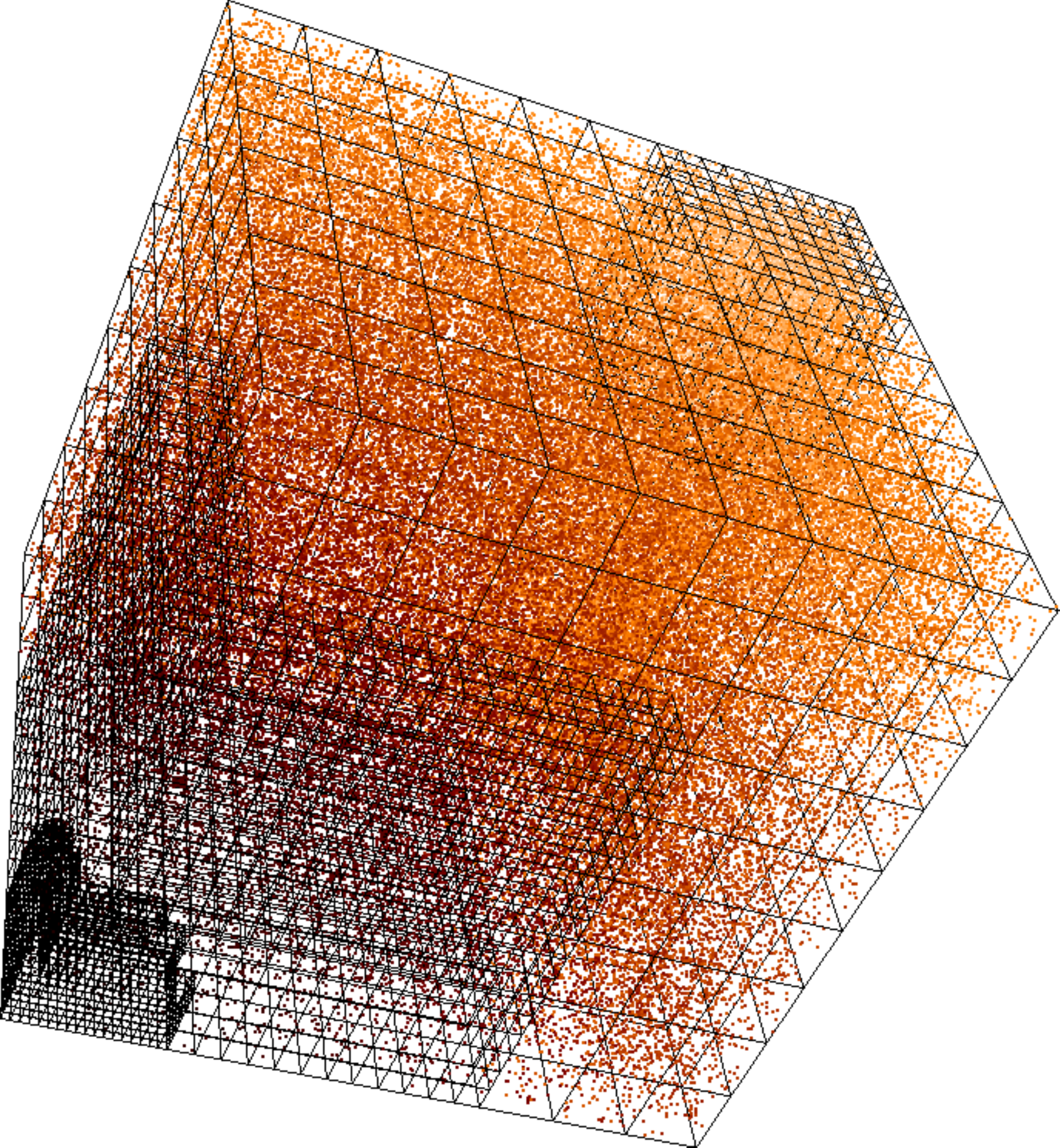}
 \caption{ An artificial test setup at time steps $t\in\{2,130,190\}$ from left
   to right:
   Particles are homogeneously distributed among a small cube embedded into the
   unit cube at time $t=0$.
   We apply reflecting boundary conditions. 
   As the particles move according to a fixed random velocity (the brighter the higher 
   their speed) and as the grid keeps the number of particles per cell bounded,  
   the mapping of grid entities to particles is permanently to be updated.
   }
 \label{fig:eyecatcher}
  \vspace{-0.2cm}
 \end{figure}

The particle-in-cell (\pic) method has originally been developed to
solve kinetic equations in plasma physics
and self-gravitating systems (the so-called Vlasov equation, or collisionless
Boltzmann equation), but it has also shown to be efficient for computational fluid simulations
(\cite{Bender:12:PackedMemoryArray,Fleissner:08:LB,Hockney:88:PIC,Ihmsen:11:Sorting,Lapenta:12:PICSpaceWeather,Lapenta:95:PIC,Li:02:MeshfreeAndParticle,Plimpton:03:LBForPIC,Verleye:13:PIC},
e.g.).
Macro particles in a Lagrangian frame---they mimic
the behaviour of a distribution function---are tracked in continuous phase
space.
Simultaneously, moments of the distribution function such as charge
density for plasma physics simulations are computed on an Eulerian frame 
(fixed cells) in $\mathbb{R}^d, \ d \in \{2,3\}$. 
The \pic\ method reduces the N-body problem by filtering out binary interactions
between particles through a so-called mean field approximation.
In return, it couples the particles to a grid accommodating a grid-based solver
of a partial differential equation.
\pic\ typically reads as follows:

\begin{enumerate}
  \item Given a set of particles at positions $x_i$, 
   the charge density $\rho$ is defined on the grid, and obtained from the
   particle through 
   $\rho = \sum _i R(x_i)$. 
    The restriction $R(x_i)$ in our case has local support and the grid based
    upon squares/cubes is chosen such that $R$ affects at most $2^d$
    vertices.
    It depends on the particle position.
  \item We solve  
    \begin{equation}
      \mathcal{L}(V) = \rho
      \label{equation:pic:pde} 
    \end{equation} 
    on the grid which yields a potential $V$.
    $V$'s semantics depend on the model. In our case it is the electric
    potential though all properties hold when we 
    consider a gravitational potential with $\rho$ being the \emph{mass}
    density.
    $\mathcal{L}$ is, in our case, the Laplacian. Throughout the
    solution process, the grid is dynamically coarsened and refined to resolve
    $V$ and the particle-cell interplay accurately.
  \item A field $E = -\nabla V$ is
    derived from the potential. This field then is interpolated from the
    grid to the particle positions $ \forall i: \qquad E _i = R^T (x_i) E $.
    We apply the transpose of the interpolation per particle. 
  \item The particle velocities $v_i$ and positions $x_i$ are updated
    (some authors call this ``push'' \cite{Plimpton:03:LBForPIC}) with 
    $\partial _t v_i = \phi _i(E_i,v_i)$ and $\partial _t x_i = v_i$
    with a generic particle property equation $\phi _i(E_i,v_i)$.
\end{enumerate}

\noindent
Moving the particles is computationally cheap once the impact $E_i$ on the
particles is known.
The particle-to-grid and grid-to-particle mappings usually stem from the 
multiplication of particle shape functions  
translated to the particle positions $x_i$ with grid-aligned test 
functions that in turn are used to discretise the PDE.
In our case, we restrict to Cartesian grids, i.e.~to square or cube cells for
the spatial discretisation, and we use $d$-linear shape functions for the PDE.
More sophisticated schemes are possible.

As $\mathcal{L}$ generally may comprise $\partial _t E$, $E$ or $V$ from the
previous time step have to be available on the grid.
Equation (\ref{equation:pic:pde}) typically is solved by an iterative scheme.
%
The present work neglects particularities of the PDE solver as well as 
dynamic load balancing and adaptivity criteria.
However, we highlight that the impact of the particles on the
PDE solution correlates to the particle density, i.e.~the more particles in a
given grid region, the rougher the PDE solution due to stimuli on the right-hand
side.
At the same time, a proper adaptivity criterion
anticipates the Debye length: the maximum grid size depends locally on both
the particle density and the mean velocity (bulk flow velocity) or mean square particle velocity (thermal velocity)
\cite{Birdsall:05:PlasmaPhysics}.
Our particle-grid realisations have to support dynamically adaptive grids.

We first focus exclusively on adaptive Cartesian grids and 
storage paradigms for the dynamically adaptive grid
holding the particles in the cells/vertices.
Different to many other particle codes where the grid is merely a helper data
structure, our grid stays persistent between any two time steps.
Second, we focus on reassignment procedures once the particles are moved.
Different to many other particle codes, it is crucial that the grid-particle
relations are updated immediately and that there is no restriction on $v_i$ with
respect to the grid size.
Third, we focus on the distributed memory parallelisation of the reassignment
induced by a non-overlapping, given grid decomposition, and we study how particles have to be exchanged if they move along
this decomposed geometric structure.
This facilitates a fast evaluation of the grid-particle operator, as particles
reside on the same node as their corresponding grid element.
A discussion of a proper choice of a grid decomposition as well as dynamic
load balancing are beyond scope.
Finally, we reiterate that we neither can make assumptions on the grid structure 
nor on the particle velocities nor on the actual particle distribution.
Particles can either traverse the grid elements smoothly or tunnel several 
grid elements a time.

\section{A distributed \spacetree\ data structure holding particles}
\label{section:spacetree}

A multitude of ways exists to formalise and implement adaptive Cartesian grids.
Tree-based approaches are popular
(cf.~overview
in \cite{Bader:13:SFCs} or
\cite{Lashuk:12:ParallelFMMOnHetergeneousArchitectures,Rahimian:2010:GordonBell,
Sampath:08:Dendro,Weinzierl:2009:Diss,Weinzierl:11:Peano}).
They facilitate dynamic adaptivity and low memory footprint storage
schemes teaming up with good memory access characteristics---in particular 
in combination with space-filling curves \cite{Bader:13:SFCs}.

In the present paper, we follow a $k$-\spacetree\ formalism 
\cite{Weinzierl:2009:Diss,Weinzierl:11:Peano}:
the computational domain is embedded into a square or cube, the {\em root}. 
This geometric primitive is split into $k$ parts along each coordinate axis. 
We end up with $k^d$ squares or cubes, respectively.
They tessellate the original primitive.
This setup can be represented by a graph with $k^d+1$ nodes and a relation
$\sqsubseteq _{child\ of}$.
Each node of this graph represents one cube or square, respectively, and is
denoted as {\it cell}.
If $a \sqsubseteq _{child\ of} b$, $a$ is contained within $b$ and is derived
from $b$ by $k$ cuts through $b$ along each coordinate axis.
We continue recursively while we decide for each cell whether
to refine further or not.
The resulting graph is a $k$-\spacetree\ given by $\sqsubseteq _{child\ of}$, a
set of cells $\mathcal{T}$ and a distinguished root.
For $k=2$ the scheme mirrors the traditional octree/quadtree concept
\cite{Bader:13:SFCs,Lashuk:12:ParallelFMMOnHetergeneousArchitectures,Rahimian:2010:GordonBell,
Sampath:08:Dendro}.
Our present code relies on the software \peano\
\cite{Weinzierl:15:Peano,Software:Peano} and thus uses $k=3$.
However, all algorithmic ideas work for any $k\geq 2$.

\begin{figure}  
\begin{center}
  \includegraphics[width=0.8\textwidth]{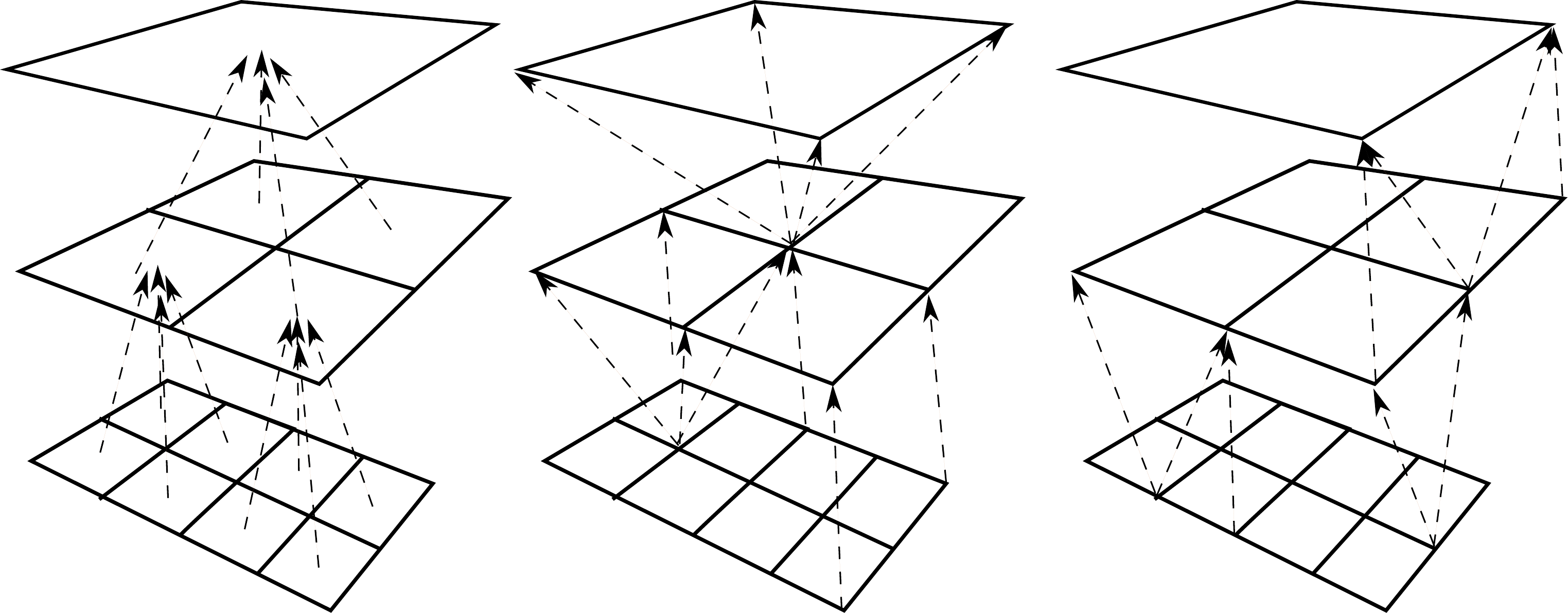}
  \vspace{-0.2cm}
  \caption{
    The spacetree introduces a parent relation on the spacetree cells
    (left; arrows point from child to parent).
    This spatial order in combination with the grid embedding induce a parent
    relation upon the vertices as well (middle and right).
    Vertices coinciding with coarser vertices on the same level have one parent,
    vertices within coarser cells have $2^d$ parents, vertices located on coarse
    grid faces and edges have a parent cardinality in-between.
  }
  \label{sketches:parent-relation}
  \vspace{-0.5cm}
\end{center}
\end{figure}

The distance from the root cell to any other cell is the cell's {\it level}.
All cells of one level represent geometric primitives of exactly the same size,
are aligned and do not overlap.
The $k$-\spacetree\ consequently yields a cascade of ragged
Cartesian grids.
The union of all these grids is an adaptive Cartesian grid.
Let $\Omega _{\ell }$ denote the grid of one \spacetree\ level $\ell$.  
The union $\Omega _h = \bigcup _\ell \Omega _{\ell }$ then yields an adaptive
Cartesian grid. 
A cell of $\mathcal{T}$ is a {\em leaf} if it does not
have a child.
If $a \sqsubseteq _{child\ of} b$, $b$ is a {\em parent} of $a$.
Due to the cascade-of-grids formalism, there may be multiple vertices in a
\spacetree\ at the same spatial location while they belong to different
grids
$\Omega _\ell$, i.e.~different levels.
Each vertex has up to $2^d$ adjacent cells on the same level. 
If it has less than $2^d$ adjacent cells, it is a {\em hanging vertex}.
Anticipating the \spacetree's partial order, a vertex has up to $2^d$ parent
vertices (Figure \ref{sketches:parent-relation}): 
A vertex $v_b$ is a parent to vertex $v_a$, if all cells that are adjacent to
$v_a$ have parent cells that are in turn adjacent to $v_b$.
While many tree-based codes deduce an adaptive grid from a \spacetree\ formalism
and then work basically on the spacetree's leaves, 
we preserve and maintain the whole tree as computational data structure though
the particle-grid interaction is often computed only in leaves.

\subsection{Dual \spacetree\ grid}
Besides the cascade of Cartesian grids, the \spacetree\ formalism also induces
dual grids. 
A dual grid $\Omega ^{(d)} _\ell $ of one level is defined as
follows:
It is a grid consisting of geometric primitives of exactly the
type as in $\Omega _\ell $. 
They are dilated such that each cell center of $\Omega _\ell $ coincides
with one non-hanging vertex of the dual grid.
$\mathcal{T}^{(d)}$ then is the cascade of dual grids to $\mathcal{T}$
(Figure~\ref{tree-grid-dual-grid}).
For odd $k$---we have $k=3$---we observe {\em dual grid consistency}:
A dual cell of one level either is contained completely within a dual cell of a
coarser level or does not intersect with coarsers cells at all.
The present algorithms did, in principle, not rely on this consistency, but
their implementation's simplicity benefits from $k=3$.

\begin{figure}  
\begin{center}
  \includegraphics[width=0.8\textwidth]{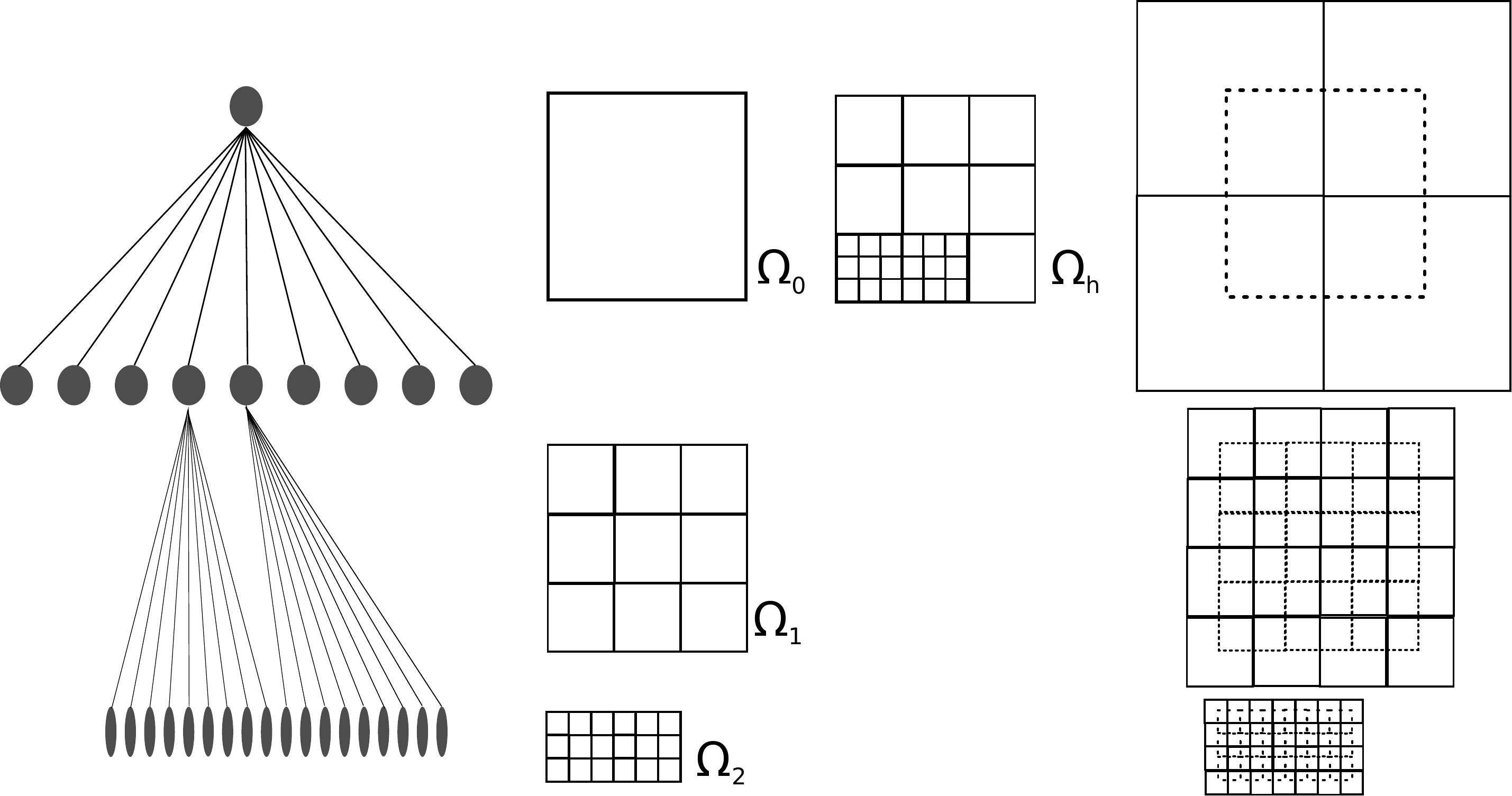}
  \vspace{-0.2cm}
  \caption{
    From left to right:
    A \spacetree\ is a directed graph $\sqsubseteq _{child\ of}$
    visualised bottom-up here.
    It yields a cascade of ragged Cartesian grids $\Omega _0, \Omega _1,
    \Omega _2$.
    The merger of the three is an adaptive Cartesian grid $\Omega _h$.
    The dual $k=3$-\spacetree\ yields a cascade of three dual grids
    (with the original grid as dotted lines).
  }
  \label{tree-grid-dual-grid}
  \vspace{-0.5cm}
\end{center}
\end{figure}

\subsection{Storing particles within the spacetree}

The present paper studies two choices to store the particle-grid relations:
either each cell holds the particles covered by it, or each vertex holds the 
particles whose positions are closer to this vertex than to any other vertex.
``Holds'' denotes that the grid entity basically links to an array of particles.
While we focus on the data handling and parallelisation, packed memory
algorithms \cite{Bender:12:PackedMemoryArray}, e.g., can replace these arrays
with more efficient realisation variants avoiding frequent reallocation.
Furthermore, we do not discuss global data layout optimisations to improve
the placement of the arrays in memory
\cite{Bender:12:PackedMemoryArray,Ihmsen:11:Sorting} but rather refer to \cite{Weinzierl:2009:Diss,Weinzierl:11:Peano} where we introduce a
multiscale ordering of the spacetree cells along the Peano space-filling curve.
Such an ordering of the cells transfers to an ordering of arrays.
We thus can assume reasonably efficient memory data access with respect to
caches.
Without loss of generality, these arrays are, for the time being, empty on the
coarse levels, whereas all particles reside within the finest grid resolution.
In our second storage scheme, particles are assigned to the vertex whose dual
cell covers their position.
It is a Voronoy-based particle assignment \cite{Springel:09:Arepo}.
As we are working in a spacetree environment, the two approaches are named
{\em particle in tree} (\pit) and {\em particle in dual tree} (\pidt).

\subsection{Parallel grid decomposition}

Let $\mathbb{P}=\{p_0,p_1,\ldots \}$ be the finite set of processes (ranks) on
a parallel computer.
$col: \mathcal{T} \mapsto \mathbb{P}$ is a {\em colouring} that assigns each 
\spacetree\ cell a colour, i.e.~a processor that is responsible for this cell. 
In practice, each processor holds only that part of a spacetree it is
responsible for \cite{Weinzierl:2009:Diss}.
Let $\sqsubseteq _{worker\ of}$ induce a tree topology on
$\mathbb{P}$ as follows:
\begin{equation}
  \forall a \sqsubseteq _{child\ of} b: \
    col(a) = p_i \wedge col(b) = p_j \Rightarrow
    p_i = p_j \vee 
    p_i \sqsubseteq _{worker\ of} p_j. 
    \label{subtree-relation}
\end{equation}
As the graph in (\ref{subtree-relation}) shall be free of cycles,   
different subtrees of the global \spacetree\ are
handled by different processes, i.e.~on different ranks.
We decompose the spacetree.
No refined cell is shared among multiple ranks, and 
$col$ introduces a master-worker relations with a distinguished global
master being responsible for the root.
This scheme differs from local essential tree constructions where coarse grid
cells are replicated among multiple nodes
(\cite{Lashuk:12:ParallelFMMOnHetergeneousArchitectures} and their references
to original work).
We rely on a unique rank responsibility for coarse grid spacetree cells
\cite{Weinzierl:15:Peano}.

Our tree colouring induces a multiscale non-overlapping \spacetree\
decomposition, i.e.~each cell is assigned to a rank uniquely, while vertices at
the partition boundaries are held and replicated on up to $2^d$ ranks.
The term non-overlapping refers to each individual level.
Let each cell of $\mathcal{T}^{(d)}$ be assigned to the
ranks that hold cells that are adjacent to the center of the cell in
$\mathcal{T}$.
The colouring then induces a multiscale 1-overlapping decomposition on
$\mathcal{T}^{(d)}$.

For \pic, it is convenient to decompose particles along $col$ as the particles
do not interact.
For \pidt, dual cells intersecting with the parallel boundary then are
replicated while the particles are never replicated but always are assigned to
one rank uniquely.


\subsection{Event-based formalism of the tree traversal}

A {\em tree traversal} is an algorithm running through $\mathcal{T}$.
An effective tree traversal shall have three properties:
\begin{enumerate}
  \item All data access is local within the grid/tree.
  \item All data access anticipates the domain decomposition.
  \item The data access scheme is efficient. In particular,
  any particle sorting shall be possible in one tree traversal even if the grid
  changes.
\end{enumerate}
For this, we process each cell of this set twice: an operation
\EventEnterCell\ is performed prior to an operation \EventLeaveCell.
We enforce 
\begin{eqnarray}
  \forall a \sqsubseteq _{child\ of} b: & \Rightarrow & 
    \mbox{\EventEnterCell}(b) \sqsubseteq _{before} \mbox{\EventEnterCell}(a) \
    \wedge \nonumber \\
    && \mbox{\EventLeaveCell}(a) \sqsubseteq _{before}
    \mbox{\EventLeaveCell}(b) \qquad \mbox{and} \nonumber \\
  \forall c: && \mbox{\EventEnterCell}(c) \sqsubseteq _{before}
  \mbox{\EventLeaveCell}(c)
  \qquad \mbox{for } a,b,c \in \mathcal{T}
    \label{traversal-constraint} 
\end{eqnarray}
where $\sqsubseteq _{before}$ is a partial temporal access order.
Obviously, both depth-first and breadth-first $k$-\spacetree\ traversal as well
as hybrid variants preserve (\ref{traversal-constraint}).

As the tree is a representation of the cascade of grids, a tree traversal
describes a strict element-wise multiscale processing of the whole grid cascade.
For any element-wise realisation of a PDE solver or a particle-based algorithm
it is then sufficient to specify which data are assigned to the $k$-\spacetree's 
vertices and cells, and to specify how individual {\em events}
\cite{Weinzierl:15:Peano} such as \EventEnterCell\ and \EventLeaveCell\ map onto algorithmic fragments often called compute kernels.
Our discussion restricts to element-wise algorithms as those algorithms fit
straightforwardly to non-overlapping domain decompositions.

Following the notion of element-wise processing, only records assigned to one 
cell and its adjacent vertices are available to an event.
Following the notion of a spacetree, the parent 
data are passed to events as well.
Besides the cell events, our implementations rely on two additional events:
\EventTouchVertexFirstTime\ is called once per \spacetree\ vertex per
traversal per rank before the vertex is used by this rank the very first time,
i.e.~before \EventEnterCell\ is invoked on any adjacent cell of this vertex.
\EventTouchVertexLastTime\ is called once per \spacetree\ vertex on each rank 
after the vertex has been used the very last time, i.e.~after
\EventLeaveCell\ has been invoked on all adjacent cells. 

In a parallel tree traversal, the global master starts to traverse the
tree. 
Its workers' tree traversals are successively started up as soon as 
(\ref{traversal-constraint}) allows for. 
This is a broadcast along a tree topology.
Whenever a tree traversal ascends again in the tree (processes \EventLeaveCell\
and \EventTouchVertexLastTime), it might have to wait for other colours to
finish their
share of the global tree due to (\ref{traversal-constraint}).
The startup of a remote colour is accompanied by the two events
\EventPrepareSendToWorker\ and \EventMergeWithWorker\ invoked on the master or
worker, respectively.
The other way round, \EventPrepareSendToMaster\ and
\EventMergeWithMaster\ are called.
Furthermore, events, i.e.~plug-in points for the algorithm, do exist to merge
vertices at the domain decomposition boundaries.
Vertex exchanges are triggered after a local vertex has been used for the
last time.
\EventMergeWithNeighbour\ events then are
invoked on each rank per parallel domain boundary vertex prior to any usage of this vertex.
As \EventMergeWithNeighbour\ is called prior to the next usage but copies are
sent after \EventTouchVertexLastTime, vertex data exchange along the rank
boundaries overlaps two iterations.
This exchange is asynchronous while the
information exchange between masters and workers is synchronous.

\section{Cell- and vertex-based particle movers}
\label{section:cell_and_vertex_based}

For each of the two particle storage schemes, \pic\ requires the scheme
to maintain the particle-to-grid mapping.
Whenever a particle moves, our code
has to analyse whether the particle remains within its cell or dual cell,
respectively, or it has to update the mapping otherwise.
While it might be straightforward to iterate first over all particles to move
them before we sort them in an update sweep, we propose to merge particle movement
and reassignment.
This way, we avoid an extra sorting step.
For both \pit\ and \pidt, we propose to move particles up and down
within the \spacetree\ to enable tunneling, but we enforce that all particles
are sorted into the leaf tree level prior to any subsequent operation on the
particle in the next traversal.
Our algorithmic sketches describe stationary grids. 
If \spacetree\ nodes are added dynamically, both approaches automatically move 
particles into the right grid entities as the sort algorithm makes all particles
reside on the finest grid level.
If \spacetree\ nodes are removed, both approaches move the particles associated 
with removed grid entities up in the spacetree and continue. 
Support of dynamic adaptivity hence is straightforward and not discussed
further.

\subsection{Particle in tree (\pit)}

\begin{figure}[tb]
 \begin{center}
  \includegraphics[width=0.8\textwidth]{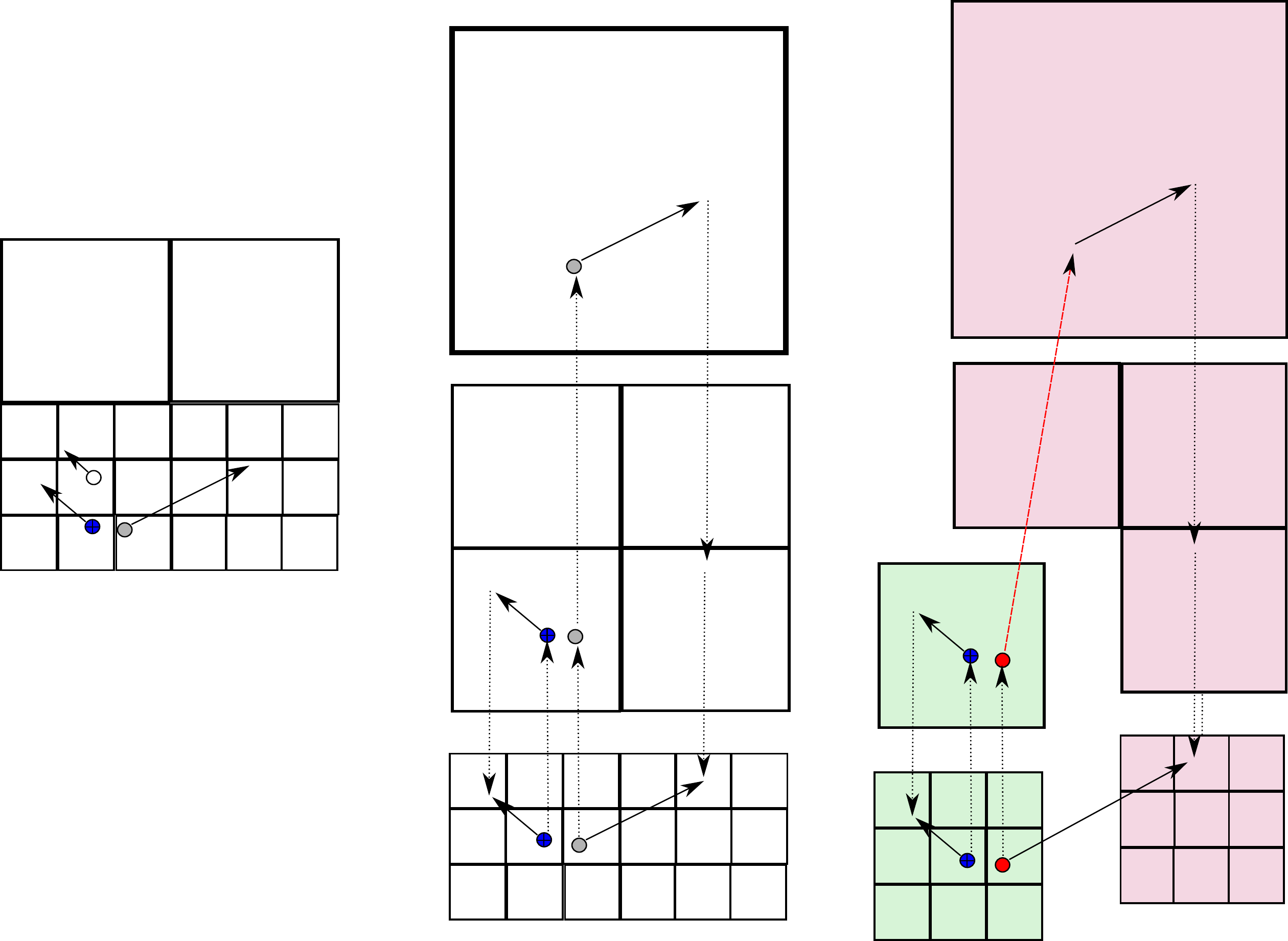}
  \caption{
    Left: Assigning particles to cells is trivial as long as particles due not
    leave `their' cell to the particle update (empty particle).
    If particles leave a cell (filled particles),
    they have to be reassigned to a new cell.
    Middle: \pit\ lifts these particles to the next coarser levels and then
    drops them back. 
    If particles tunnel (light gray), they have to be lifted several levels.
    Right: 
    A tree decomposed into two colours is traversed by \pit. 
    Lifts (red) and
    drops (blue) are embedded into this parallel traversal.
    Therefore, particles are exchanged in-between processors only along the
    parallelisation's master-worker topology. 
    }
  \label{figure:pit}
  \vspace{-0.5cm}
  \end{center}
\end{figure}

\pit\ maps each particle from the particle set $\mathbb{M}$ onto a
leaf of the \spacetree, i.e.~$\mathcal{X}_{\mbox{\tiny \pit}}: \mathbb{M} \mapsto \mathcal{T}$.
Only the particle position $x \in \mathbb{R}^d$ determines this mapping.
\pit's particle update then reads as Algorithm \ref{algo:pit} and
integrates into any tree traversal preserving (\ref{traversal-constraint})
(Figure \ref{figure:pit}).

\begin{algorithm}
  \caption{Particle-in-tree algorithm.}
    \label{algo:pit}
    {\footnotesize
    \begin{algorithmic}[1]
       \Function{\EventEnterCell}{cell $c \in \mathcal{T}$}
         \If{$c$ refined}
           \ForAll{$m \in \mathbb{M}$ associated to $c$}
           \Comment preamble
             \State identify $c' \sqsubseteq _{child\ of} c$ containing $m$
             \State assign $m$ to $c'$ \Comment drop
           \EndFor
         \EndIf
         \State invoke application-specific operations
         \If{$c$ unrefined}
           \ForAll{$m \in \mathbb{M}$ associated to $c$}
             \State update position $x(m)$ 
             \Comment move
           \EndFor
         \EndIf
         \State invoke application-specific operations
       \EndFunction

       \Function{\EventLeaveCell}{cell $c \in \mathcal{T}$}
         \State invoke application-specific operations
         \ForAll{$m \in \mathbb{M}$ associated to $c$}
           \Comment epilogue
           \If{position $x$ of $m$ not contained in $c$}
             \State assign $m$ to $c'$ with $c \sqsubseteq _{child\ of} c'$
             \Comment lift
           \EndIf
         \EndFor
       \EndFunction
   \end{algorithmic}
   }
\end{algorithm} 


Each global particle sorting is split among two traversals. 
Let traversal $t_2$ follow $t_1$. 
All lift operations with respect to particle updates in $t_1$ are embedded into
the traversal $t_1$, whereas the drops and the particle position updates are embedded into $t_2$.
$t_2$ already realises the lifts of the subsequent time step whereas $t_1$ 
realises the drops of the previous one.
Hence, one (amortised) traversal per resort is sufficient---the first 1.5
traversals realise the first sort, the next 2.5 the second, \ldots---and
the following statement holds:

\begin{theorem}
\label{theorem:one}
Whenever \EventEnterCell\ is invoked on a leaf $a$ and all its preamble
operations terminate, all particles within the cell have a position
$x$ covered by $a$.
\end{theorem}

\begin{proof}
  We restrict to a
  single particle $m \in \mathbb{M}$, and we assume that the theorem holds prior
  to traversal $t_1$ with $\mathcal{X}_{\mbox{\tiny \pit}}(m)$ being a leaf.
  We first show that $x \in \mathcal{X}_{\mbox{\tiny \pit}}(m)$ as soon as $t_1$
  terminates though the image can be a refined spacetree cell.
  $x \in \mathcal{X}_{\mbox{\tiny \pit}}(m)$ denotes that the spacetree cell to which $m$
  is mapped to covers the position $x$.
  For this, we make a simple top down induction on the tree depth.
  If $\mathcal{X}_{\mbox{\tiny \pit}}(m)$ is the root, $x \in
  \mathcal{X}_{\mbox{\tiny \pit}}(m)$ is trivial as the root spans the whole computational domain. 
  Otherwise we distinguish two cases for \EventLeaveCell\ invoked for a
  \spacetree\ cell of level $\ell +1$:
  \begin{itemize}
    \item $x \in \mathcal{X}_{\mbox{\tiny \pit}}(m)$: the particle resided in
    $\mathcal{X}_{\mbox{\tiny \pit}}(m)$ prior to a position update and doesn't leave the cell.
    \pit\ does not modify $\mathcal{X}_{\mbox{\tiny \pit}}$.
    \item $x \not\in \mathcal{X}_{\mbox{\tiny \pit}}(m)$: the particle leaves the cell in
    which it was contained before. In \EventLeaveCell, we assign it to its
    parent, i.e.
    \begin{eqnarray*}
      \mathcal{X}_{\mbox{\tiny \pit}} & \mapsto &
      \mathcal{X}^{\mbox{\tiny (new)}}_{\mbox{\tiny \pit}} \qquad \mbox{with} \\
      \mathcal{X}^{\mbox{\tiny (new)}}_{\mbox{\tiny \pit}}(m) & = & p  \qquad \mbox{while} \\
      \mathcal{X}_{\mbox{\tiny \pit}}(m) & \sqsubseteq _{child\ of} & p.
    \end{eqnarray*} 
    Due to (\ref{traversal-constraint}),
    this happens after \EventEnterCell\ is invoked for $m$ in $t_1$ and
    does not harm our initial assumption. Due to (\ref{traversal-constraint}),
    this transition is triggered before the \EventLeaveCell\ call for the parent
    on level $\ell$. The statement for level $\ell$ holds by induction.
  \end{itemize}
  For the subsequent traversal $t_2$, we again use an induction over the tree
  height and distinct two cases. However, we argue bottom-up. For trees of
  height zero the algorithm's correctness is trivial. Let the root be refined.
  \begin{itemize}
    \item $\mathcal{X}_{\mbox{\tiny \pit}}(m)$ on level $\ell$ is a leaf: The theorem holds. 
    \item $\mathcal{X}_{\mbox{\tiny \pit}}(m)$ on level $\ell$ is refined. The algorithm
    bucket sorts the particle into the child cell covering its new position, 
    i.e.~
    \begin{eqnarray*}
      \mathcal{X}_{\mbox{\tiny \pit}} & \mapsto &
      \mathcal{X}^{\mbox{\tiny (new)}}_{\mbox{\tiny \pit}} \qquad \mbox{with} \\
      \mathcal{X}_{\mbox{\tiny \pit}}^{\mbox{\tiny (new)}}(m) & = & c  \qquad \mbox{while} \\
      c & \sqsubseteq _{child\ of} & \mathcal{X}_{\mbox{\tiny \pit}}(m).
    \end{eqnarray*} 
    $c$ has level $\ell +1$. Due to (\ref{traversal-constraint}), \EventEnterCell\ for the parent is
    invoked prior to \EventEnterCell\ for the \spacetree\ node on level $\ell
    +1$.
    The statement for level $\ell+1$ holds by induction.
  \end{itemize}
\end{proof}

\noindent
A distributed memory parallel version of \pit\ adds two case distinctions. 
Whenever \pit\ lifts a particle from a local root,
i.e.~a \spacetree\ node whose parent is assigned to a different colour, this
particle is sent to the parent's rank by the event
\EventPrepareSendToMaster.
In return, \EventMergeWithMaster\ receives and inserts it into the local data
structure.
Whenever the tree traversal encounters a cell of a different colour, 
\EventPrepareSendToWorker\ replaces the drop by a send to the worker rank.
\EventMergeWithWorker\ receives them and continues to drop them (Figure
\ref{figure:pit}).
All parallel data flow is aligned with the tree topology on $\mathbb{P}$,
i.e.~particles are only sent up and down within the \spacetree.
All data exchange is synchronous. 

\subsection{Particle in dual tree (\pidt)}
\label{particle-in-dual-tree}

\begin{figure}  
\begin{center}
 \includegraphics[width=0.5\textwidth]{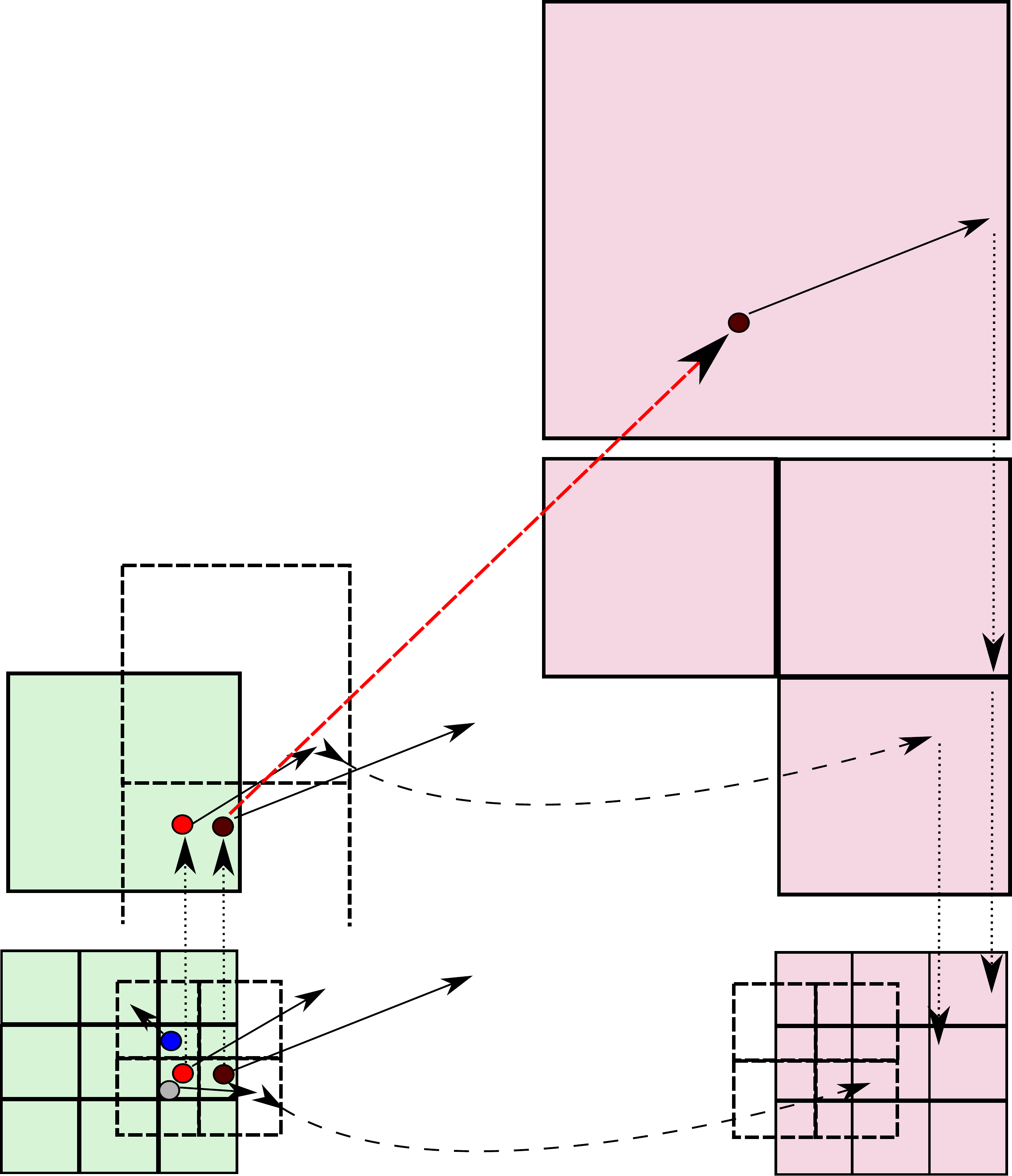}
  \caption{
    Parallel \pidt: 
    There are three different types of particle movements.
    Particles might remain in their dual cell (blue), might move to an
    adjacent dual cell (gray) or tunnel (red and dark red).
    Only for the latter, multilevel movement comes into play.
    Copies of particles leaving through vertices without further tunneling 
    are sent away immediately throughout the traversal but are not received 
    prior to the next iteration (asynchronous communication, dotted arrows
    from left to right).
    Only very few tunneling particles are communicated synchronously among the
    tree topology (dark red particle, red dotted diagonal arrow).
  }
  \label{figure:pidt}
 \vspace{-0.5cm}
\end{center}
\end{figure}

\pidt\ augments each particle by a boolean flag $moved \in \{\top,\bot\}$
(``has moved already'' $\top$ and ``has not not moved yet'' $\bot$) and
maps each particle onto a vertex of the \spacetree,
i.e.~$\mathcal{X}_{\mbox{\tiny \pidt}}: \mathbb{M} \mapsto \mathcal{T}^{(d)}$.
The particle update then reads as as Algorithm \ref{algo:pidt} and
integrates into any tree traversal preserving (\ref{traversal-constraint}) (Figure \ref{figure:pidt}).
For efficiency reasons, the particle loops in \EventEnterCell\ and
\EventLeaveCell\ can be merged for leaf cells.

\begin{algorithm}
  \caption{Particle-in-dual-tree algorithm (continued in 
  Algorithm \ref{algo:pidt2}).}
    \label{algo:pidt}
    {\footnotesize
    \begin{algorithmic}[1]
       \Function{\EventTouchVertexFirstTime}{$v$}
         \If{$v$ refined} \Comment i.e.~all surrounding cells are refined
           \ForAll{$m \in \mathbb{M}$ associated to $v$}
           \Comment preamble
             \State $moved(m) \gets \bot$
           \EndFor
         \EndIf
         \State invoke application-specific operations
       \EndFunction

       \Function{enterCell}{cell $c \in \mathcal{T}$}
         \If{$c$ refined}
            \Comment preamble
            \ForAll{$2^d$ adjacent vertices $v$}
              \ForAll{$m \in \mathbb{M}$ associated to $v$}
               \If{$moved(m)=\bot$ and $x(m)$ contained in $c$}
                 \State identify child vertex $v'$ whose dual cell holds $x(m)$
                 \State assign $m$ to $v'$
                 \Comment drop
                \EndIf
              \EndFor
            \EndFor
         \EndIf
         \State invoke application-specific operations
         \If{$c$ unrefined}
           \ForAll{$2^d$ adjacent vertices $v$}
             \ForAll{$m \in \mathbb{M}$ associated to $v$}
               \If{$moved(m) = \bot \wedge x(m)$ contained in $c$}
                 \State update position $x$ of $m$  \Comment move
                 \State $moved(m) \gets \top$
               \EndIf
             \EndFor
           \EndFor
         \EndIf
         \State invoke application-specific operations
       \EndFunction
   \end{algorithmic}
   }
\end{algorithm}

\begin{algorithm}
  \caption{Particle-in-dual-tree algorithm continued from Algorithm
  \ref{algo:pidt}.}
    \label{algo:pidt2}
    {\footnotesize
    \begin{algorithmic}[1]
       \Function{leaveCell}{cell $c \in \mathcal{T}$}
         \State invoke application-specific operations
         \ForAll{$2^d$ adjacent vertices $v$}
           \Comment epilogue
           \ForAll{$m \in \mathbb{M}$ associated to $v$}
             \If{$moved(m) = \top \wedge x(m)$ contained in dual cell of
             other $v'$ adjacent to $c$}
               \State assign $m$ to $v'$
               \Comment linked-list type reassignment 
               \State
               \Comment (no tunneling)
             \EndIf 
           \EndFor
         \EndFor
       \EndFunction

       \Function{\EventTouchVertexLastTime}{$v$}
         \State invoke application-specific operations
         \ForAll{$m \in \mathbb{M}$ associated to $v$}
           \Comment epilogue
           \If{$moved(m)=\top \wedge x(m)$ not contained within dual cell of
           $v$}
             \State assign $m$ to parent vertex 
             \Comment lift (cmp.~Figure \ref{sketches:parent-relation})
             \State
             \Comment tunneling
           \EndIf 
         \EndFor
       \EndFunction
   \end{algorithmic}
   }
\end{algorithm}

\begin{theorem}
Whenever \EventTouchVertexFirstTime\ is invoked on a vertex $v$ and all its
$preamble$ operations have terminated, all particles within the dual cell of
$v$ have a position covered by this dual cell.
\end{theorem}

\begin{proof}
The proof follows the proof of Theorem
\ref{theorem:one}.
\end{proof}

\noindent
Again, parallel \pidt\ is straightforward.
Whenever the grid traversal has invoked \EventTouchVertexLastTime\ on one
parallel rank for a vertex, each particle of this vertex falls into one of four
categories:
\begin{itemize}
  \item It is to be lifted but the vertex does not belong to the coarsest level
  on the local rank. 
  These particles are neglected by the parallelisation as the particle is lifted 
  locally with \EventTouchVertexLastTime\ being called bottom-up throughout the
  traversal.
  \item It is to be lifted and the vertex is adjacent to the coarsest cell held by a 
  rank. 
  Such a particle is sent to the rank's master node where it is received in the 
  same iteration prior to the master's \EventLeaveCell.
  \item It belongs to the vertex's dual cell and intersects the local domain. 
  Such a particle already is assigned to the correct vertex and skipped.
  \item It belongs to the vertex's dual cell but it is not covered by the 
  local domain. 
  In the latter case, the particle is removed from the local vertex and sent to
  the respective destination rank merging all received particles into the local
  vertices prior to \EventTouchVertexFirstTime\ in the next tree traversal.
\end{itemize}

The data flow from workers to masters is aligned with the tree topology on $\mathbb{P}$.
This data flow comprises only particles that have to be lifted between the
coarsest worker cell and its parent residing on the master. 
It is synchronised with the tree traversal.
The drop mechanism also uses synchronous particle exchange whenever a rank 
descends into a cell not handled locally.
The exchange through the vertices follows a Jacobi-style update:
particles are sent away in one traversal and received prior to \EventTouchVertexFirstTime\ 
of the subsequent traversal.
That is convenient, as the particle lift and drop are split among two tree
traversals.
The latter data exchange is asynchronous and can be realised in the background
in parallel to the tree traversal.

\subsection{Remarks}

While \pit\ and \pidt\ realise the same fundamental ideas, they differ in code
complexity and the handling of not-tunneling particles.
If particles do not tunnel but move from one cell to a neighbouring
cell, \pit\ moves them up in the tree one level and down one
level.
\pidt\ exchanges them directly mirroring linked-list techniques.
This pattern applies in a multiscale sense for \pidt.
\pidt\ thus is a multiscale linked-list approach, where its 
``radius of operation'' is twice the grid width.
It is expected to have fewer lifts and drops compared to \pit.
Contrary, \pidt's higher algorithmic complexity mirrors directly to a more complex code.
While this might not influence the runtime, the redundant check of particles
does:
we need a flag $moved$ to mark updated particles which might double the number
of particle loads, i.e.~memory accesses.

Finally, we antedate experimental insight. 
Both algorithms suffer from a synchronisation of 
masters and workers.
Workers cannot start their tree traversal prior to their master having entered
the parent cell, as particles might drop from coarser resolution levels.
Masters cannot continue their tree traversal before the workers have finished,
as particles might have to be lifted from finer resolution levels.
Both synchronisation points introduce a latency and a bandwidth penalty
and make perfectly balanced load a must.
The bandwidth penalty scales with the number of particles to be exchanged.
It is hence less significant for \pidt\ than for \pit.
Both master-worker and worker-master communication describe a global all-to-all
communication that is realised as tree communication.
We state that the resulting global synchronisation were needless if no particles
would be lifted or dropped.

\subsection{Restriction-avoiding \pidt\ (\rapidt)}

To identify cases where we could skip the global communication, 
we introduce a marker 
$v_{max}: \mathcal{T} \mapsto \mathbb{R}^+_0$.
On any leaf of the spacetree, $v_{max}$ shall hold the maximum velocity 
component of all particles associated to the $2^d$ adjacent vertices. 
$v_{max}$ is a cell-based value though it results from vertex-associated 
data.
It is the maximum norm of all velocity components of all particles 
contained within the $2^d$ dual cells intersecting the current cell.
This value is determined by \EventLeaveCell\ in each traversal.
On a refined cell, we update the value with
\begin{equation}
  v_{max}(c) \gets max \{ v_{max}(c') : c' \sqsubseteq _{child\ of} c \}.
  \label{equation:vmax}
\end{equation}
Here, $v_{max}(c)$ is an analysed tree attribute
\cite{Knuth:90:AttributeGrammar} (re-)computable on-the-fly.

Given any \spacetree\ cell of width $h=3^{-level}$, the time step size
$\Delta t$ and a correct $v_{max}$, we know that no particles will be removed
from this cell or any of its children and successors due to lifts if all
particles contained have a velocity smaller than $h / \Delta t$.
We can check this due to $v_{max}$.
We augment \pidt\ with the following mechanisms:
\begin{itemize}
  \item Each rank holds a boolean map for all its workers.
  \item On \EventEnterCell\ for a cell deployed to a worker, we first update the
  local $v_{max}$ due to all particles dropped into this cell.
  \item The updated $v_{max} \leq h / \Delta h$ identifies a priori 
  whether particles will be lifted again from the remote subtree.
  We store this information within the local boolean map.
  \item We then continue with \pidt, i.e.~start up the remote rank.
  \item On \EventLeaveCell\ on the respective deployed cell, we have two
  opportunities:
  \begin{enumerate}
    \item If lifts are to be expected, we receive all particles from the
    worker, redetermine (\ref{equation:vmax}) and continue.
    \item If no lifts are to be expected, we continue with the tree
    traversal without waiting for any worker data. The reduction is avoided.
  \end{enumerate}
\end{itemize}
The receive branch is mirrored with the same rules on the worker to ensure that
no worker sends data in the reduction-avoiding case at all.

\section{Results}
\label{section:performance}

Our case studies were executed on SuperMUC at the Leibniz
Supercomputing Centre and the N8 Polaris system at Leeds.
SuperMUC hosts Sandy Bridge E5-2680 processors clocked to 2.3 GHz, Polaris
hosts Sandy Bridge E5-2670 processors at 2.6 GHz plus Turbo Boost.
Both have two eight-core processors per node.
SuperMUC runs Intel MPI 4.1, and every 8192 cores (512 two-processor nodes) are
called an island and are connected via a fully non-blocking
Infiniband network.
Beyond that core count, SuperMUC relies on a 4:1 blocking network. 
Polaris runs Open MPI 1.6.1, and every 192 cores are
fully non-blocking connected to one Mellanox QDR InfiniBand switch. 
Beyond that core count, Polaris realises 2:1 blocking.
Our code relies on good, i.e.~close to optimal, placement, i.e.~it uses as few
switches as possible.
All performance data is specified as particle updates per second, all
machine sizes are specified in cores.

For the performance tests, we remove the PDE solver part as
well as the interplay of particles and
grid, and thus focus exclusively on the particle handling.
The appendix presents a real-world run.
We study worst-case setups where the particle handling's performance
characteristics are not interwoven with other application phases.
Yet, we artificially impose, where highlighted, a fixed number of 0, 128, 256,
1024 or 4096 floating point operations (flops) onto each particle move to
quantify the impact of arithmetics on the runtime.
All particles are initially placed randomly and homogeneously within  
the unit square $(0,1)^d \subset \mathbf{R}^d$  
or a subdomain $(0.1,0.1)^d \subset (0,1)^d$.
The total number of particles results from the particle density $\rho
_{\mbox{\footnotesize particles}}$ within the initially populated area. 
Each
particle is assigned a random velocity $0 \leq |v| \leq 1$ uniformly distributed
(Figure \ref{fig:eyecatcher}).
This way, we make the
characteristic particle movement per step depend only on one quantity---the time
step size.
A confusion with particle-grid mapping in the PIC blueprint of Section
\ref{section:pic}~is out of question as we neglect the PDE. We hence write
$\rho = \rho _{\mbox{\footnotesize particles}}$.
Time step size $\Delta t$ and the maximum number of particles per cell
(\texttt{ppc}) both determine the particle movement/tunneling and the grid structure.
Whenever the given \texttt{ppc} is overrun in a particular spacetree cell, 
this cell is refined and the particles are hence sorted into the new spacetree 
nodes. 
This mirrors dynamic adaptivity assuming that the particle density correlates
to the smoothness of $E$ in (\ref{equation:pic:pde}).
Whenever multiple children of one refined spacetree node can be coarsened 
without violating the \texttt{ppc} constraint, we remove these children and 
lift the particles into the formerly refined cell.
All experimental code supports dynamic AMR.
All experiments apply reflecting boundary conditions at the border of
the unit square or cube, respectively, but do not change the particle 
velocities otherwise.
All experimental data result from a code with a static domain decomposition
deriving $\mathbb{P}$ from the spacetree due to graph partitioning.
The experiments are, if not stated otherwise, stopped after few time steps and
thus do not suffer (significantly) from ill-balancing. 
All domain decomposition/initialisation overhead is removed from the
measurements.

With all these parameters at hand, we can study the impact of total particle
count relative to the computational domain, we can study setups with extremely
inhomogeneous particle distribution vs.~homogeneous particle distributions, 
we can analyse the impact of the ratio of particle speed to minimal grid size,  
and we can study the interplay of particle density, \texttt{ppc} and time step size.

\subsection{Algorithmic properties}
\label{section:properties}

\begin{figure}[!ht]
\centering
  \includegraphics[width=0.44\linewidth]{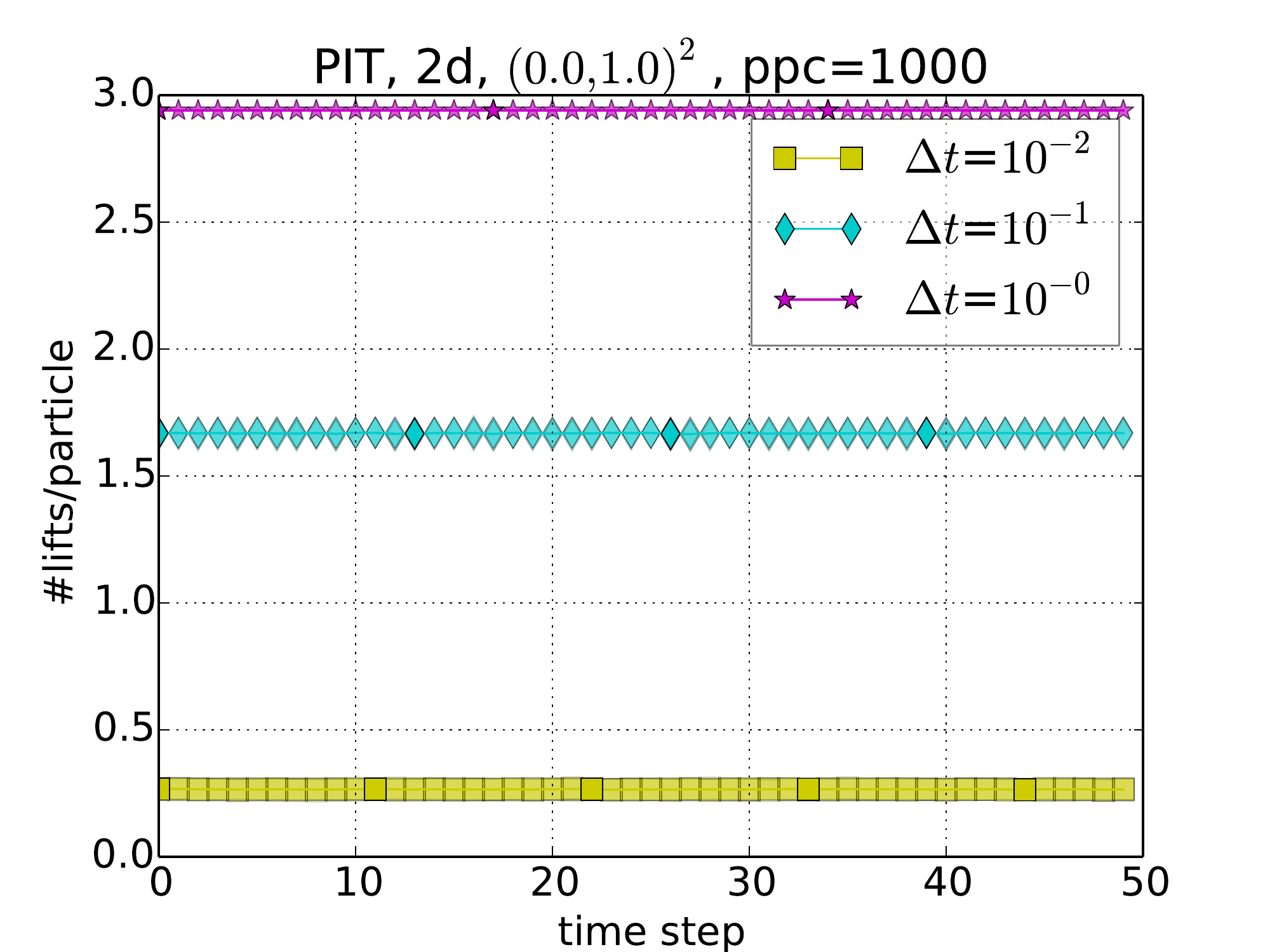}
  \includegraphics[width=0.44\linewidth]{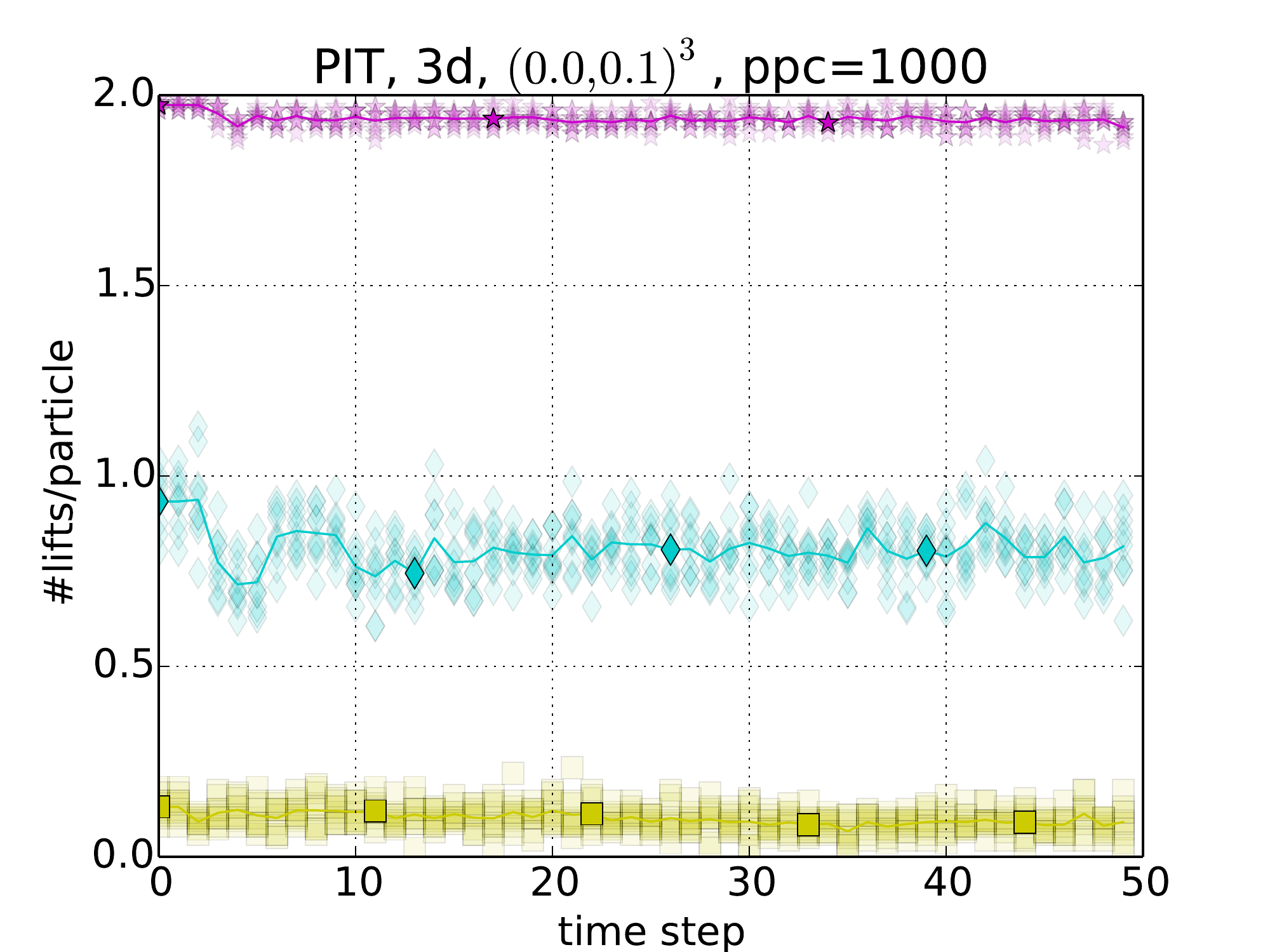}
  \includegraphics[width=0.44\linewidth]{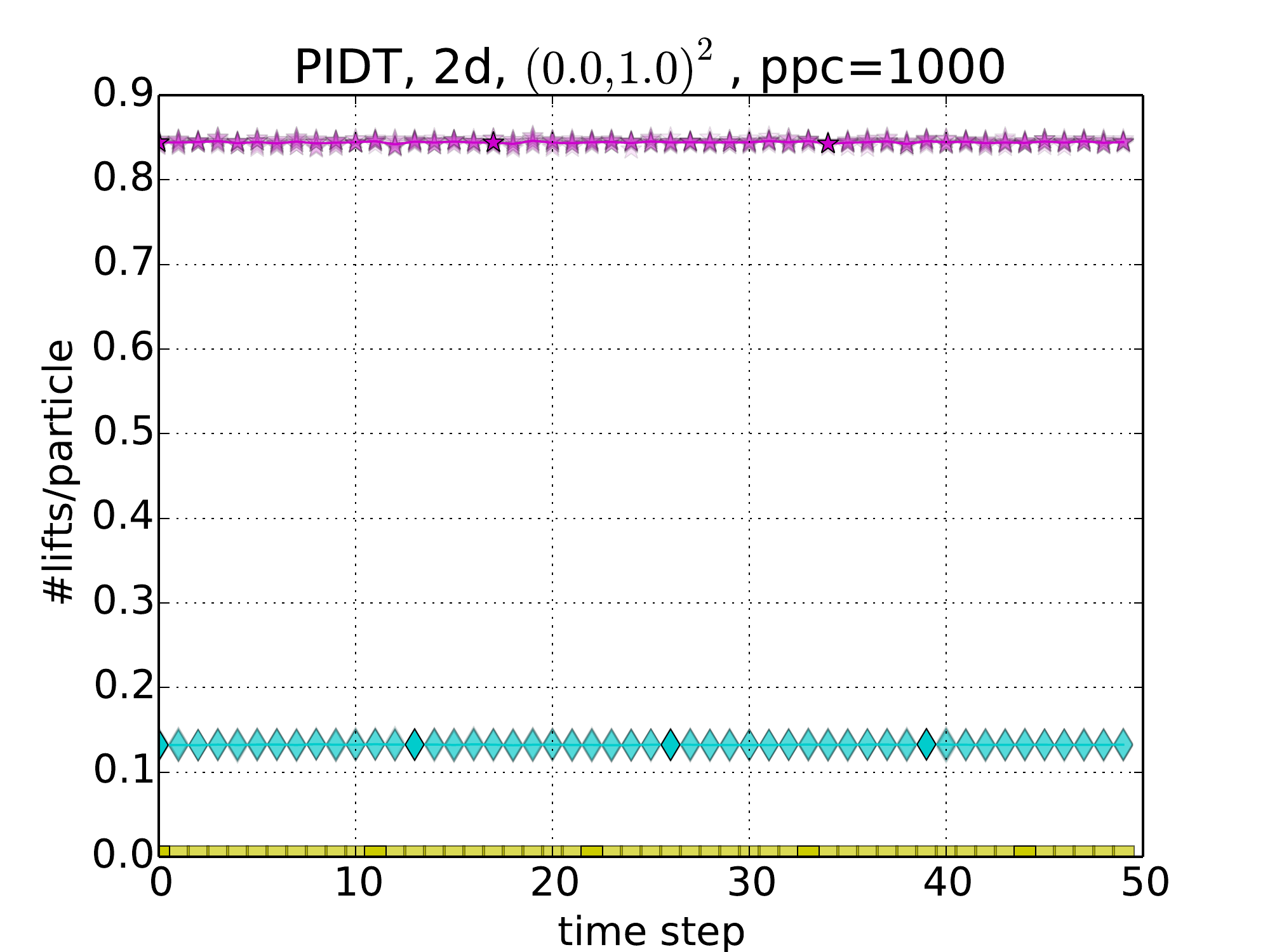}
  \includegraphics[width=0.44\linewidth]{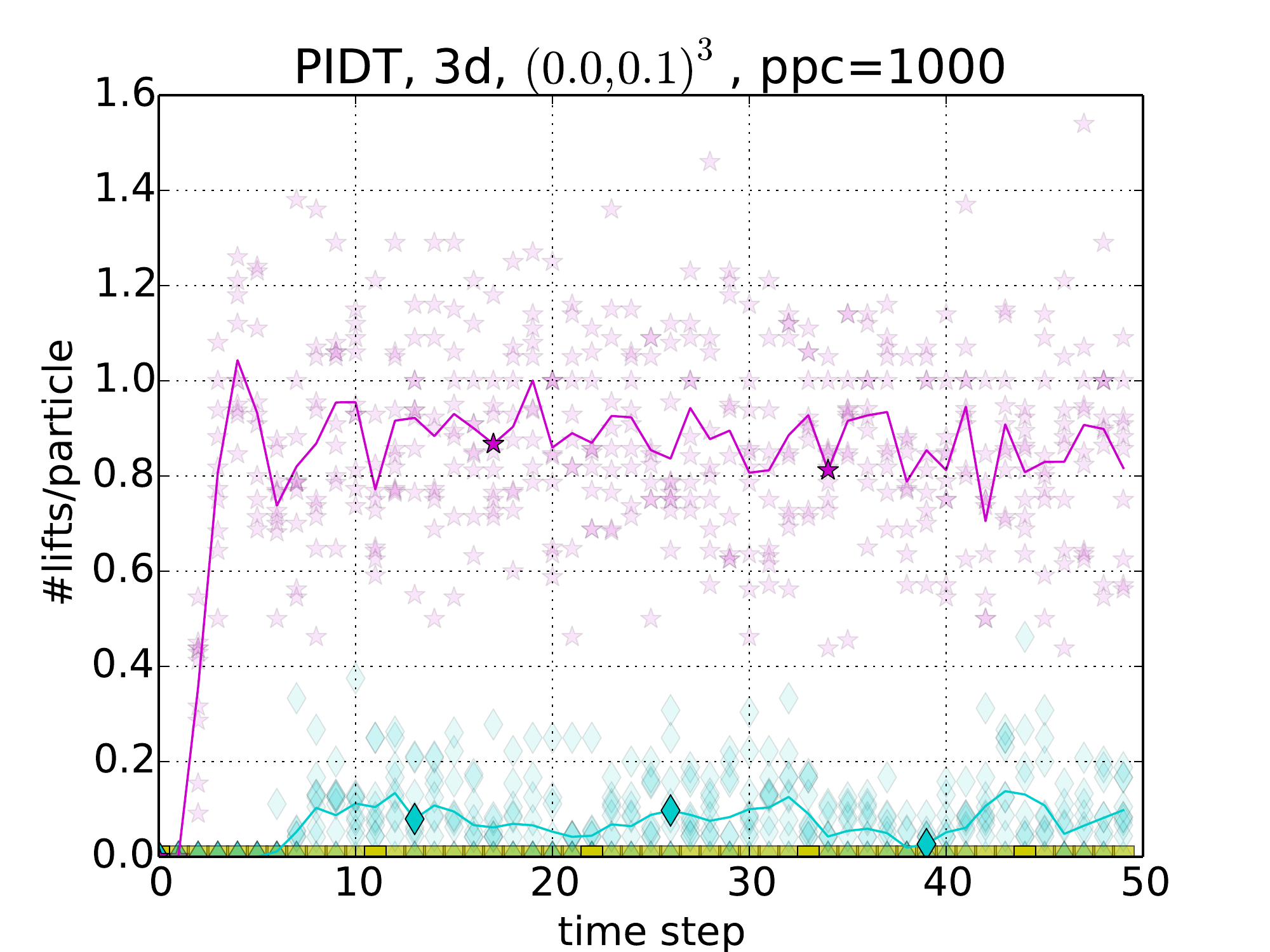}
  \caption{
    Lifts per particle per time step (\texttt{ppc}=1000).
    Left column: particles are homogeneously distributed
    among the unit square ($d=2$).
    Right column: particles are initially homogeneously distributed in
    $(0,0.1)^3 \subset \mathbb{R}^3$, i.e.~the grid is dynamically adaptive. 
    Results from ten random initial setups (blurry) with averages as solid
    lines.
    No lifts are observed at all for $\Delta t \leq 10^{-3}$.
  }
\label{figure:lifts-per-particle-with-different-dt}
\end{figure}

Prior to measuring the runtime, we focus on the lift behaviour for different
setups and count the average number of lifts per particle during the first 50
time steps.
Drop observations deduce from these.
All measurements are almost independent of (sufficiently big)
$\rho $---here set to $10^7$.
If a particle is lifted $n$ levels up in the spacetree, we count this as $n$
independent lifts.
When we fix \texttt{ppc} (\texttt{ppc}=1000 in Figure
\ref{figure:lifts-per-particle-with-different-dt}, e.g.), we observe that the
number of lifts remains de facto invariant throughout the simulation for a
globally homogeneous particle distribution.
For an inhomogeneous initial distribution, lifts occur only sporadically
(unless the time step size is very large), i.e.~the average number of
lifts is below one.
For both realisation variants, the number of lifts scales with the time step 
size, i.e.~the bigger the time step size the more lifts per particle, while 
\pit\ yields significantly more lifts than \pidt.  
If the time step sizes underrun a certain threshold,
\pidt\ comes along completely without lifts while \pit\ has a very low lift
count.

%

If we fix in turn the time step size and make
\texttt{ppc} a free parameter (see \ref{section:additional-experiments}), we
observe that any \texttt{ppc} yields, on average, time-independent behaviour for homogeneous
setups.
The bigger \texttt{ppc} the fewer lifts, and the number of lifts per
particle is always bigger for \pit\ than for \pidt\ where the number of lifts
is negligible for reasonably big \texttt{ppc} or coarse grids, respectively.
Larger \texttt{ppc} make each leaf represent larger geometric domains and thus
explain the \texttt{ppc} dependence.
If the particle distribution is inhomogeneous, the corresponding grid at setup
time is strongly adaptive if \texttt{ppc} is reasonably small.
Initial time steps of \pidt\ then face almost no lifts; similar as in the
homogeneous setup.
However, the number of lifts increases as the grid becomes more regular
(and coarser), and the particles distribute more homogeneously, until the lift
count drops again to its regular characteristics.
With increasing \texttt{ppc} the curve flattens out and shifts
to the right.

Due to the invariant velocity profile besides the reflecting boundary
conditions, particles are expelled from the subarea by the large (electric)
field.
Therefore, the global grid determined by \texttt{ppc} becomes coarser and more
regular (Figure \ref{fig:eyecatcher}).
It smoothens out.
As a consequence, subtrees of a certain depth hold more and more particles with
velocities of the same magnitude.
In return, the average particle density within each cell decreases.
The further (fast) particles move away from the initially dense
area, the more often they hit adaptivity boundaries, i.e.~hanging nodes.
\pidt\ has to lift them there.
For \pit, these additional lifts make almost no difference compared to the lifts
required anyway.
Particles hit adaptivity boundaries often if \texttt{ppc} is small.
As the grid smoothens out, also the lift counts drop.

Referring to real-world runs (\ref{section:validation}), such a
behaviour mirrors a setup where initially all particles are held within the subarea where the solution to (\ref{equation:pic:pde}) depending on the particle density is non-smooth and yields large particle accelerations.
The push out from this area then results from a mixed neutralising positively charged
background and negatively charged particles.

\subsection{Memory throughput}

\begin{table}[!ht]
  \caption{
    Stream particle throughput, i.e.~particles per second, for different
    particle counts $p$ using 1,2,4,12,16 cores or 16 cores plus
    hyperthreading (32) on SuperMUC.
    The upper section gives results for $d=2$, the lower for $d=3$. Best case throughputs per row are bold.
  }
  \label{table:memory-throughput-supermuc}
  \begin{center}
  { 
    \tiny
    \begin{tabular}{c|cccccc|c}
 $p$ & 1 & 2 &      4 &      8 &      12 &      16 &      32 \\
 \hline 
 $10^4$ & $\mathbf{3.03\cdot 10^7}$ & $2.98\cdot 10^7$ & $4.14\cdot 10^6$ & $2.46\cdot 10^7$ & $1.46\cdot 10^7$ & $2.21\cdot 10^7$ & $1.94\cdot 10^7$   \\ 
 $10^5$ & $6.66\cdot 10^7$ & $1.03\cdot 10^8$ & $\mathbf{1.35\cdot 10^8}$ & $1.06\cdot 10^8$ & $8.96\cdot 10^7$ & $1.02\cdot 10^8$ & $6.15\cdot 10^7$   \\ 
 $10^6$ & $7.03\cdot 10^7$ & $1.16\cdot 10^8$ & $1.58\cdot 10^8$ & $1.56\cdot 10^8$ & $\mathbf{2.25\cdot 10^8}$ & $2.41\cdot 10^8$ & $5.51\cdot 10^7$   \\ 
 $10^7$ & $7.14\cdot 10^7$ & $1.25\cdot 10^8$ & $1.63\cdot 10^8$ & $2.44\cdot 10^8$ & $\mathbf{2.63\cdot 10^8}$ & $2.42\cdot 10^8$ & $7.99\cdot 10^7$   \\ 
 $10^8$ & $7.17\cdot 10^7$ & $1.26\cdot 10^8$ & $1.68\cdot 10^8$ & $2.28\cdot 10^8$ & $2.65\cdot 10^8$ & $\mathbf{2.84\cdot 10^8}$ & $2.59\cdot 10^8$   \\ 
 \hline 
 $10^4$ & $2.48\cdot 10^7$ & $3.22\cdot 10^7$ & $\mathbf{3.25\cdot 10^7}$ & $2.39\cdot 10^7$ & $2.42\cdot 10^7$ & $2.27\cdot 10^7$ & $1.44\cdot 10^7$   \\ 
 $10^5$ & $5.11\cdot 10^7$ & $8.43\cdot 10^7$ & $1.14\cdot 10^8$ & $\mathbf{1.23\cdot 10^8}$ & $7.67\cdot 10^7$ & $2.44\cdot 10^7$ & $2.43\cdot 10^7$   \\ 
 $10^6$ & $4.93\cdot 10^7$ & $9.15\cdot 10^7$ & $1.18\cdot 10^8$ & $1.73\cdot 10^8$ & $1.81\cdot 10^8$ & $\mathbf{1.94\cdot 10^8}$ & $4.56\cdot 10^7$   \\ 
 $10^7$ & $5.40\cdot 10^7$ & $9.30\cdot 10^7$ & $1.24\cdot 10^8$ & $\mathbf{1.82\cdot 10^8}$ & $1.77\cdot 10^8$ & $1.63\cdot 10^8$ & $9.24\cdot 10^7$  \\ 
 $10^8$ & $5.41\cdot 10^7$ & $9.33\cdot 10^7$ & $1.25\cdot 10^8$ & $1.65\cdot 10^8$ & $1.84\cdot 10^8$ & $\mathbf{1.92\cdot 10^8}$ & $1.65\cdot 10^8$   \\
\end{tabular}
  }
  \end{center}
\end{table}

To be able to put runtime measurements into context, we first run a 
benchmark holding an array of particles that is iterated once per time step without any grid.
Each particle position is updated according to an explicit Euler integration
step, it is reflected at the domain boundaries, and the resulting position is
written back to the corresponding array position. 
Besides the position update, no computation or assignment to grid entities is
done and no data is reordered.
The implementation is the same source code fragment we use in
the spacetree algorithms.
As we aim to compare it with our parallel implementation running multiple
MPI ranks per node, we parallelise this embarrassingly parallel benchmark with a
plain parallel-for along the lines of the Stream benchmark
\cite{McCalpin:95:Stream}.
This yields upper bounds, as our spacetree implementation relies on MPI only and
thus has message passing overhead.

We analyse the throughput behaviour at hands of SuperMUC (Table
\ref{table:memory-throughput-supermuc}).
Polaris exhibits qualitatively the same behaviour.
On both systems, a single core cannot exploit the memory subsystem alone.
The bigger the particle count the more cores can be used effectively and the higher the throughput.
However, the throughput does not scale linearly with the core count due to
bandwidth restrictions.
$3.00 \cdot 10^8$ for $d=2$ and $2.00 \cdot 10^8$ for 
$d=3$ are upper bounds for the throughput of our particle codes in
the absence of any solver, i.e.~of any `real' computation, on SuperMUC.
On Polaris, these best case thresholds have to be doubled due to the higher
clock rate and the Turbo Boost.
Besides for small particle numbers, less than
eight threads/ranks per node do not make the nodes run into bandwidth
saturation.
As a consequence, all parallel experiments deploy six MPI ranks per node from
hereon.
This heuristic choice valid for both SuperMUC and Polaris yields a reasonable
core usage while memory subsystem effects are not dominant.
It also anticipates that the presence of a grid to maintain increases the
average per-particle memory footprint.


\subsection{Single core results}

\begin{figure}[!ht]
\centering
\includegraphics[width=0.44\linewidth]{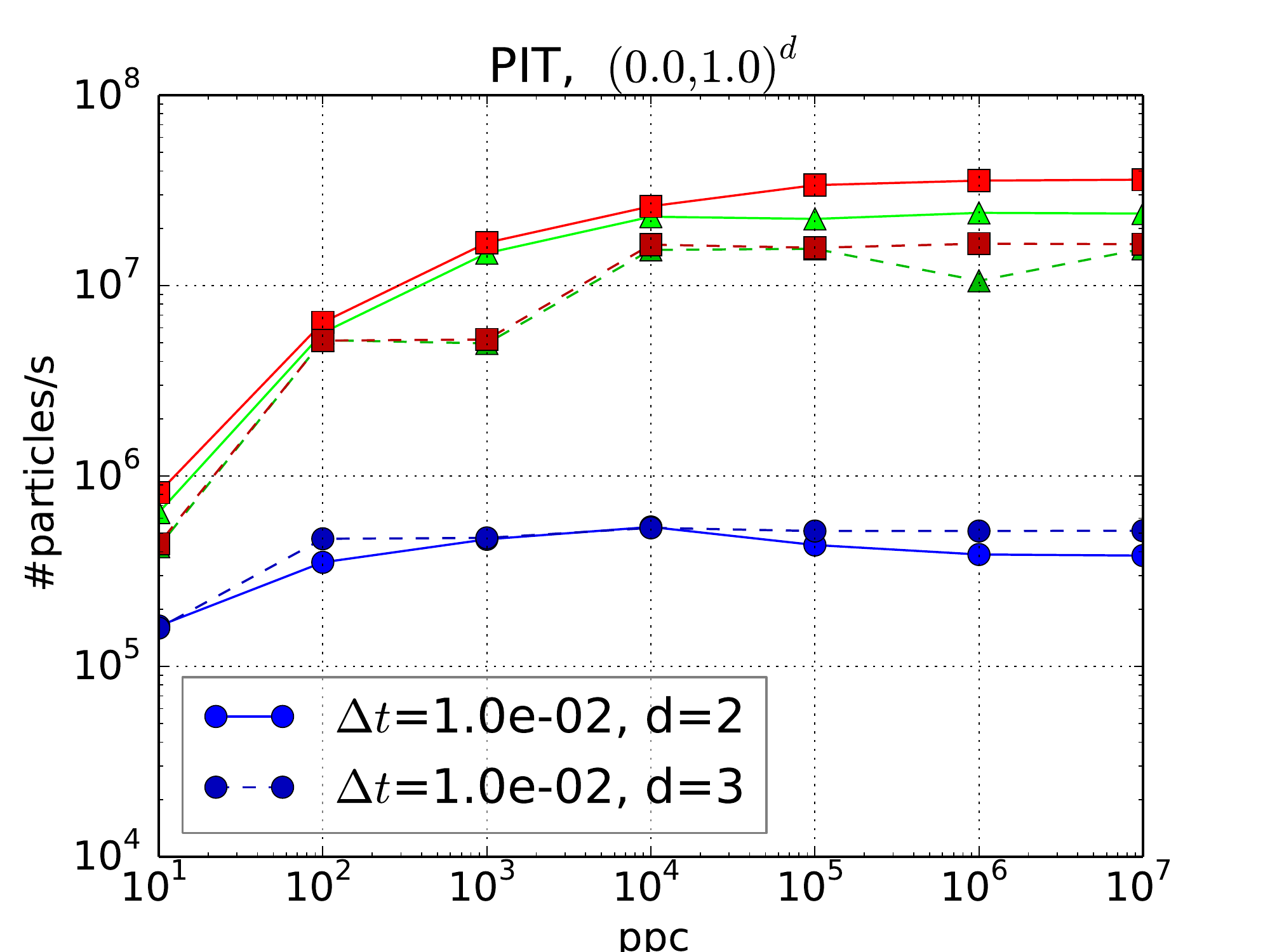}
\includegraphics[width=0.44\linewidth]{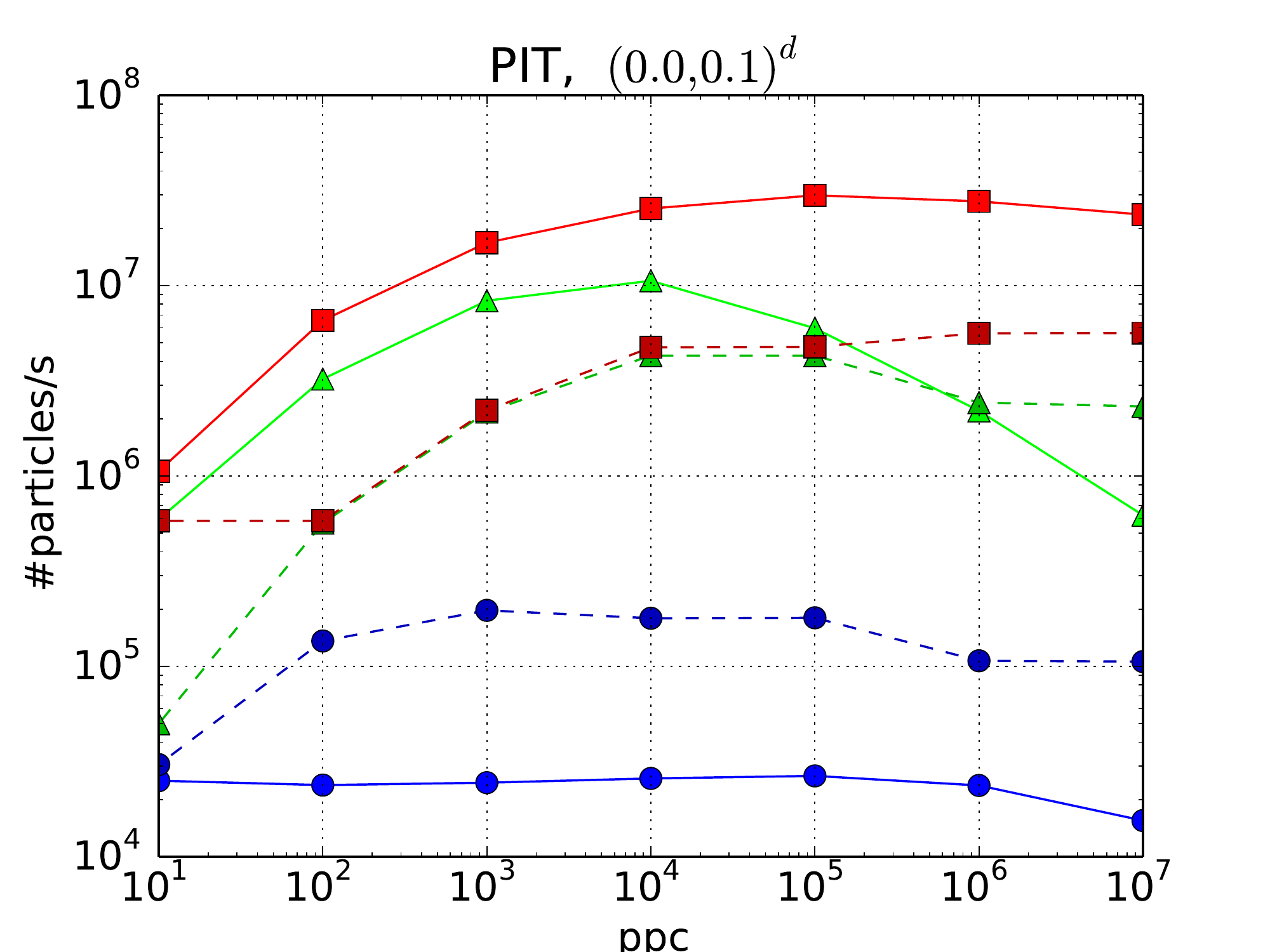}
\includegraphics[width=0.44\linewidth]{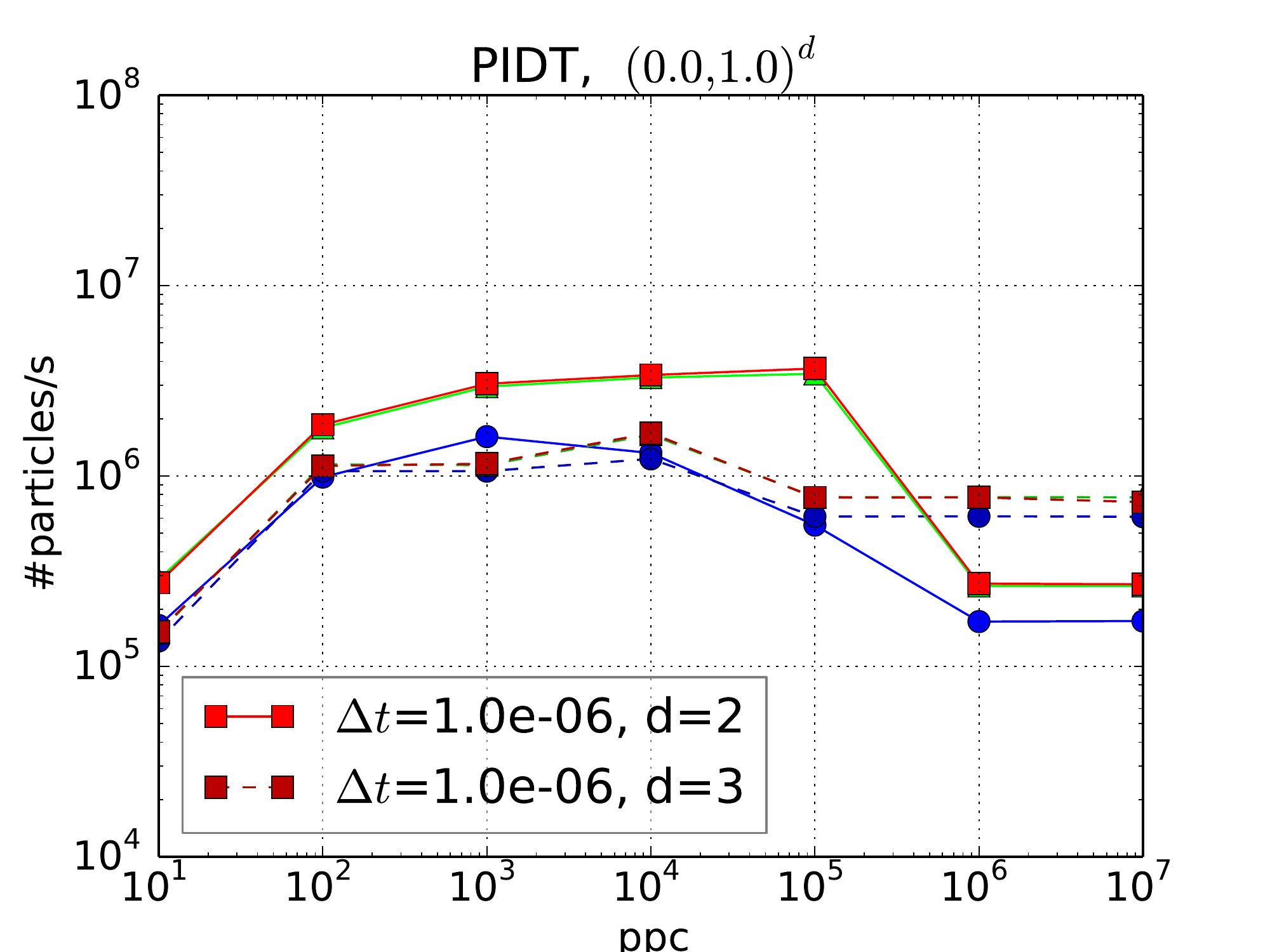}
\includegraphics[width=0.44\linewidth]{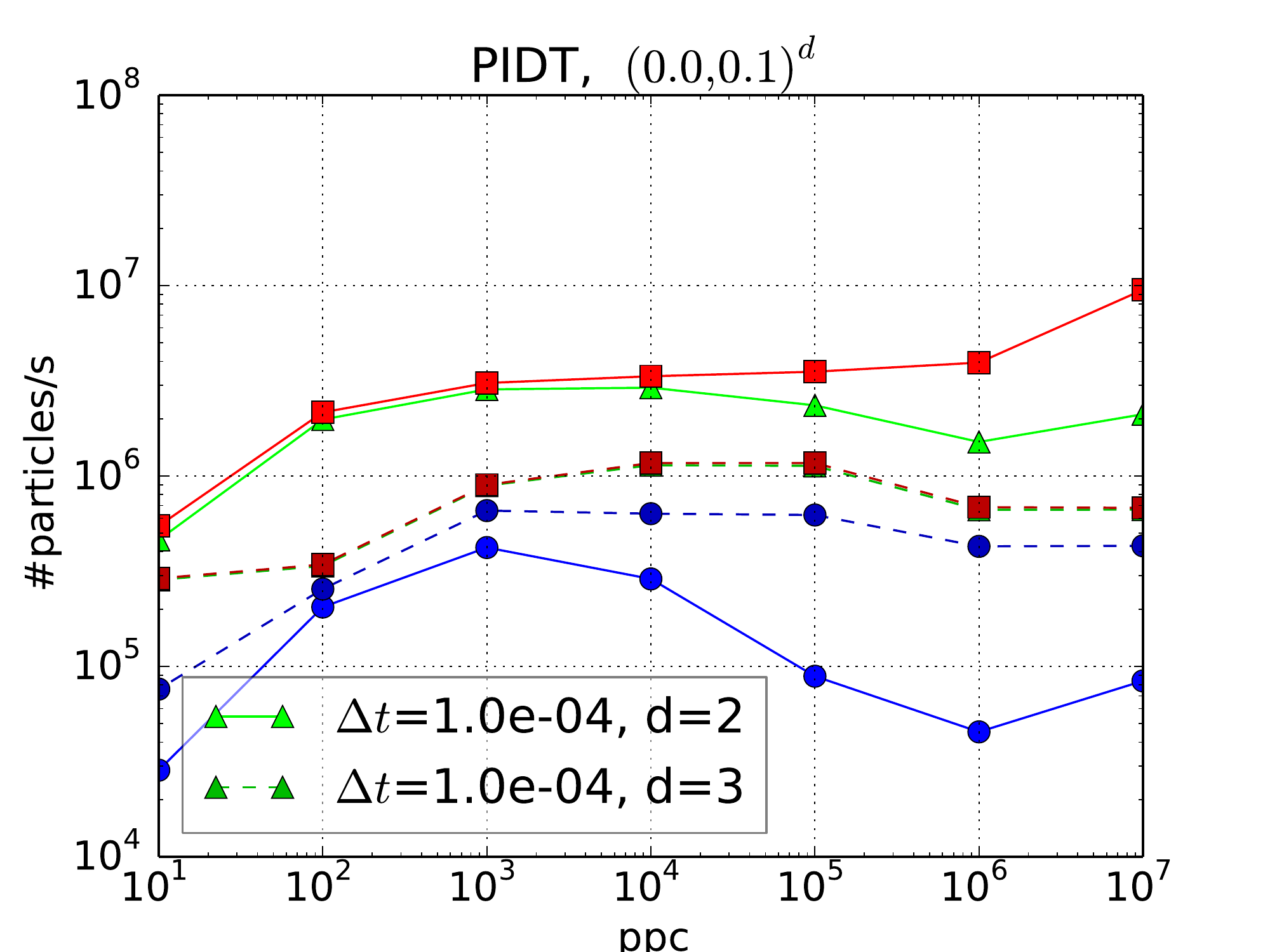}
\caption{
  Single core throughput with homogeneous (left) or inhomogeneous/breaking dam
  (right) start conditions for \pit\ (top) and \pidt\ (bottom).
  $d=2$ marked by solid lines and $d=3$ by dashed lines.  
  Figures compare impact of \texttt{ppc} choice on throughput for regular
  (homogeneous) and dynamically adaptive (breaking dam) grids.
  Results from SuperMUC.
  }  
\label{figure:single-core-throughput}
\end{figure}

We next measure the particle throughput for homogeneous and
inhomogeneous start scenarios with different $\Delta t$ on a single core with
$10^7$ particles for 50 time steps, while we still neglect arithmetics per particle (Figure \ref{figure:single-core-throughput}).
\pit's throughput monotonously increases with increasing reasonable \texttt{ppc}.
It decreases with increasing time step size. 
Any reduction of lifts due to a bigger \texttt{ppc} or smaller time step
sizes pays off for a homogeneous start setup.
For an inhomogeneous start setup, we observe a decreasing performance with
bigger time step sizes as the grid then changes faster.
This picture would change if we fixed the simulation time rather than the time
step count.
\pidt\ is typically outperformed by \pit\ due to the more complicated
algorithm despite in situations where particles move very fast in an adaptive
setting.
$\mbox{\texttt{ppc}}\approx1000$ here yields the best throughput.
Polaris and SuperMUC results do coincide if scaled with the clock rate,
i.e.~the duplication of throughput rates on Polaris cannot be observed again.

%
%

Hardware counter measurements with Likwid \cite{Treibig:10:Likwid} reveal
that cache misses for both approaches are negligible.
They resemble exactly the results reported in
\cite{Bader:13:SFCs,Weinzierl:2009:Diss,Weinzierl:11:Peano} and references
therein for other application areas.
Our algorithms' AMR code base Peano \cite{Software:Peano} relies on a
depth-first alike tree traversal \cite{Weinzierl:15:Peano}  that picks up
H\"older continuity properties of the Peano space-filling curve
\cite{Weinzierl:11:Peano}.
This implies advantageous spatial and temporal memory access characteristics.
However, we expect a reimplementation with another code base to
yield advantageous properties as well, as all \pic\ and \pidt\
ingredients follow a strict element-wise/local formulation and as the memory
subsystem in above measurements is underutilised.
A slight increase of cache misses thus does not neccessarily pollute runtime
results significantly, if proper prefetching is applied, while the basic
memory usage profile is advantageous since all operations are local---either
accessing neighbours or parents/children within the grid.

There is a sweet spot from which it pays off to use \pidt\ rather than
\pit.
This payoff point depends on time step size and \texttt{ppc}. 
For sufficiently big time steps and reasonable small \texttt{ppc}, \pidt\
outperforms \pit.
This is due to the reduction of lift operations, i.e.~due to \pit\ having more
particles tunneling.
\pit\ is around a factor of three slower than the pure particle throughput on a
single core if $\Delta t$ is very small.
\pidt\ is around a factor of $3^d$ slower due to the particle sorting overhead.
The remainder of the experiments runs all settings for $\Delta t=10^{-4}$ as
this is an interesting regime where \pit\ has not yet overtaken \pidt\ for the
majority of experiments.
Putting runtime evaluations into relation to the number of lifts always allows
us to predict properties of all particle handling algorithms.
The remainder of the experiments also restricts to homogeneous settings,
i.e.~the particles initially are distributed homogeneously among the computational
domain, as particle characteristics then remain invariant.
%
Putting runtime evaluations into relation to the number of characteristic
\texttt{ppc} allows us to predict properties of simulations where the
computational domain comprises different spatial regions with different particle
distributions.

\subsection{Parallel \pit}
\label{section:results:pit}

\begin{figure}[!ht]
\centering
  \includegraphics[width=0.49\linewidth]{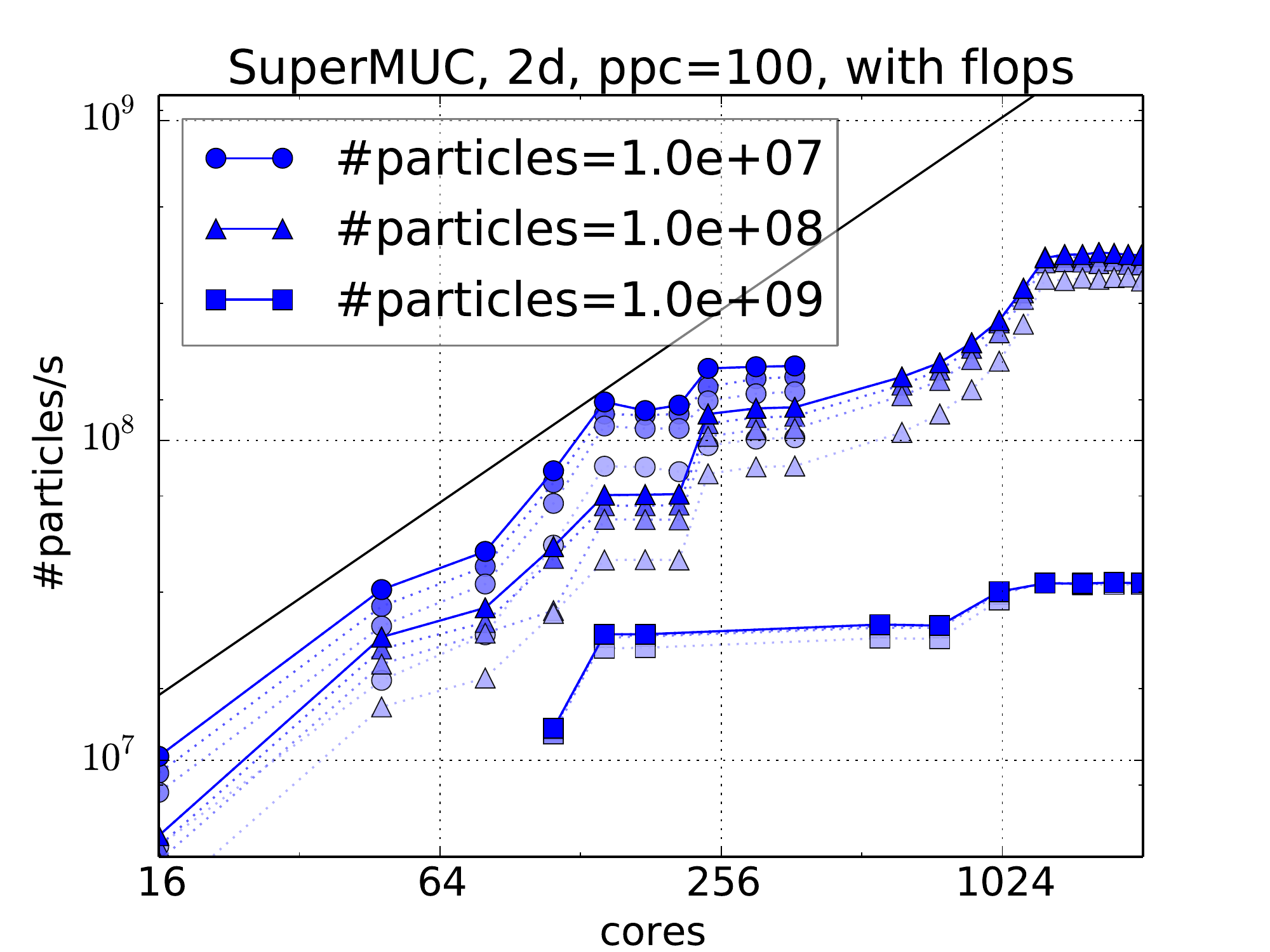}
  \includegraphics[width=0.49\linewidth]{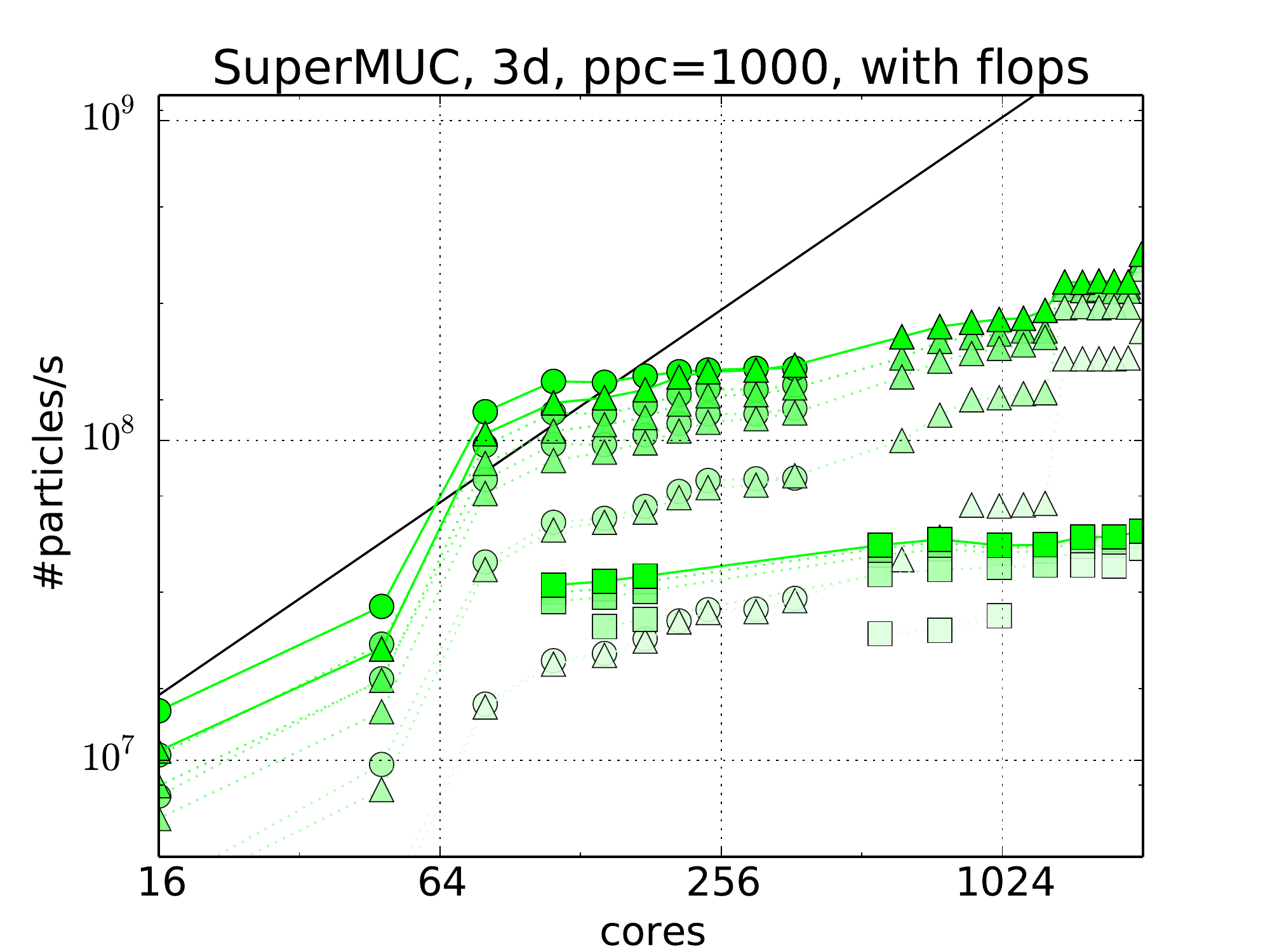}
  \\
  \includegraphics[width=0.49\linewidth]{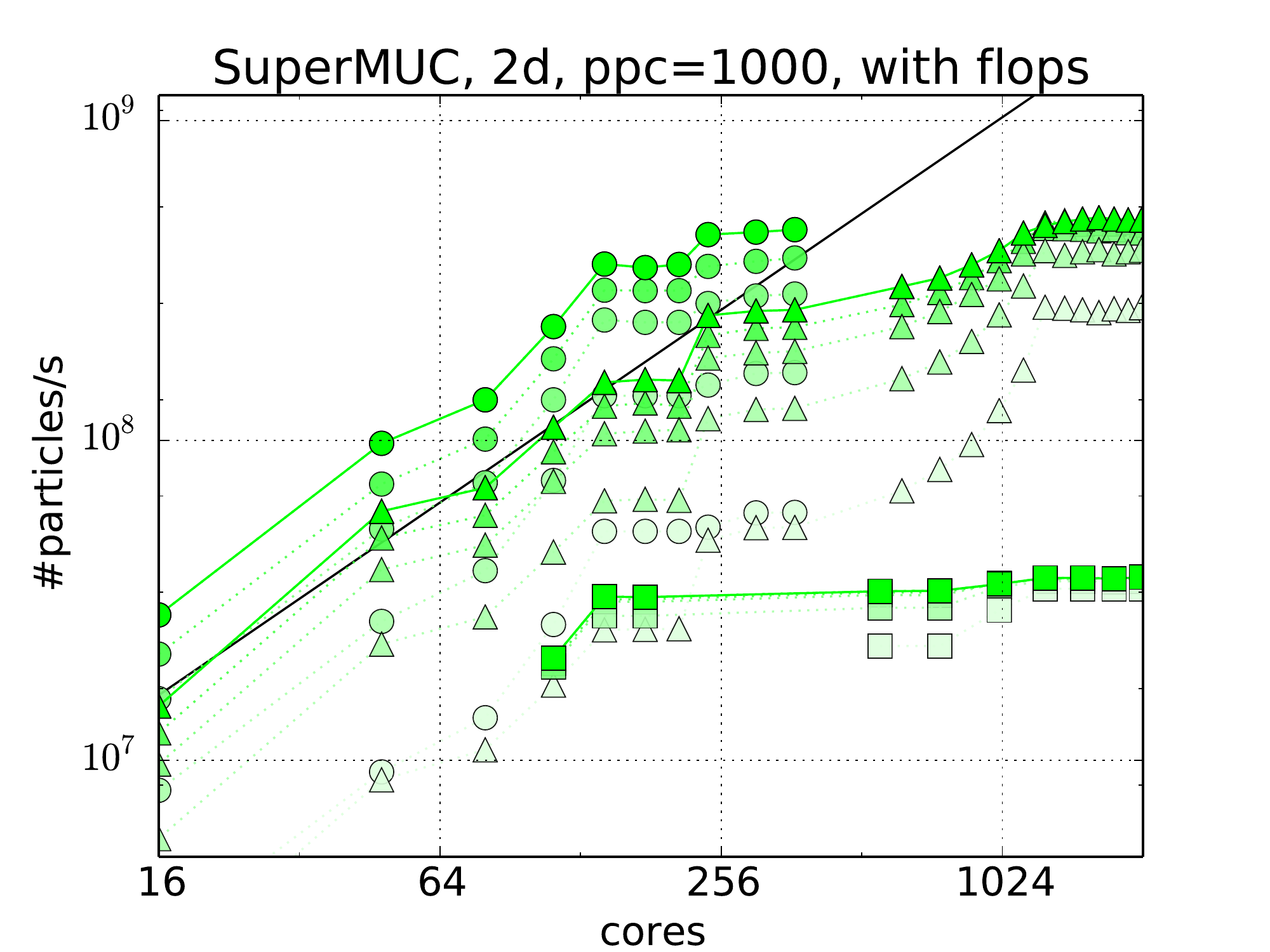}
  \includegraphics[width=0.49\linewidth]{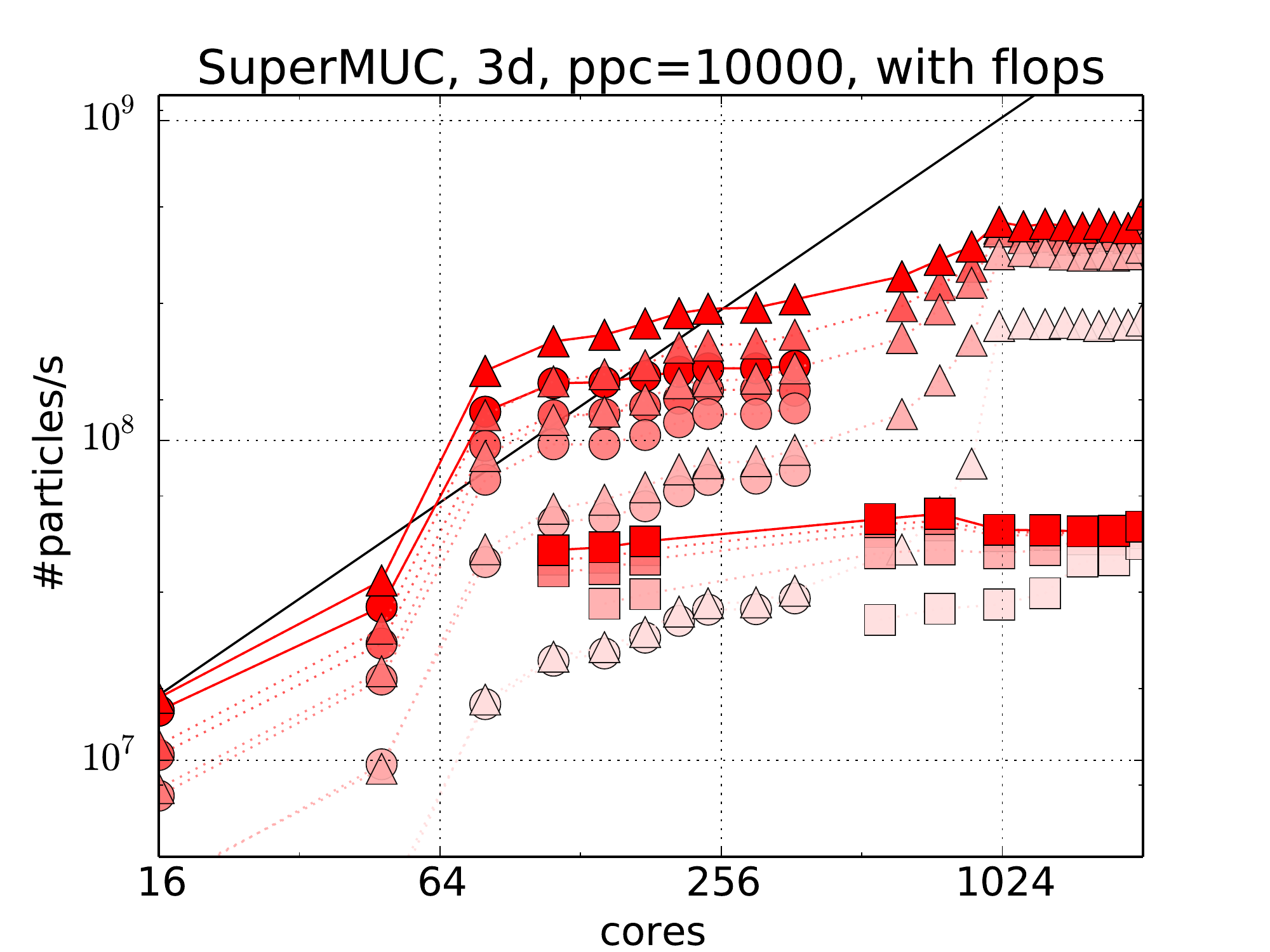}
  \caption{
    Parallel throughput of \pit.
    Each plot studies three different total particle numbers (different
    markers) for one fixed \texttt{ppc} and dimension $d$. 
    Each particle study is ran five times with 0, 128, 256, 1024 or
    4096 flops per particle.
    The brighter the points the higher the number of flops. 
    The saturated points with a solid line are measurements without any compute
    flops per particle, i.e.~solely the particle move and resort into the grid.
    Different colours here pick up the colouring from Figure
    \ref{figure:results:pit-polaris} and \pidt\ experiments.
  }
  \label{figure:results:pit-supermuc}
\end{figure}

We start our parallel studies with \pit\ and artificially perform an additional
0, 128, 256, 1024 or 4096 floating point operations per particle to emulate the
solving of a PDE and to be able to study the impact of this additional workload
on the particle performance. 
Each experiment hence was performed five times. 

Four properties become apparent from measurements on SuperMUC (Figure
\ref{figure:results:pit-supermuc}):
First, the more computation is done per particle the lower the throughput, but this difference almost vanishes for high core counts.
Second, the throughputs for $d=2$ and $d=3$ approach each other for high core
counts.
Third, \pit's $2d$ scaling is close to linear up to a given threshold.
For $d=3$, the results are rougher, but a similar threshold is hit for
higher core counts.
If the core count exceeds this threshold, the throughput stagnates.
Fourth, the more particles the lower the parallel efficiency.
Eventually, \pit\ does not scale for $10^9$ particles.
We observe an inverse weak scaling with respect to particle numbers.

The behaviour results purely from the particle handling, as the synchronisation
of the grid along subdomain boundaries is neglectable due to our realisation
from \cite{Schreiber:13:MetaData}.
Any master has to wait for all its workers before it may ascend, as the workers
might lift particles.
Any node triggers its send to its master after the traversal of its local domain
has finished.
This is a partial synchronisation of the ranks and a blocking data exchange
realising a reduction.
The latter is sensitive to latency and bandwidth restrictions.
As master-worker data exchange prelude or follow the actual local particle
handling, the computational work per particle does influence the runtime, as  
it cannot overlap with the communication. 
The fewer computation to be done, however, i.e.~the more cores participate in
the computation, the lower the impact of this work and the more severe 
communication bounds.
As the logical master-worker tree topology is broader for $d=3$ than for $d=2$,
the synchronisation runtime pressure at the workers is higher for $d=3$.
As latency is critical for the global reduction and its counterpart when we
start up the cores and drop particles, latency dominates the runtime if ``too''
many cores collaborate. 
The performance then stagnates.
As bandwidth is critical for the master-worker data exchange, more particles
slow down the reduction phase and its startup counterpart. 
\pit\ scales only in a very limited parameter regime.

%

 \begin{figure}[!ht]
 \centering
   \includegraphics[width=0.49\linewidth]{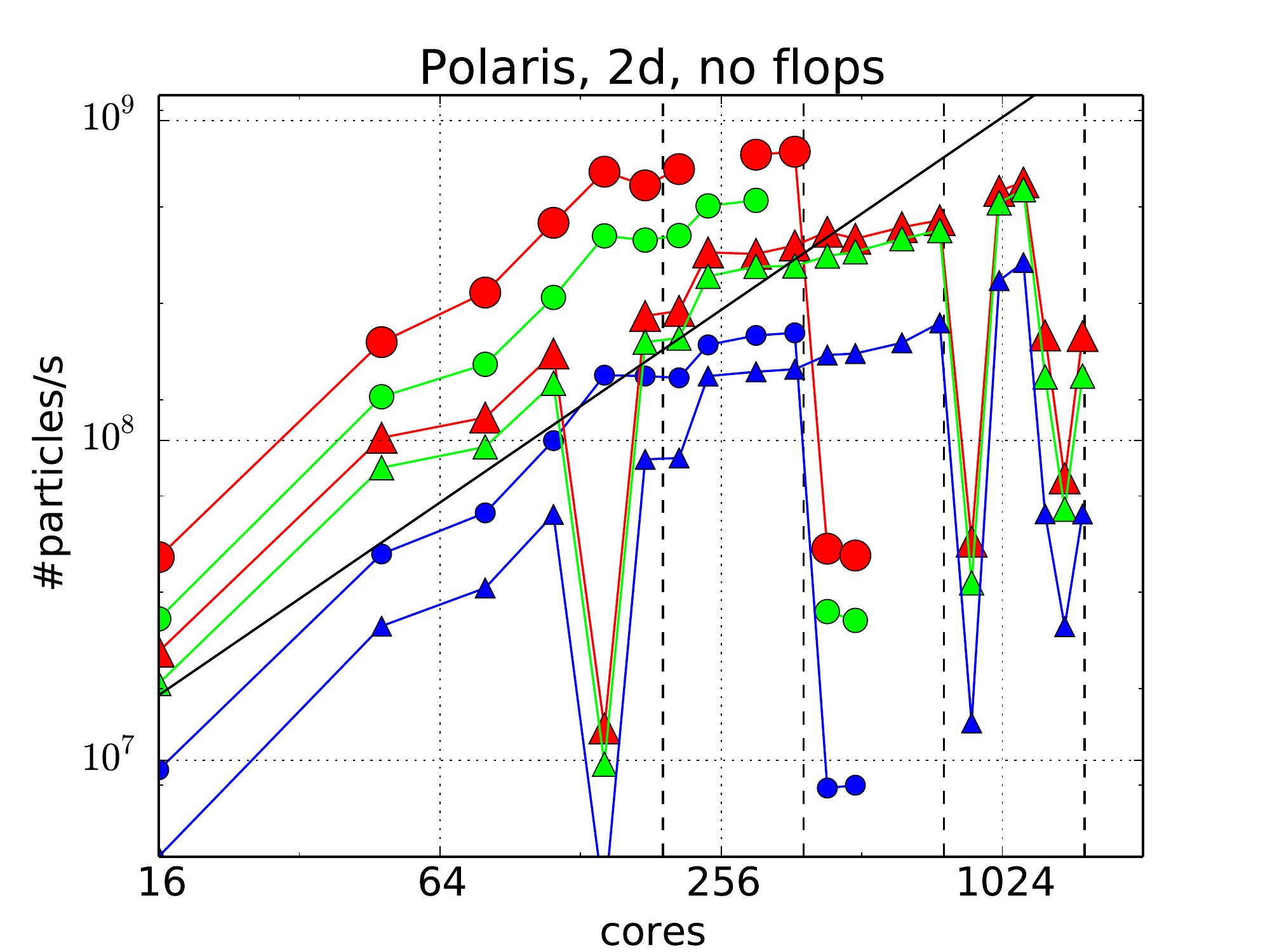}
   \includegraphics[width=0.49\linewidth]{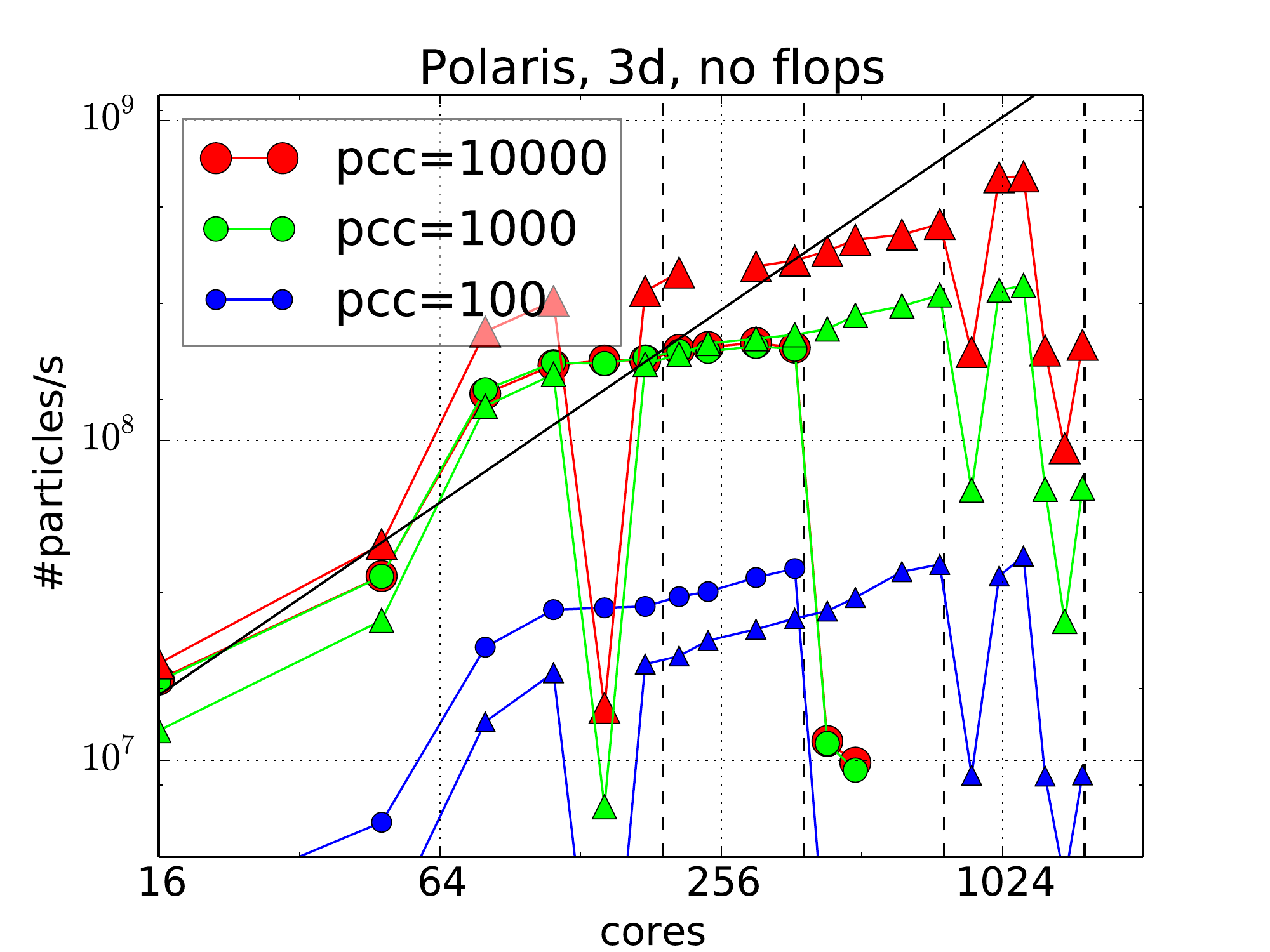}
   \caption{
     Results from
     Figure \ref{figure:results:pit-supermuc} without any additional flops rerun on Polaris. 
     The marker size/colour identifies the \texttt{ppc}, while circles plot
     $10^7$ particles and triangles $10^8$ particles.
     Vertical lines mark whenever the ports of one switch layer theoretically
     are exceeded.
     In practice, more switches are used for lower (small) core counts already.
   }
 \label{figure:results:pit-polaris}
 \end{figure}

%
%

Basically, its scaling is determined by the hardware characteristics plus the
tight synchronisation. 
This effect becomes evident for the same experiments on Polaris (Figure
\ref{figure:results:pit-polaris}).
Polaris exhibits slightly better results for moderate core counts.
However, the throughput drops whenever the core count requires the supercomputer
to employ another level of switches.
Due to a good placement that anticipates broken or overbooked nodes, the
resulting performance drops appear slightly prior to the theoretical switch
capacity.
Notably, the more restrictive network topology introduces an inverse weak
scaling effect for a high core count to particles ratio, where an increase of
cores introduces a degeneration of the throughput due to increased
communication/synchronisation pressure.

\subsection{Parallel \pidt}
\label{section:results:pidt}

%
%
%
%

\begin{figure}[!ht]
\centering
  \includegraphics[width=0.49\linewidth]{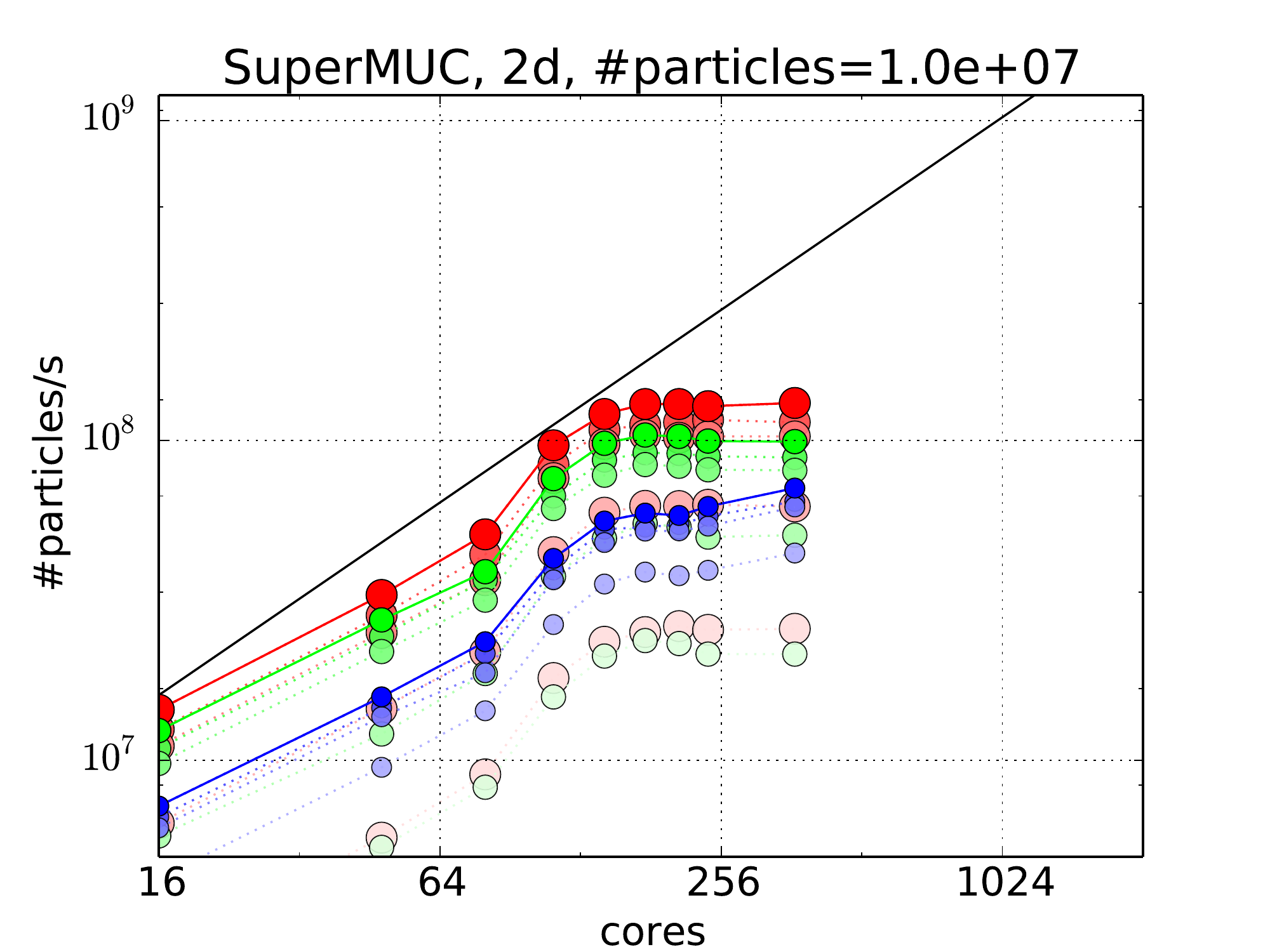}
  \includegraphics[width=0.49\linewidth]{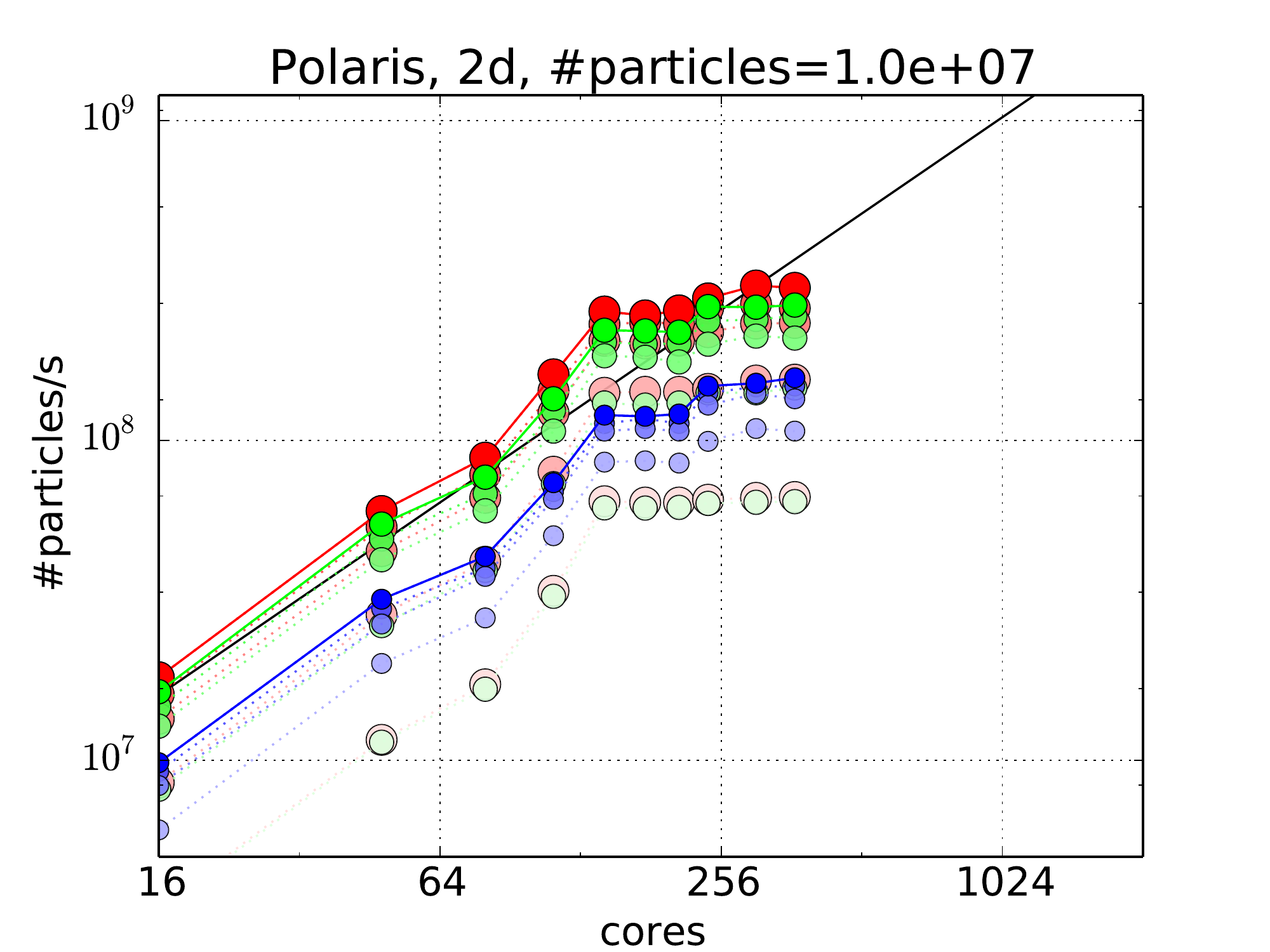}
  \\
  \includegraphics[width=0.49\linewidth]{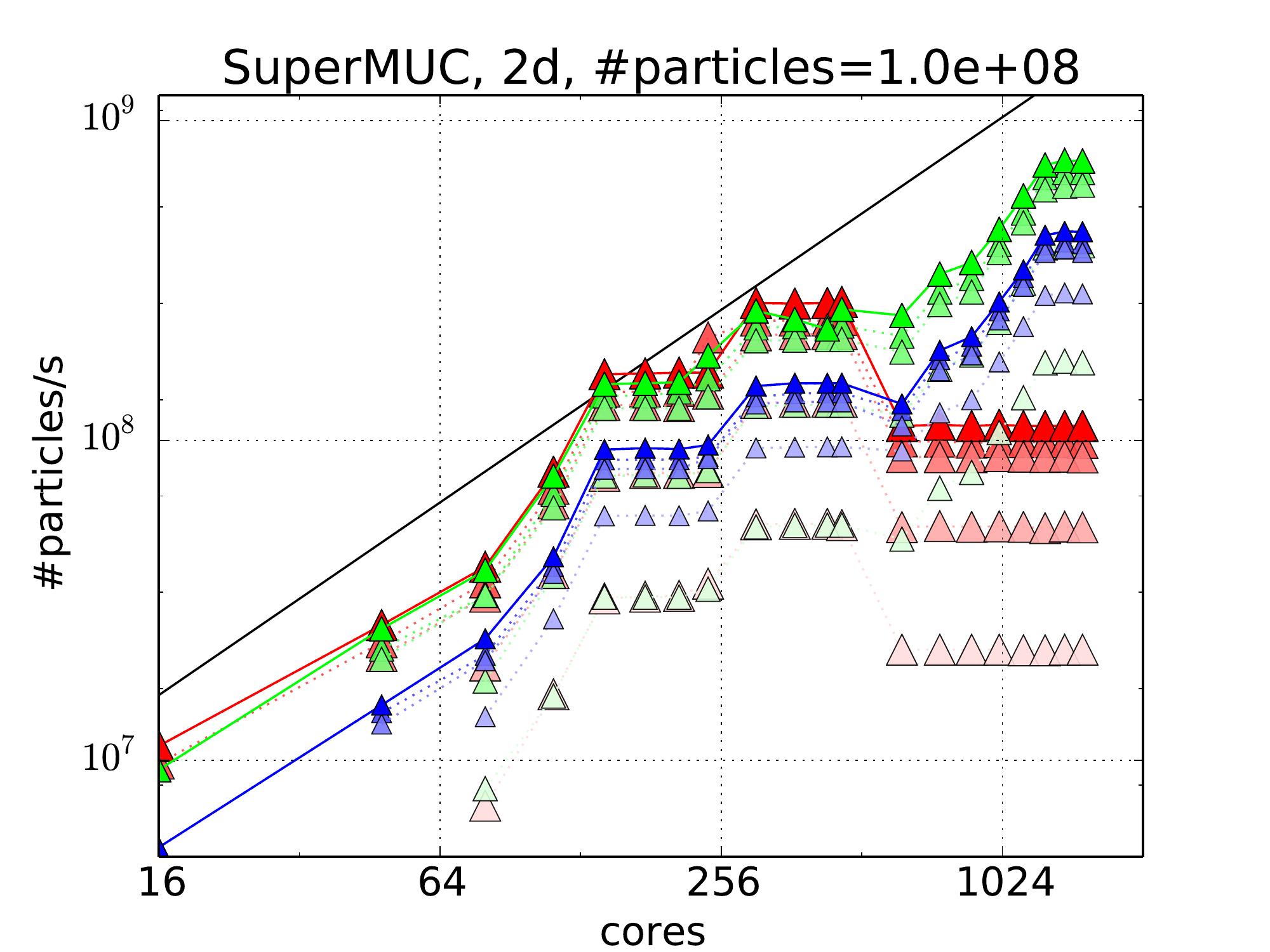}
  \includegraphics[width=0.49\linewidth]{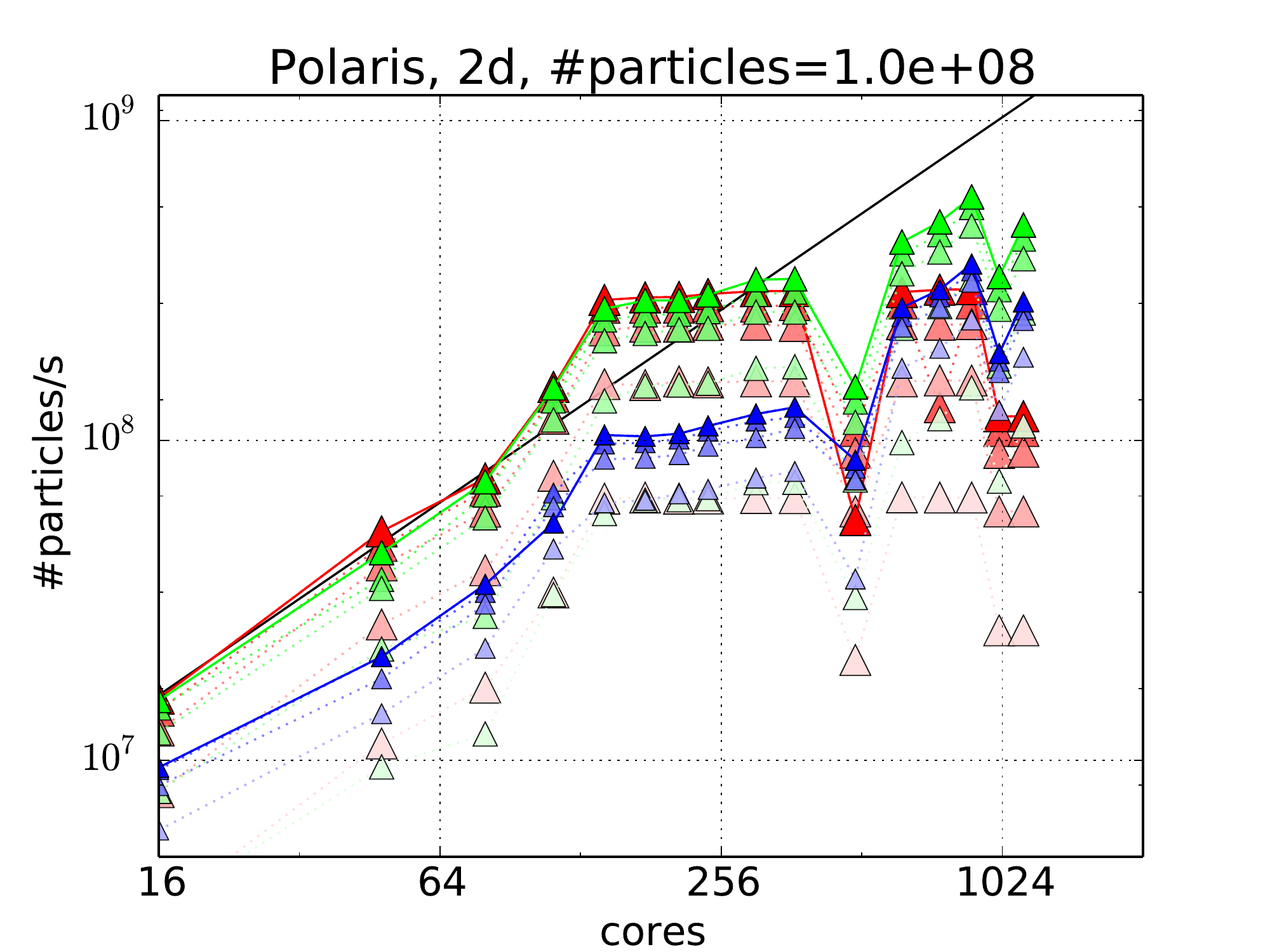}
  \\
  \includegraphics[width=0.49\linewidth]{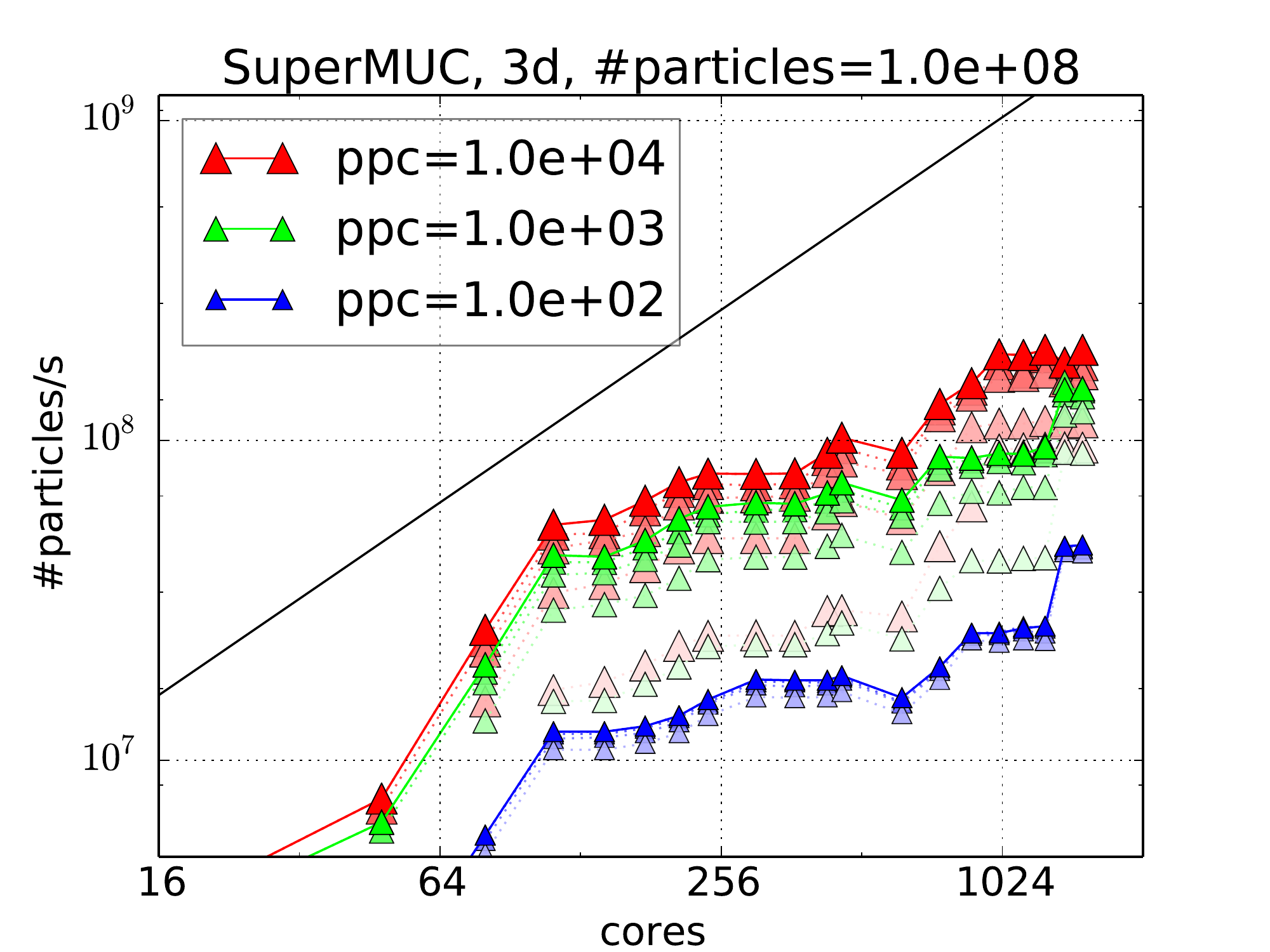}
  \includegraphics[width=0.49\linewidth]{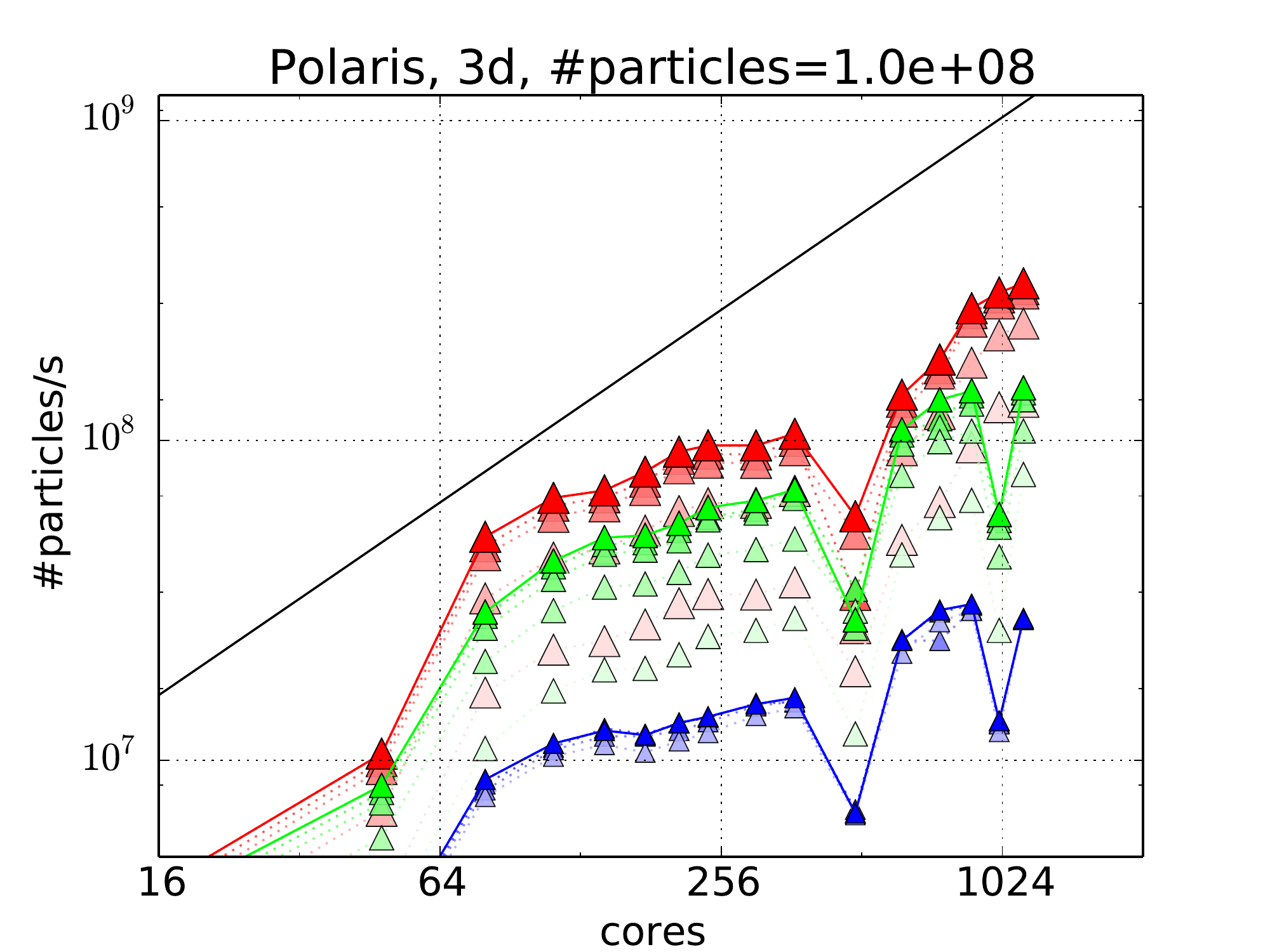}
  \\
  \caption{
    Parallel throughput of \pidt\ for SuperMUC (left) and Polaris (right).
    The darker the dots the fewer flops are added to each particle. 
    Solid lines are worst case setups with no computation per particle
    besides position updates.
    The bigger the marker the bigger \texttt{ppc} $\in \{10^2,10^3,10^4\}$
    (blue,green,red).
    All marker types pick up conventions of the other plots.
  }
\label{figure:results:pidt}
\end{figure}

We next rerun all experiments for \pidt\ and compare the outcomes to the \pit\
results.
Three differences become apparent (Figure \ref{figure:results:pidt}):
First, the impact of the operations per particle diminishes. 
Yet, lower arithmetic intensity still means higher throughput.
Second, \pidt\ inherits its lower throughput compared to \pit\ for low core
counts.
It still is slower.
Third however, \pidt\ scales better than \pit\ if the number of particles is
reasonably big and \texttt{ppc} is small.
There still is a stagnation threshold, but this threshold is higher than for
\pit, i.e.~\pidt\ overtakes its sibling algorithm for a decent core count.

The improvement of \pidt\ compared to \pit\ stems from the fact that \pidt\
exchanges particles both via master-worker relations and along subdomain boundaries.
The exchange along these boundaries can be realised asynchronously.
This allows \pidt\ to hide computations behind non-blocking MPI calls.
Such a hiding is the more effective the more computational workload per
particle.
However, it also depends on the dimensionality as the domain boundary is a
$d-1$-dimensional submanifold.
Still, \pidt\ restricts data along the spacetree each iteration and thus is
vulnerable to latency effects.
However, the actual number of lifts along the master-worker hierarchy is
significantly smaller than for \pit\ and thus attenuates the impact of bandwidth
constraints.
Big \texttt{ppc} make the underlying spacetree more shallow and thus counteract
to this effect while small \texttt{ppc} lead to deep spacetrees where boundary
data exchange and parallel domain handling gain weight in the total runtime
profile.

\subsection{\rapidt}

\begin{figure}[!ht]
\centering
  \includegraphics[width=0.49\linewidth]{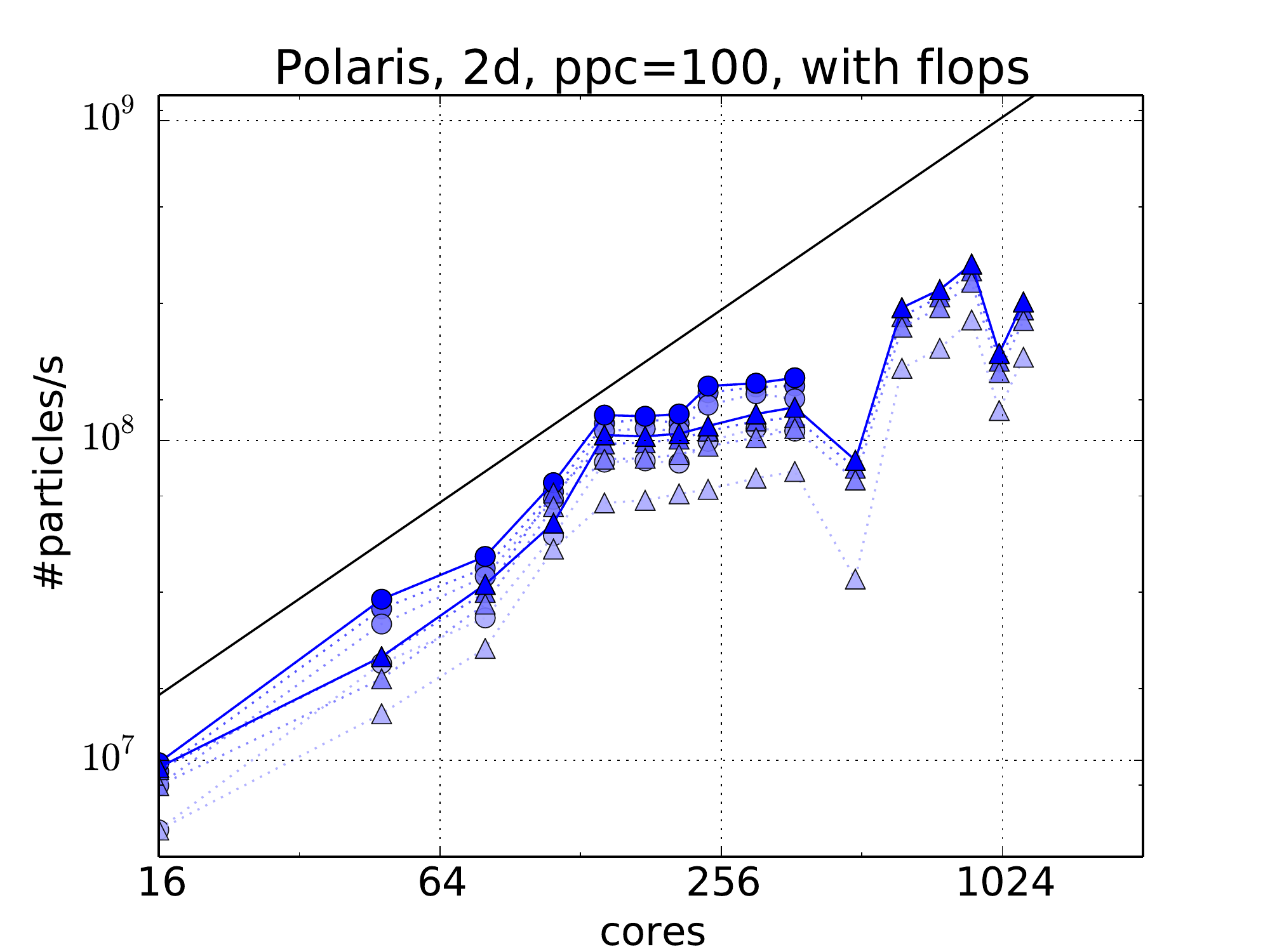}
  \includegraphics[width=0.49\linewidth]{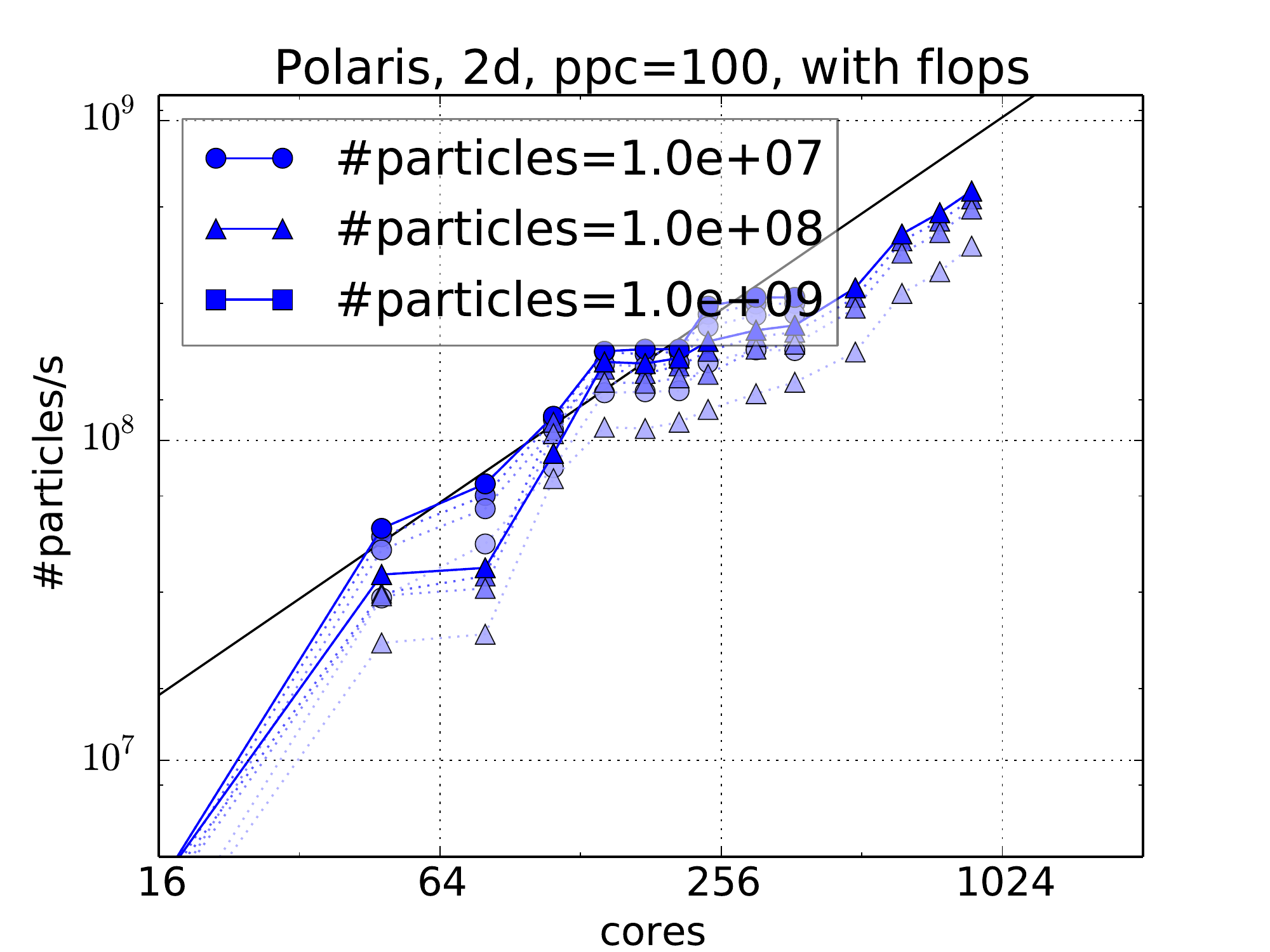}
  \\
  \includegraphics[width=0.49\linewidth]{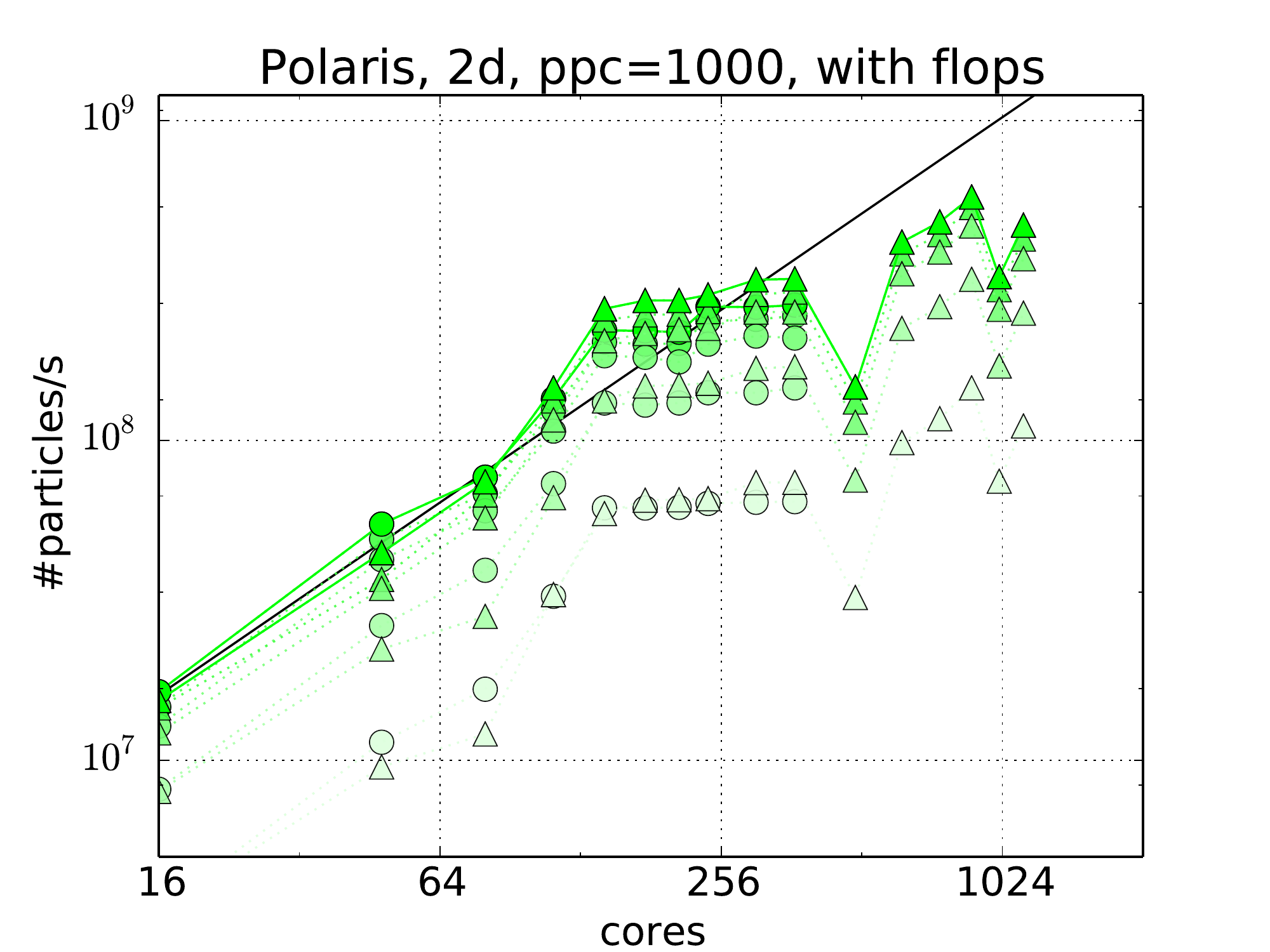}
  \includegraphics[width=0.49\linewidth]{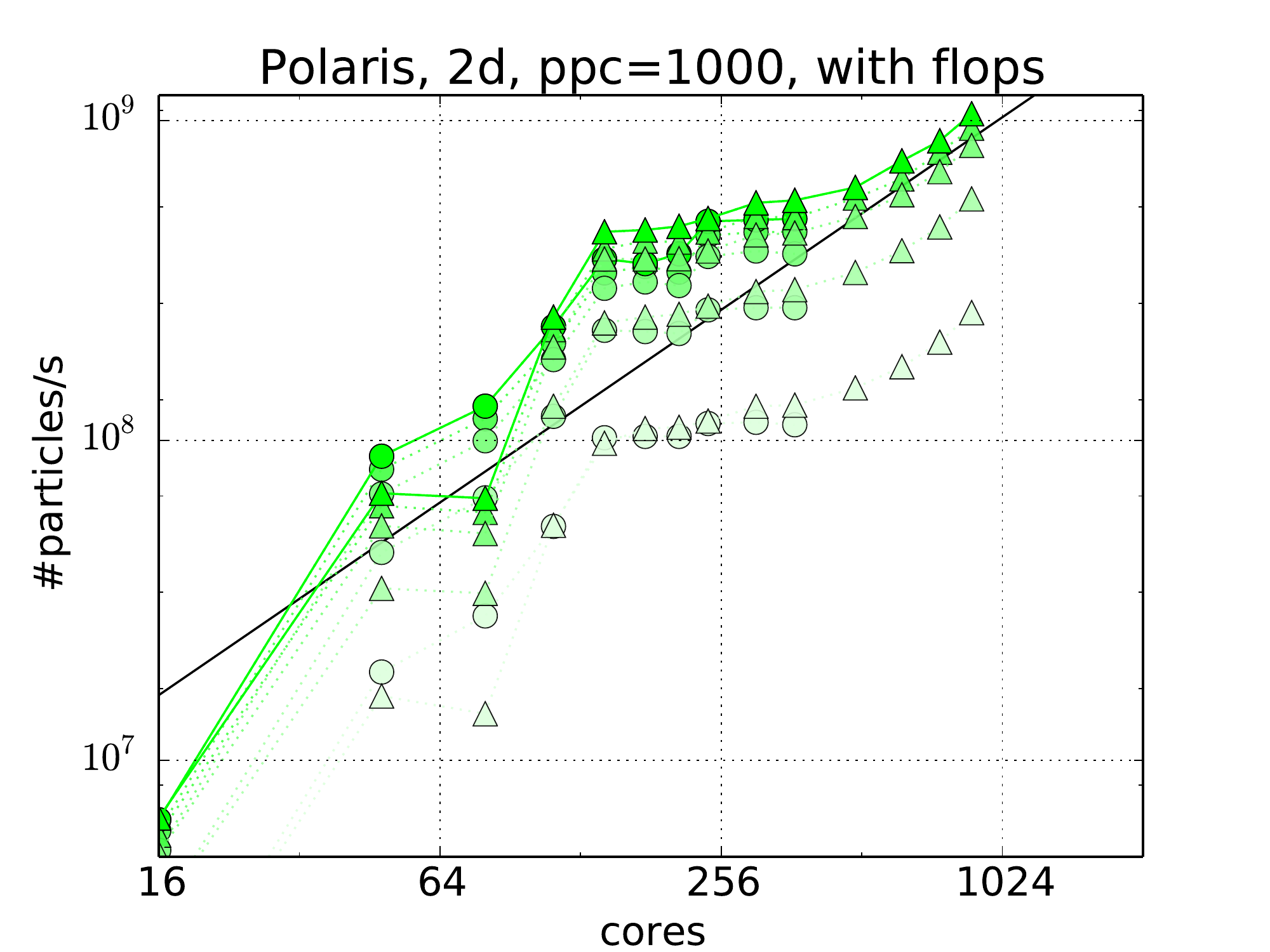}
  \\
  \caption{
    Parallel throughput of \pidt\ (left) and \rapidt\ (right) on Polaris.
  }
\label{figure:results:rapidt-small-core-counts}
\end{figure}

\begin{figure}[!ht]
\centering
    \includegraphics[width=0.49\linewidth]{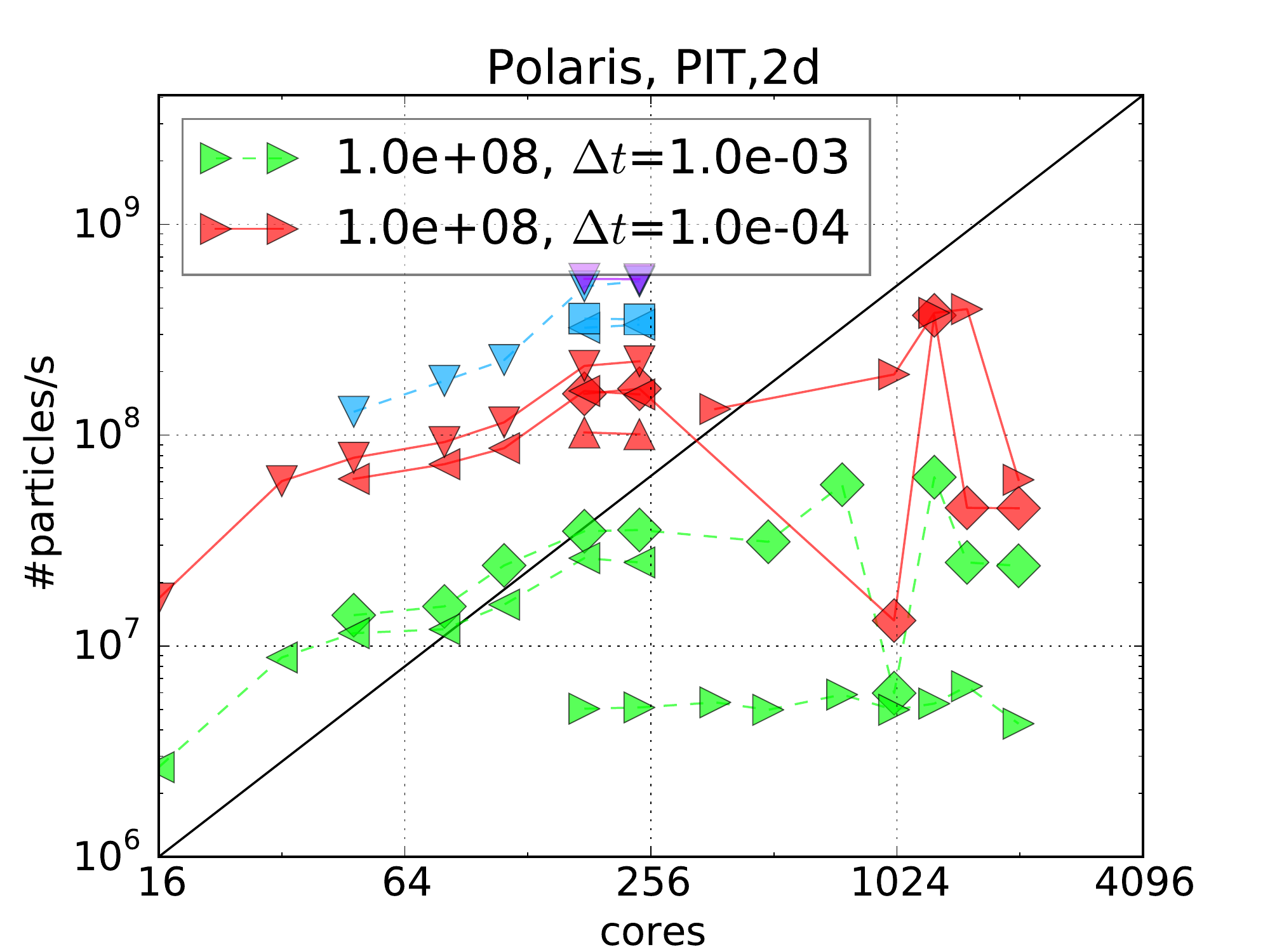}
    \includegraphics[width=0.49\linewidth]{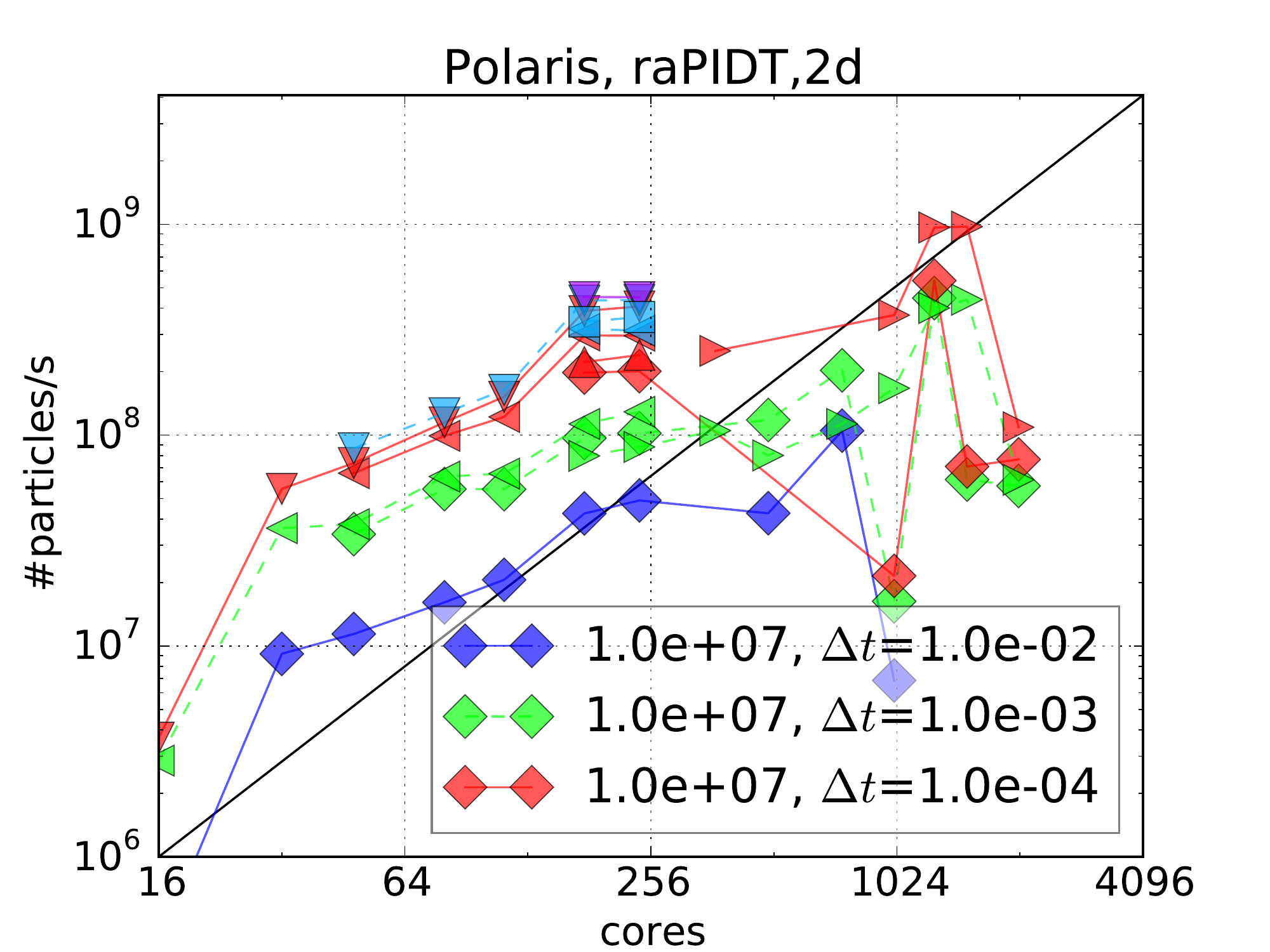}
   \\
    \includegraphics[width=0.49\linewidth]{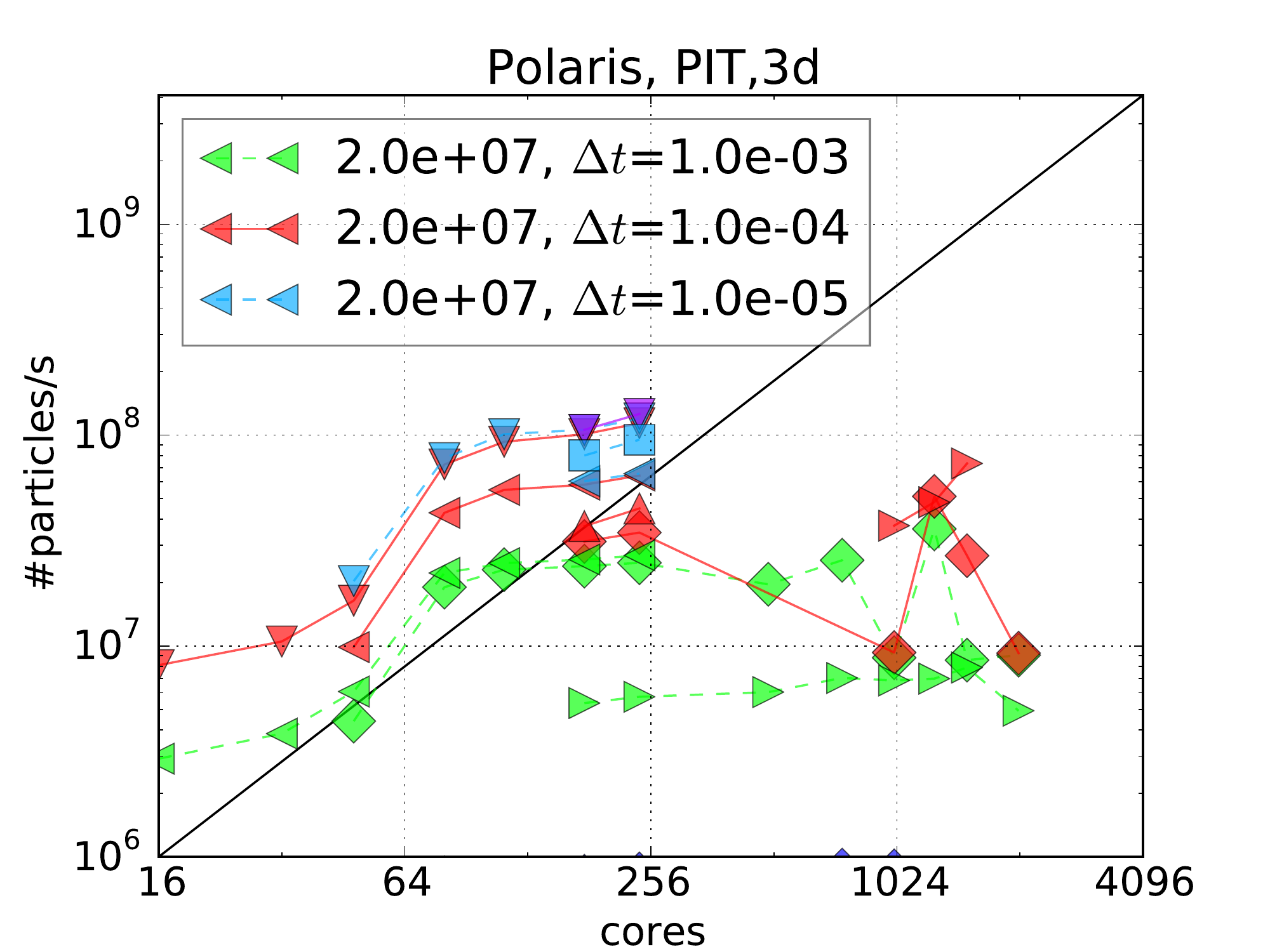}
    \includegraphics[width=0.49\linewidth]{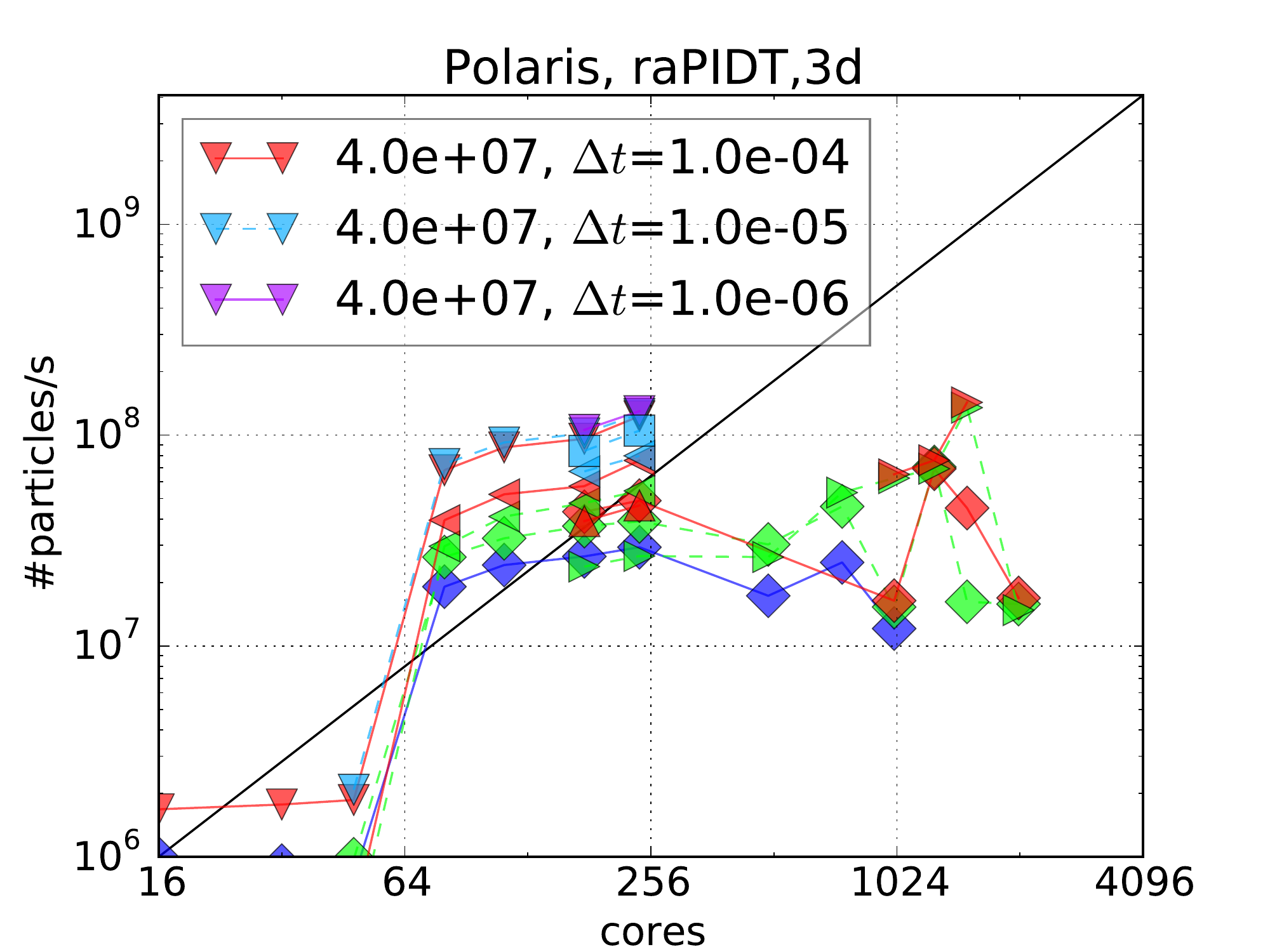}
   \\
    \caption{
    Scaling results for \texttt{ppc}=100 regarding different time step sizes
    for $d=2$ (top) and $d=3$ (bottom).
    \pit\ (left) is outperformed by \rapidt\ (right) while both approaches
    profit from decreasing time step sizes.
    Large core count measurements become dominated by network effects that
    eventually stop any scaling for both approaches on Polaris as opposed to
    Figure \ref{figure:results:rapidt-scaling-supermuc}.
  }
\label{figure:results:rapidt-scaling-polaris}
\end{figure}

\begin{figure}[!ht]
\centering
    \includegraphics[width=0.49\linewidth]{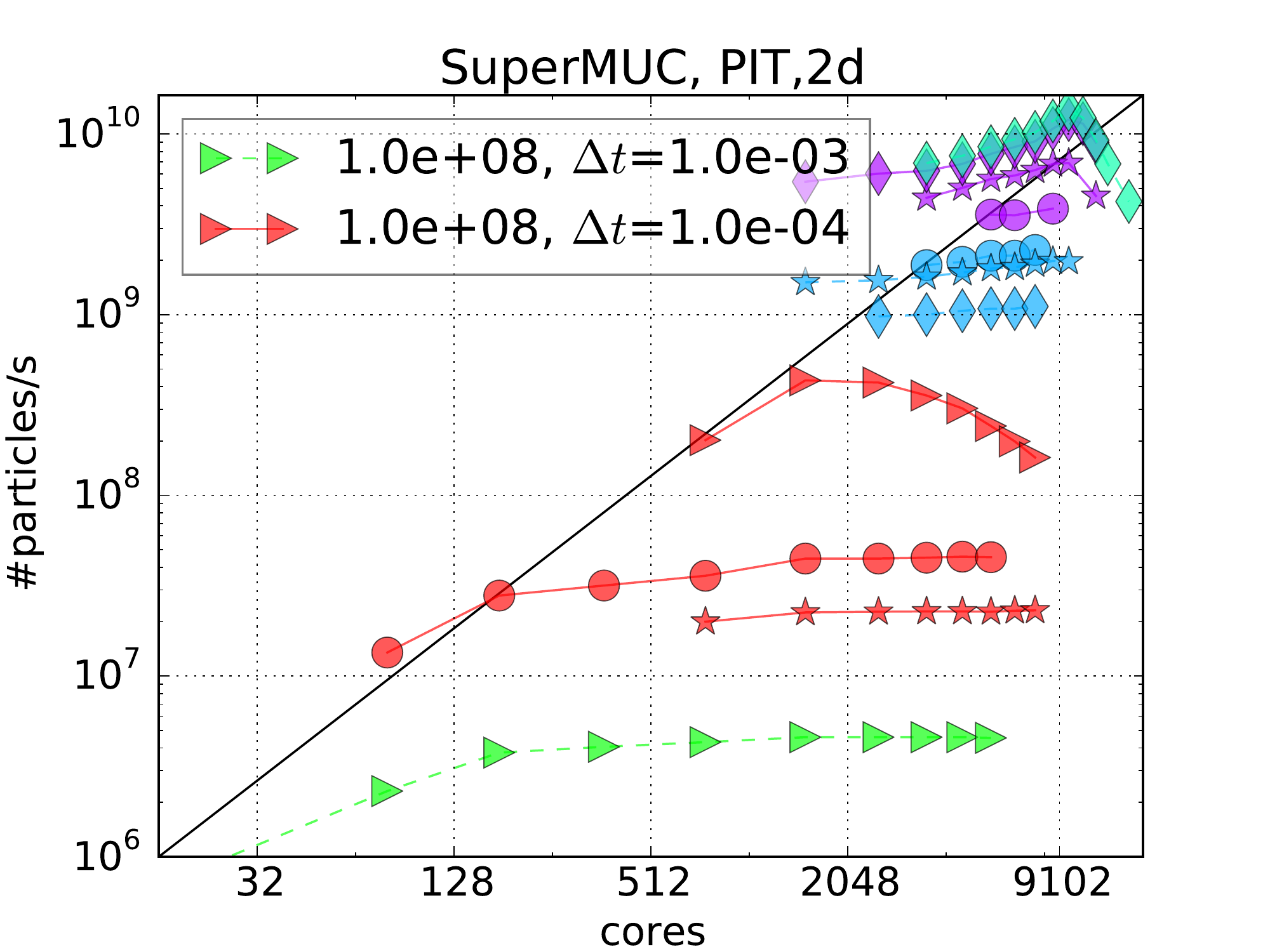}
    \includegraphics[width=0.49\linewidth]{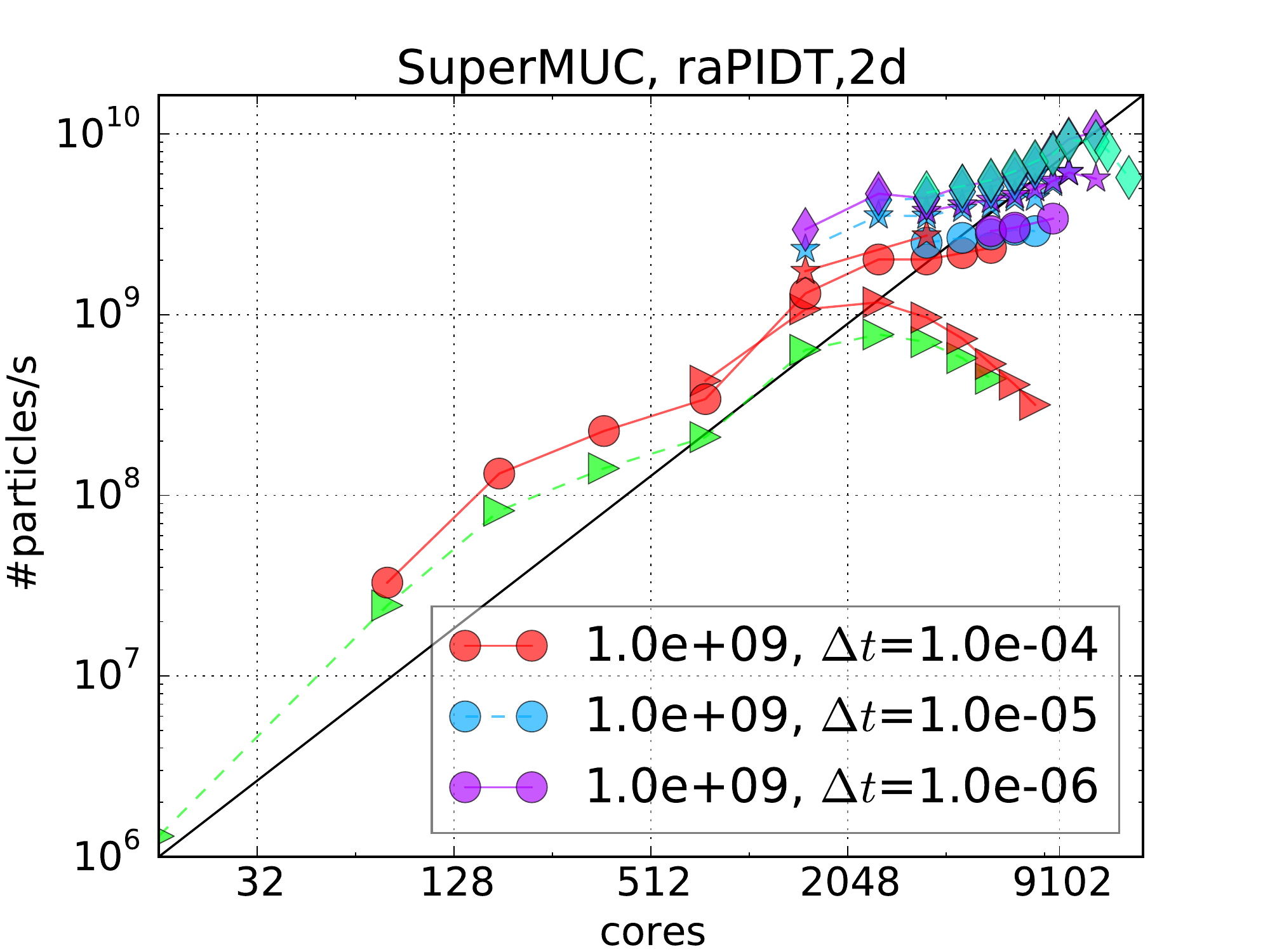}
   \\
    \includegraphics[width=0.49\linewidth]{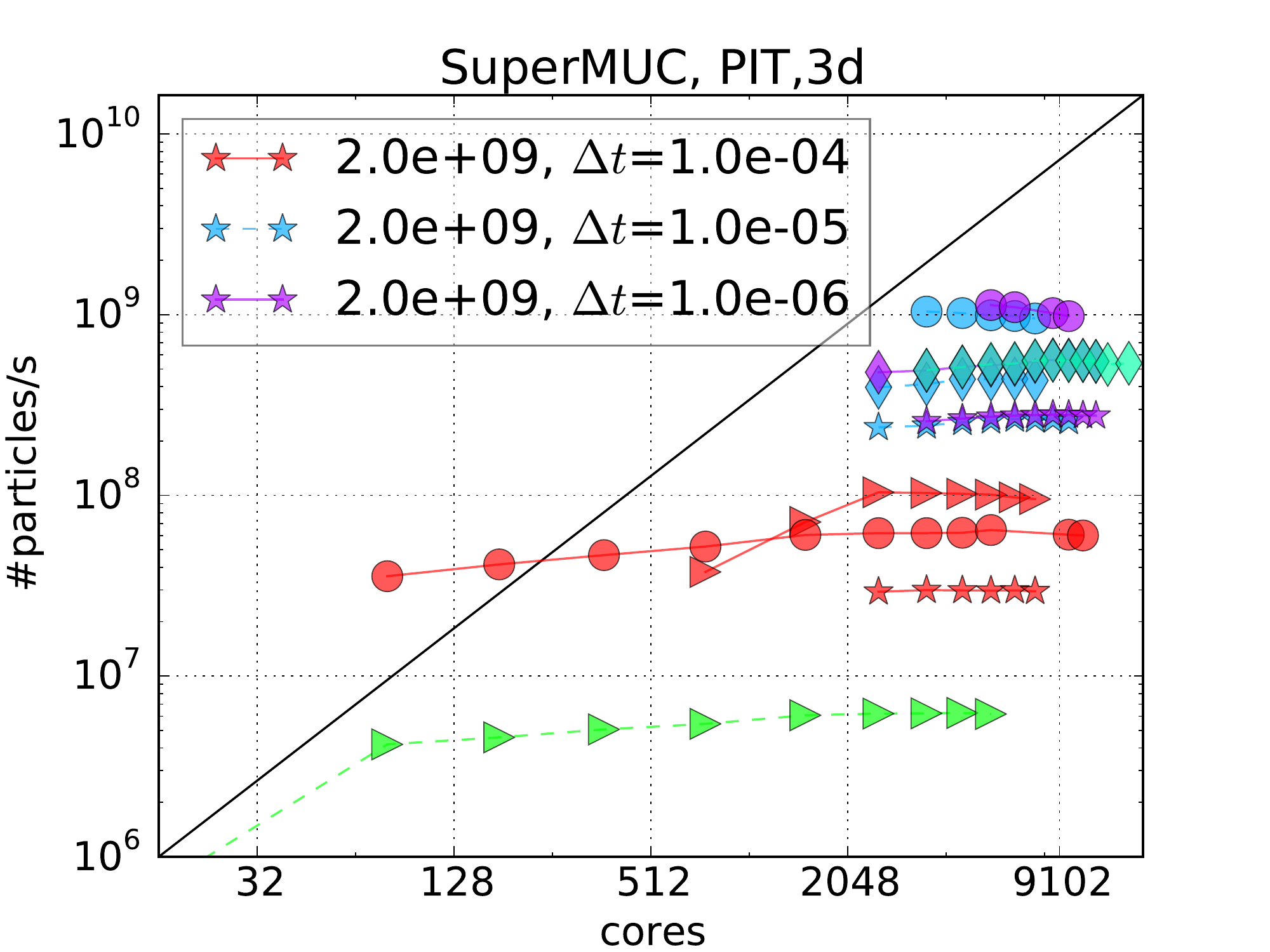}
    \includegraphics[width=0.49\linewidth]{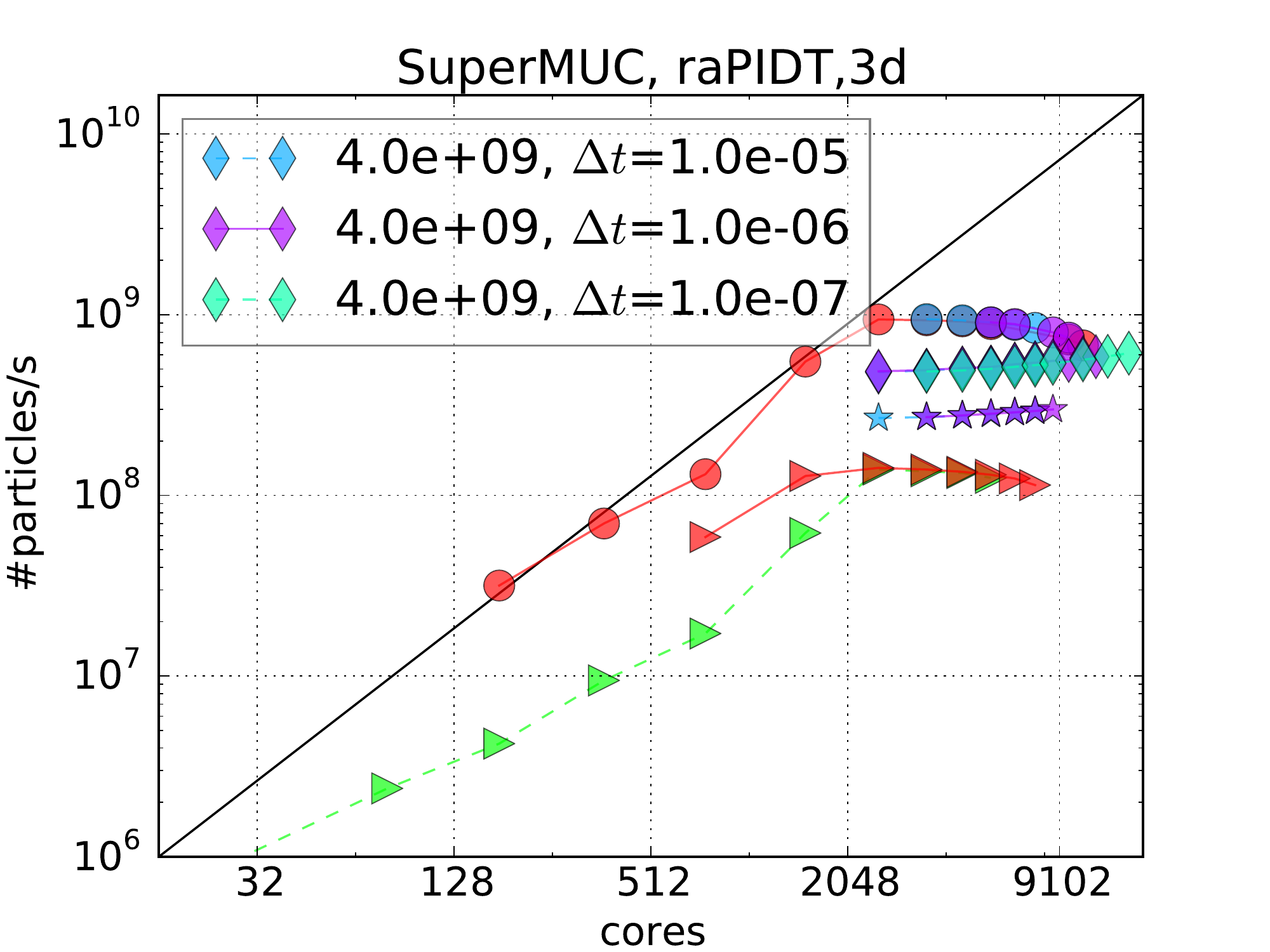}
    \caption{
    Experiments from Figure \ref{figure:results:rapidt-scaling-polaris} for
    different particle counts rerun on SuperMUC.
   }
\label{figure:results:rapidt-scaling-supermuc}
\end{figure}

\rapidt\ tackles the reduction challenge of our particle
management, i.e.~skips them if possible, and thus yields improved
throughput rates for the previously studied setups (Figure
\ref{figure:results:rapidt-small-core-counts}) where the time step size is kept
constant.
It makes the \pidt\ idea to yield throughput rates comparable to \pit.
Notably, it weakens the impact of the network topology.

An interpretation of \rapidt\ on a bigger core scale requires us to
emphasise how the scenario properties normalise experiment parameters.
Due to \texttt{ppc}, we control the particle density per spacetree cell,
while the total particle count determines the spacetree depth, i.e.~the cell
sizes.
Since we standardise the domain to the unit square or cube, prescribe the time
step size and have a uniform velocity distribution, the relative velocity of particles to cell sizes
scales with the total particle count by roughly $\sqrt[d]{\#particles}$.
An increase of the particle count induces an increase of the particles'
velocity.
While a study of weak scaling with respect to particles in
Section \ref{section:results:pit} and \ref{section:results:pidt} is reasonable
to understand algorithmic properties, our subsequent performance studies focus
on weak scaling where the overall computational domain is successively
increased.
To enable use to compare results with previous results, we 
stick with the unit length domain but decrease the time step size
when we increase the particle count.

Once we apply this normalisation, we observe that \rapidt\ outperforms \pit\
(Figures \ref{figure:results:rapidt-scaling-polaris} and
\ref{figure:results:rapidt-scaling-supermuc}) for reasonable big time step
sizes.
For small time step sizes (relative to the particle count, i.e.~the grid
structure), both perform with the same throughput.
\rapidt\ is more robust than \pit\ on Polaris (Figures
\ref{figure:results:rapidt-scaling-polaris}).
The $d=3$ throughput is lower than the two-dimensional counterpart for both
setups, as the typical grid structure for fixed \texttt{ppc} differs.
If we normalised with respect to the grid depth, $d=2$ and $d=3$ yield similar
results.
On SuperMUC, the 4:1 blocking topology for more than 8192 cores stops both
variants to pass the $10^{10}$ particles per second threshold and
makes the throughput stagnate or degenerate.
\rapidt\ here does not perform significantly more robust than \pit.

\section{Comparison with other algorithms}
\label{section:comparison}

\noindent
We finally compare \pit, \pidt\ and \rapidt\ to alternative implementations.
Obviously, only alternatives that support both tunneling and
dynamically adaptive grids candidate.  
The most prominent class of spacetree AMR codes is the family of spacetree
algorithms that rely on space-filling curves for the partitioning
\cite{Bader:13:SFCs,Hamada:09:GordonBell,Lashuk:12:ParallelFMMOnHetergeneousArchitectures,Rahimian:2010:GordonBell,Sampath:08:Dendro}.
Textbook variants of SFC codes typically hold information about their
subpartitioning on each rank---they basically store where the SFC's preimage is
cut into pieces.
Each rank's tree construction starts from its locally owned cells, i.e.~the
cells in-between the rank's start index and its end index along the SFC.
From hereon, the rank constructs the local spacetree in a bottom-up
manner:
if a cell is held by a rank, also its parent is held by this rank which yields, on
coarser levels, a partial replication.
Detailed comparisons to the present Peano approach used in our \pit\ and
\rapidt\ implementation can be found in \cite{Weinzierl:15:Peano}.

If the cuts along the SFC are known, it is possible due to the SFC code
\cite{Weinzierl:15:Peano} to compute per particle the preimage of their
position, i.e.~an index along the SFC.
This immediately yields the information which rank holds a particle.
One thus can send any particle directly to the right rank after the move.
Though our implementation base relies on a space-filling curve as well 
\cite{Schreiber:13:MetaData,Weinzierl:2009:Diss,Software:Peano,Weinzierl:11:Peano},
it does neither exhibit curve information to the application nor does
it hold information about the global partitioning on any rank.
However, we can manually expose this information to the ranks and thus
reconstruct a SFC-based implementation.
This allows us to compare an SFC-based code directly to the present
alternatives:
Our implementation using SFC cuts runs through the grid and moves the particles
as for \pit.
Whenever a particle leaves the local domain, we move it into a buffer.
For each rank, there's one buffer.
Upon termination of the local traversal, all buffers are sent away. 
As there is no information available whether particles tunnel, each rank then
has to wait for a notification from all other ranks.
Though this is a global operation, our implementation follows
\cite{Sundar:13:HykSort} and realises it with point-to-point exchange. 
If one rank finishes prior to other ranks, the in-between time is then already
used for data exchange.
Throughout our experiments, any overhead to maintain the global decomposition
data is neglected, i.e.~we restrict purely to the particle sorting.
However, we emphasise that a rank receiving particles does not hold any
additional information about the particle position, i.e.~where within the local
domain it has to be inserted. 
We thus rely on the drop mechanism.
Comparisons to a direct insert based upon SFCs again (the particle's cell is
identified due to the preimage and inserted right into the correct cell) show
that both variants yield comparable results.
This is an important difference to \pit\ and \pidt\ as the latter schemes encode
remote sorting information implicitly within their send order.

\begin{figure}[!ht]
\centering

\includegraphics[width=0.44\linewidth]{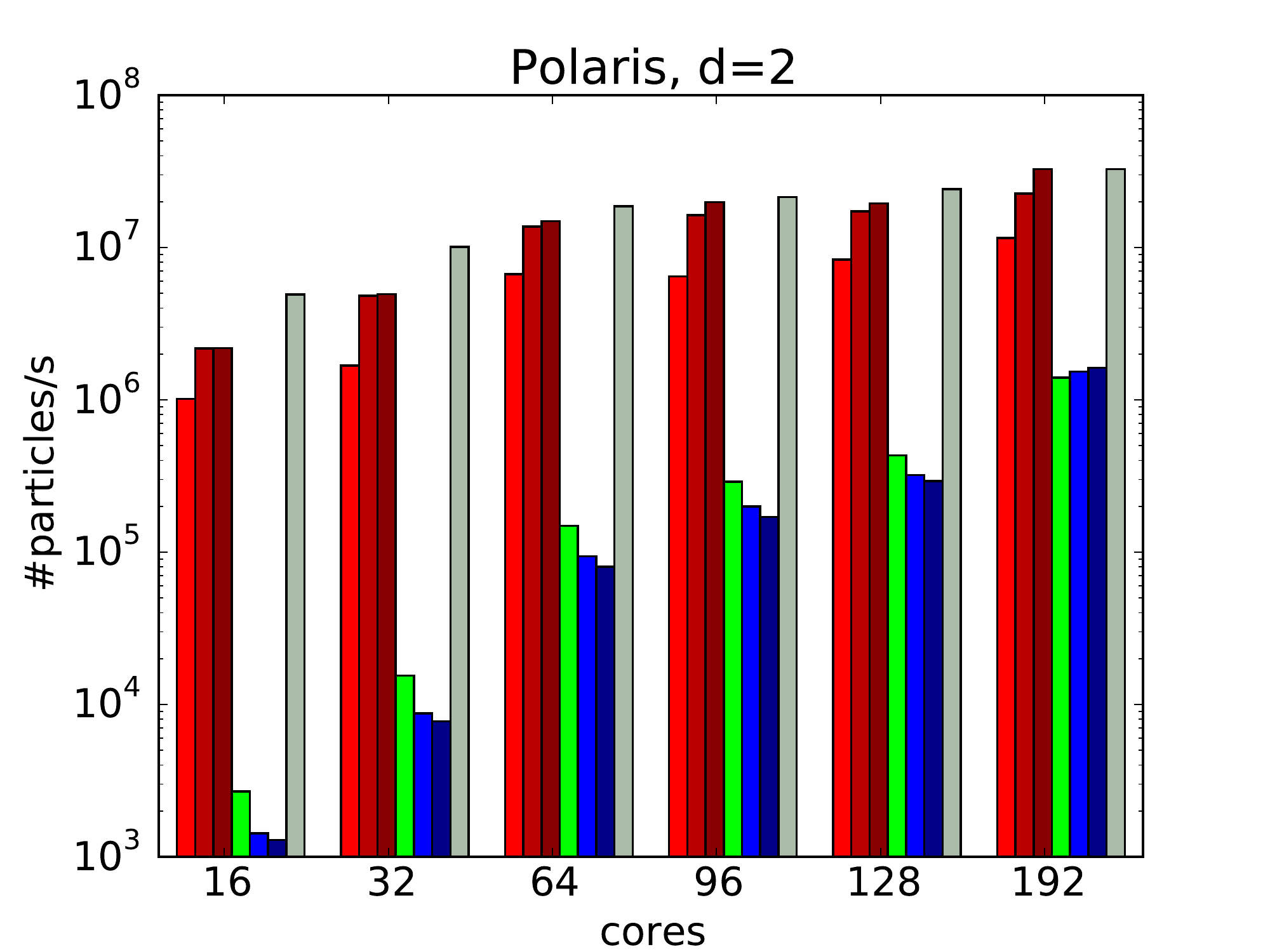}
\includegraphics[width=0.44\linewidth]{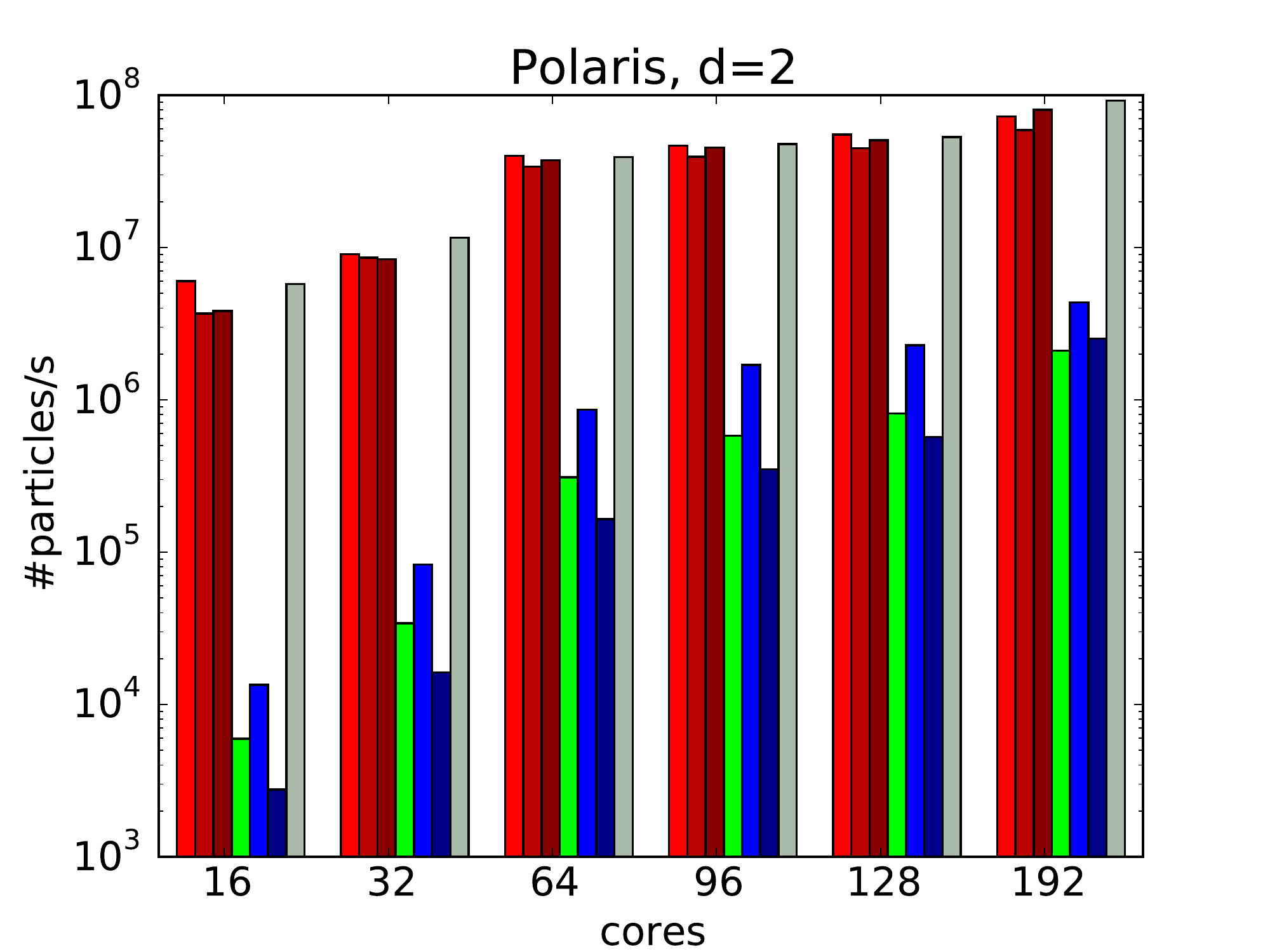}
\includegraphics[width=0.44\linewidth]{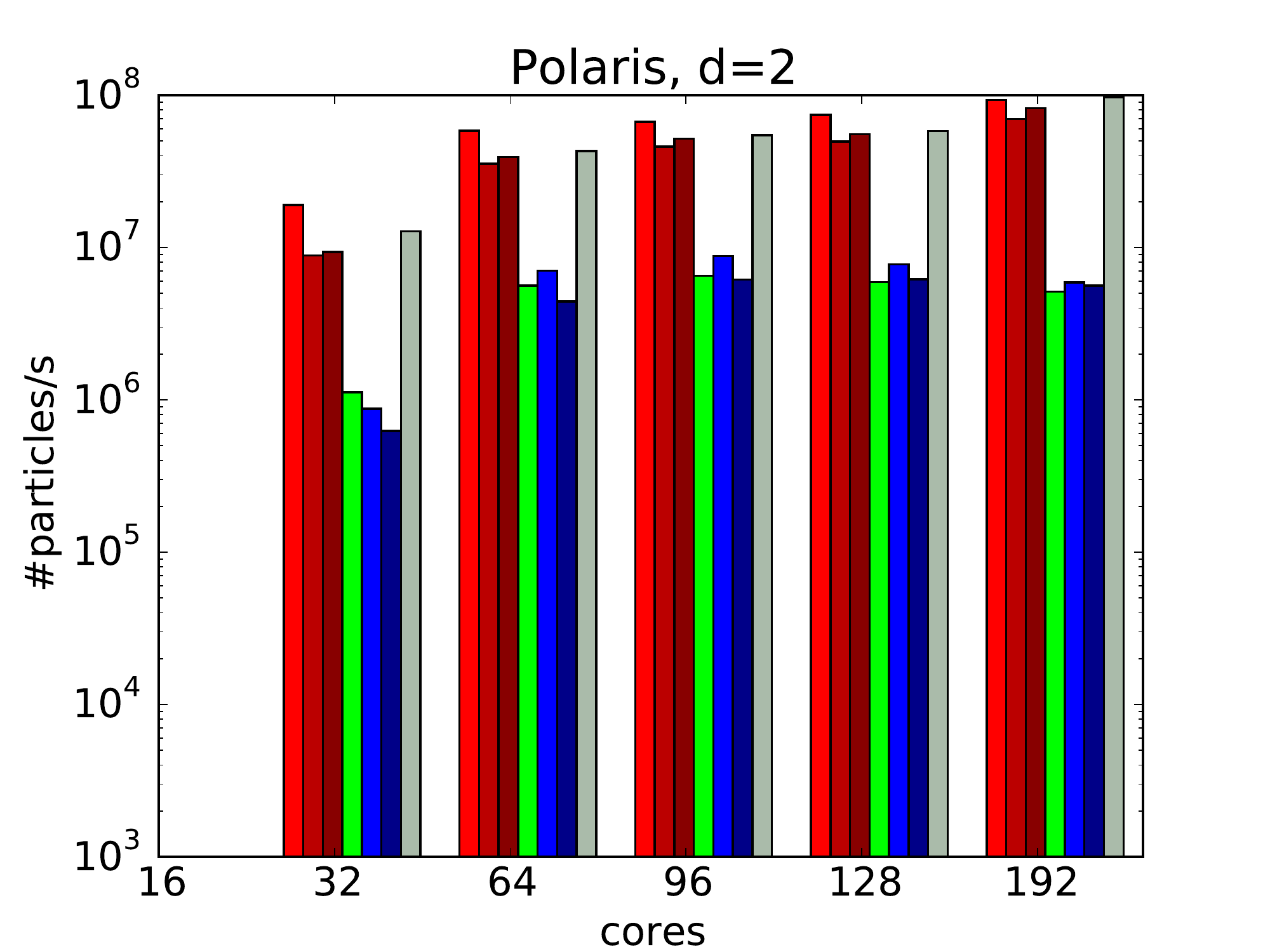}
\includegraphics[width=0.44\linewidth]{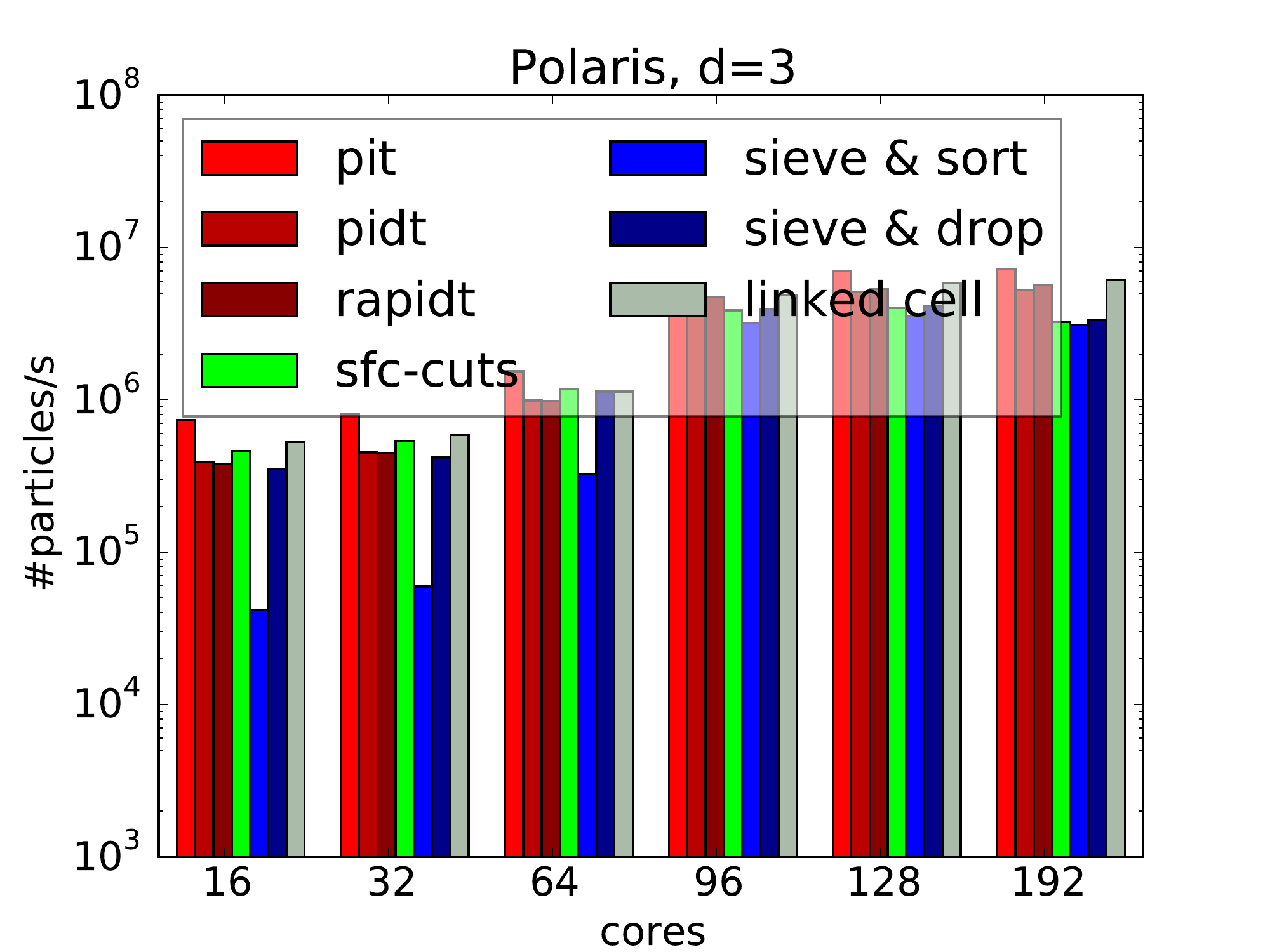}

\caption{
  Comparison of different solvers for small core counts and characteristic
  setups with \texttt{ppc}=100, $\Delta t=10^{-2}$ (top left),
  \texttt{ppc}=100, $\Delta t=10^{-3}$ (top right),
  \texttt{ppc}=100, $\Delta t=10^{-4}$ (bottom left).
  \texttt{ppc}=10, $\Delta t=10^{-4}$ for $d=3$ (bottom right) illustrates a
  regime where all algorithms yield comparable results because of a low
  particle density.
  All experiments are ran with $10^6$ particles. 
} 
\label{figure:comparison}
\end{figure}

A second competing idea is the sieve approach from
\cite{Dubey:11:ParticlesInMesh}.
Different to SFCs, it comes along without global domain decomposition knowledge.
Each rank collects all particles that leave its domain. 
Again, we rely on local buffers.
These sets of particles then are passed around between the ranks cyclically.
Each rank extracts the particles falling into its domain.
It sieves and then passes on the array.
We distinguish two variants to sort particles into the local grid if we receive
a set of particles from the neighbour node:
in one variant (sieve \& drop), we rely on \pit's drop mechanism.
In a second variant, we use the local SFC information to sort the received
particles directly into the correct cells.

Besides the alternative variants, we also compare \pit, \pidt\ and \rapidt\ to a
plain linked-cell algorithm.
For the latter, we use a modified \pidt\ code.
The velocity profile is chosen such that no particle may tunnel.
As a consequence, we manually switch off all multiscale checks, and we, in
particular, skip any master-worker or worker-master communication.
Such a test is thus a valid setup to quantify the overhead introduced by the tunneling for the present software independent of the quality of
implementation.
No software ingredient is particular tuned and all codes rely on the same
ingredients where possible.
All setups rely on exactly the same domain decomposition, i.e.~balancing
differences are eliminated.

We reiterate that \pidt\ is faster than \pit\ for big time step sizes of $\Delta
t=10^{-2}$ (Figure \ref{figure:comparison}). 
However, each run with the linked-cell approach remains faster though the 
performance gap closes with increasing core counts. 
Neither our SFC-cut nor our sieve implementations can cope with these results.
However, they always scale better.
For sufficiently big $\Delta t$, they eventually close the
performance gap.
If we reduce the time step size, there is a sweet spot where \pit\ becomes
faster than \pidt\ as well as the linked-cell approach (visible here only for $d=3$).
Still, the sieve algorithm yields significantly lower throughput.
The SFC-based variants are slower, too.

The differences between \pit\ and \pidt\ have been discussed already.
Our experiments validate, for most setups, statements from
\cite{Dubey:11:ParticlesInMesh} that observe that all variants supporting
tunneling are slower than a pure linked-cell like code.
However, both \pit\ and \pidt\ reduce this runtime difference significantly.
This is insofar interesting, as \pit\ exhibits similarities to the
\texttt{up\_down\_tree} algorithm from \cite{Dubey:11:ParticlesInMesh} which
is reported to be even slower than sieve.
Sieve suffers from a global communication phase after each
iteration where the length of the phase (number of data exchange steps)
increases with more ranks participating. 
It thus cannot scale.
\pit\ circumnavigates such a global data exchange phase.
While the SFC-cut implementation scales and allows all the ranks to process
their local grid asynchronously, it still runs into a synchronisation phase, as
each rank has to wait per time step for every other rank for notification.
This point-to-point synchronisation determines the performance, since we found
no significant runtime difference when we studied whether a drop mechanism for
local sorting or the direct application of SFC indices yields better results for SFCs.
\pit\ and \pidt\ are able to spread all data transfer more evenly
accross the executation time than their competitors.
It is in particular interesting that \pit\ manages to outperform the linked-cell
approach for very small time steps and $d=3$ with small \texttt{ppc}, i.e.~a
deep spacetree.
This results from the fact that particles within a cell are not sorted and few
particles cross the cell faces per time step.
Obviously this is an unfair competition: 
as no neighbourhood information is available in \pit\ (different to \pidt\
where we have half a cell overlap), no particle-particle interaction can be
realised.
\pit\ is designed for particle-mesh interaction only.
Our SFC implementation was able to compete with \pit\ and \pidt\ only for small
core counts and deep spacetree hierarchies.
For bigger core counts, the SFC's behaviour resembles the sieve algorithm which
results from the global communication phase found there, too.

We emphasise that the whole runtime picture might change
if the grid were regular, if the grid were not fixed prior to the particle
sorting, i.e.~if we could create the grid depending on the particles without persistent grid data, or if
we translated our reduction-avoiding mechanism from \rapidt\ into
the SFC world.
The latter would yield a higher asynchronity level and remove a global SFC
communication phase.

\section{Conclusion and outlook}
\label{section:conclusion_and_outlook}

We introduce three particle management variants facilitating
solvers that have to resort particles into an existing
dynamically adaptive Cartesian grid frequently and can run into tunneling.
Particle in tree (\pit) holds the particles within the tree, i.e.~the cells of a
multiscale adaptive Cartesian grid, and lifts and drops the particles between the levels to
move particles in-between cells.
Particle in dual tree (\pidt) stores the particles within the vertices,
i.e.~switches to a dual grid, and thus fuses ideas of a tree-based particle management with a multiscale
linked-list paradigm.
Particles either can move up or down within the (dual) tree or directly into
neighbouring cells.
The latter works on any level of the multiscale grid.
\rapidt\ finally augments \pidt\ by a simple velocity analysis. 
This analysis labels regions where no particle is moved in-between
certain grid levels in the subsequent time steps.
This allows us to eliminate reductions within the tree locally. 
 
Comparisons of the straightforward \pit\ with the dual grid strategy \pidt\
reveal that no scheme is superior to the other schemes by default.
For small problem setups with slowly moving particles, \pit\ yields higher
performance than \pidt.
For big problem setups or fast particles, there is a sweet spot where the
higher code complexity of \pidt\ pays off.
Particular interesting is the fact that \rapidt\ picks up advantages of linked
cell strategies, i.e.~asynchronity of the ranks and exchange of data in the
background, while arbitrary tunneling still is enabled.
This might make it a promising candidate for the upcoming massively
parallel age.
It is also not surprising that the particle throughput depends on 
differences in the hardware characteristics.
High clock rates pay off for moderate particle counts.
For high core counts, topology properties of the interconnection network outshine
clocking considerations.
Regarding latency effects, \rapidt\ is more robust than
\pit\ and \pidt.
For big core counts and high particle numbers, all variants however
continue to suffer from latency and bandwidth constraints introduced by x:1
blocking.

Our particle maintenance strategies are of relevance for multiple
particle-in-cell (\pic) simulations as well as other particle-grid codes with
different algorithmic properties.
Examples for \pic\ are plasma processes with shocks or reconnection, 
global large-scale 
simulations with
multifaceted particle velocity characteristics, or self-gravitating systems whose dynamics imply
large particle inhomogeneities.
While such setups require substantial investment into the PDE solver, our
particle treatment transfers directly.
Another example is local particle time stepping where some particles march in
time fast and thus might tunnel, while others then follow up with tiny time
steps where the reduction skips come into play. 
It is part of our future work to use the present algorithms in such a context.
Of particular interest to us is the fusion with applications that realise
implicit schemes or particle-particle interactions.
While the former weaken the grid size and, hence, time step 
constraints---though we still can make the grid size adapt to fine-scale
inhomogeneities of the PDE without algorithmic constraints resulting from the
particle velocities---both increase the arithmetic
intensity and thus damp the impact of communication on the scaling.
We further recognise that the dual grid approach allows us to realise any
interaction with a cut-off radius smaller than half the grid width without any additional
helper data structure such as linked-cell lists or modified linked lists
\cite{Gonnet:07:VerletLists,Mattson:99:ModifiedLinkedList}.
In this sense, our approach is related to dual tree traversals
\cite{Dehnen:02:DualTreeTraversal,Warren:95:DualTree} that also avoid to build
up connectivity links.

It seems to be important to stress two facts.
On the one hand, our lift and drop mechanisms can also be used by
particle-interaction kernels to push particles to the ``right'' resolution
level, i.e.~a level fitting to their cut-off radius.
For complex applications, there is no need to drop particles always into the
finest grid all the time. 
This enables particles suspended within the tree.
It is subject of future work
with respect to applications with particle-particle interaction to exploit such
a multilevel storage.
On the other hand, we appreciate how fast linked-list algorithms perform in
many applications.
In our formalism, they basically induce an overlapping domain topology on top
of a given tree distribution.
And we recognise that the fastest variants of these algorithms rely on
incremental updates of these interaction lists.
It is hence a straightforward idea 
to combine our
grammar paradigm plus the particle management with an update of these lists
where the grammar eliminates also list updates. 

\section*{Acknowledgements}

\noindent
The authors gratefully acknowledge the Gauss Centre for Supercomputing e.V.
(www.gauss-centre.eu) for funding this project by providing computing time on
the GCS Supercomputer SuperMUC at Leibniz Supercomputing Centre (LRZ,
www.lrz.de), as well for the support of the LRZ. 
This work also made use of the facilities of N8 HPC provided and funded by the
N8 consortium and \linebreak EPSRC (Grant No.~N8HPC\_DUR\_TW\_\-PEANO). 
The Centre is co-ordinated by the Universities of Leeds and Manchester.
Special thanks are due to Bram Reps and Kristof Unterweger for their support on
linear algebra and implementation details.
Kristof also contributed to make the numerical experiments run.
Robert Glas and Stefan Wallner made a major contribution to the maturity and
performance of the presented implementations as they used it for their
astrophysics code and thus pushed the development.
Bart Verleye was funded by Intel and by the Institute
for the Promotion of Innovation through Science and Technology in Flanders
(IWT).
All underlying software is open source and available at
\cite{Software:Peano}.

\bibliographystyle{plain}
\bibliography{./paper}

\appendix
\section{Validation}
\label{section:validation}

\begin{figure}[!ht]
\centering
\includegraphics [width=0.8\linewidth]{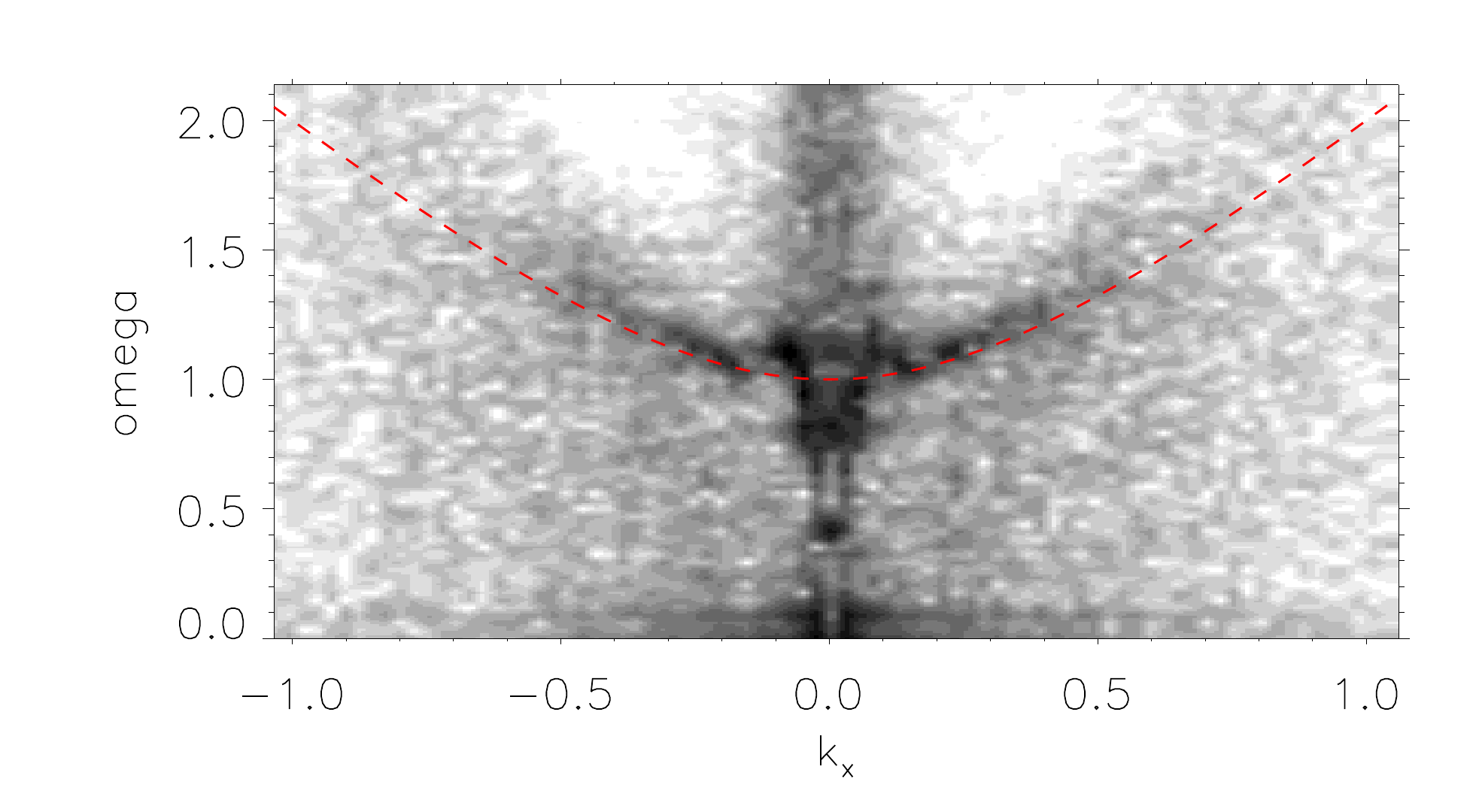}
  \vspace{-0.5cm}
\caption{Electric field spectrum in the wavevector-frequency plane
$(k_1,\omega)$, at $k_2=0$. The colors range from white (low signal) to black
(high signal). The red dashed line shows the theoretical Langmuir wave dispersion relation obtained from the Vlasov-Poisson theory. }
\label{fig:RelDisp2D_1} 
\end{figure}

As validation of our \pic\ implementation---which is independent of the choice
of \pit\ or \pidt---we study a
thermal plasma at rest, and we confirm that our code retrieves the theoretical
Langmuir wave dispersion relation obtained from the Vlasov-Poisson theory.
For our test, we rely on a regular $d=2$ grid with $243 \times 243$ grid cells.
Time is normalized to the inverse plasma frequency $\omega_{p}^{-1}$, and space
is normalized to the Debye length $\lambda_{d}$, i.e.\ a grid cell has size
$\lambda_{d} \times \lambda_{d}$.
The (angular) plasma frequency is defined by
\[
\omega_{p} = \sqrt{ \frac{n e^2}{\epsilon_{0} m_{e}} }
\]
with $n$ being the electron density, $m_{e}$ being  the electron mass and $e$
being  the elementary charge.
The Debye length $\lambda_d$, the typical length scale for the electric screening of a charge in a plasma, is defined by
\[
\lambda_d = \sqrt{ \frac{\epsilon_{0} k_{_{B}} T}{n e^2} } = v_{th} /
\omega_{p}
\]
with $\epsilon_{0}$ being the vacuum permittivity, $k_{_{B}}$ being the
Boltzmann constant, and $T$ being the electron temperature.
All parameters follow \cite{Krall:73:PrinciplesOfPlasmaPhysics}.
Our observed Langmuir waves, also called electron plasma waves, are high
frequency oscillations of the electron density over a fixed ion background. The frequency is high enough for the ions not to have time to respond because of their higher inertia. The charge separation generates an electric force acting as the restoring force.

In the cold plasma approximation, i.e.~when the thermal velocity of the plasma is much smaller than the phase velocity of the wave, the Langmuir waves oscillate at the plasma frequency.
When taking into account the finite temperature of the electrons, a (small) frequency correction appears, leading to the following dispersion relation for a Maxwellian velocity distribution function
\[
\omega^2 = \omega_{p}^2 + 3 k^2 v_{th}^2 = \omega_{p}^2 ( 1 + 3 k^2 \lambda_d^2 )
\]
where $k_{_{L}}$ is the wave vector, $v_{th}$ the electron thermal velocity.

The electron plasma is initially loaded with $100$ macro-particles per cell, a
charge-to-mass ratio $q/m = -1$, and a random velocity characterized by a
Gaussian distribution of mean velocity $0$ and thermal velocity $v_{th}=1$. The
electron charge density is initially uniform $n=-1$.
We also set a neutralising fixed
ion background---this contribution adds as constant term to the
right-hand side in (\ref{equation:pic:pde}) besides the simulated
particles---with a constant charge density $n_i = 1$.
Our simulation runs with a time step size $\Delta t=0.01$ and the
potential is output every 150 time steps. 

\begin{figure}[!ht]
\centering
\includegraphics [width=0.8\linewidth]{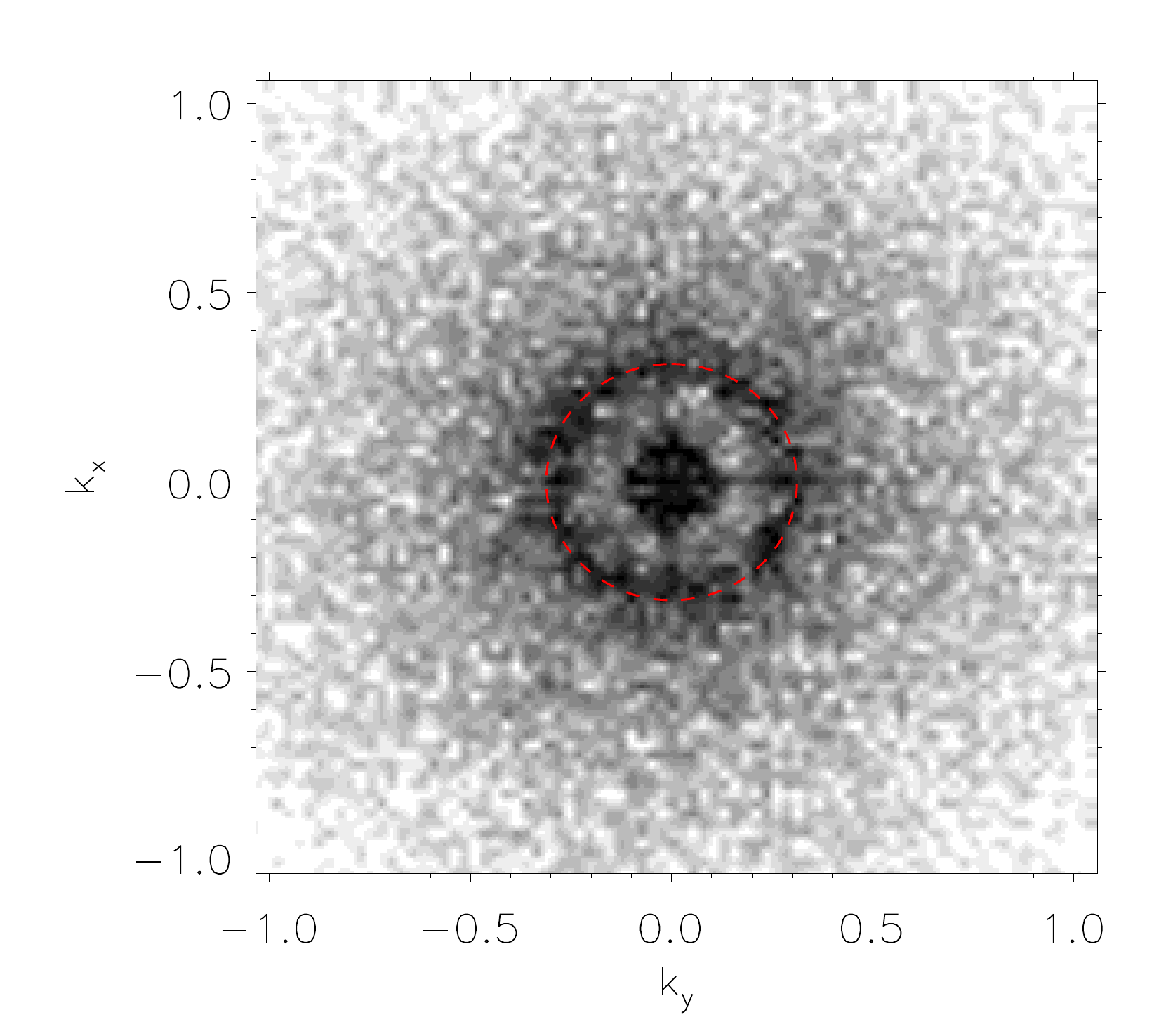}
  \vspace{-0.5cm}
\caption{Electric field spectrum in the wavevector-wavevector plane $(k_1,k_2)$,
at $\omega=1.14$. The colors range from white (low signal) to black (high
signal). The red dashed line shows the wavevectors in the $d=2$ plane
corresponding to Langmuir waves at frequency $\omega=1.14$, from the theoretical Langmuir wave dispersion relation obtained from the Vlasov-Poisson theory. }
\label{fig:RelDisp2D_2}
\end{figure}

The initial random velocity seeds a sea of Langmuir waves whose
dispersion relation can be checked from the analysis of the electric potential:
the output electric potential $V(x,y,t)$ is
Fourier transformed both in time and space which yields $\hat
V(k_1,k_2,\omega)$ where $k_1$ and $k_2$ are the wave vectors associated to the
space coordinates $x_i \ i \in \{1,2\}$.
$\omega$ is the frequency. The Fourier transform of the electric field is
shown for a cut at $k_y=0$ in the plane $(k_x,\omega)$ in
Figure \ref{fig:RelDisp2D_1}. The electric fluctuations are organized in
Fourier space so to match the theoretical Langmuir wave dispersion relation obtained
from the Vlasov-Poisson theory $\omega^2 = \omega_{p}^2 \ \big( 1 + 3\ (k
\lambda_d)^2 \big)$ (dashed line). The result is identical in a cut at
$k_x=0$ in the plane $(k_y,\omega)$. 
The Fourier transform of the electric field is also shown in a
two-dimensional cut at fixed frequency $\omega=1.14$, in the wavevectors plane
$(k_1,k_2)$, in Figure \ref{fig:RelDisp2D_2}. 
Again, the electric fluctuations are organized
in Fourier space at the wave vectors corresponding to Langmuir waves (dashed
line) at the prescribed frequency. 

\section{Additional experimental data}
\label{section:additional-experiments}

\begin{table}[!ht]
  \caption{
    Experiments from Table \ref{table:memory-throughput-supermuc} reran on
    Polaris.
  }
  \label{table:memory-throughput-polaris}
  \begin{center}
  { 
    \tiny
    \begin{tabular}{c|cccccc|c}
 p & 1 & 2 & 4 & 8 & 12 & 16 & 32 \\
 \hline 
 $10^3$ & $4.59\cdot 10^7$ & $4.86\cdot 10^7$ & $\mathbf{5.26\cdot 10^7}$ & $3.73\cdot 10^7$ & $2.94\cdot 10^7$ & $3.03\cdot 10^7$ & $1.92\cdot 10^7$ \\ 
 $10^4$ & $1.04\cdot 10^8$ & $1.65\cdot 10^8$ & $2.24\cdot 10^8$ & $\mathbf{2.36\cdot 10^8}$ & $2.24\cdot 10^8$ & $1.98\cdot 10^8$ & $1.32\cdot 10^8$ \\ 
 $10^5$ & $1.17\cdot 10^8$ & $2.31\cdot 10^8$ & $4.18\cdot 10^8$ & $6.60\cdot 10^8$ & $8.08\cdot 10^8$ & $\mathbf{8.83\cdot 10^8}$ & $7.04\cdot 10^8$ \\ 
 $10^6$ & $1.14\cdot 10^8$ & $2.35\cdot 10^8$ & $4.57\cdot 10^8$ & $8.52\cdot 10^8$ & $1.20\cdot 10^9$ & $\mathbf{1.38\cdot 10^9}$ & $1.33\cdot 10^9$ \\ 
 $10^7$ & $1.11\cdot 10^8$ & $2.20\cdot 10^8$ & $4.26\cdot 10^8$ & $5.99\cdot 10^8$ & $6.24\cdot 10^8$ & $\mathbf{6.29\cdot 10^8}$ & $5.93\cdot 10^8$ \\ 
 $10^8$ & $1.05\cdot 10^8$ & $2.20\cdot 10^8$ & $4.30\cdot 10^8$ & $6.02\cdot 10^8$ & $6.28\cdot 10^8$ & $\mathbf{6.31\cdot 10^8}$ & $6.26\cdot 10^8$ \\ 
 \hline 
 $10^3$ & $\mathbf{5.92\cdot 10^7}$ & $4.52\cdot 10^7$ & $3.91\cdot 10^7$ & $3.26\cdot 10^7$ & $3.01\cdot 10^7$ & $3.40\cdot 10^7$ & $1.90\cdot 10^7$ \\ 
 $10^4$ & $7.85\cdot 10^7$ & $1.28\cdot 10^8$ & $1.83\cdot 10^8$ & $\mathbf{2.13\cdot 10^8}$ & $1.94\cdot 10^8$ & $2.00\cdot 10^8$ & $1.18\cdot 10^8$ \\ 
 $10^5$ & $9.77\cdot 10^7$ & $1.90\cdot 10^8$ & $3.40\cdot 10^8$ & $5.64\cdot 10^8$ & $6.82\cdot 10^8$ & $\mathbf{7.20\cdot 10^8}$ & $5.42\cdot 10^8$ \\ 
 $10^6$ & $8.94\cdot 10^7$ & $1.71\cdot 10^8$ & $3.36\cdot 10^8$ & $5.53\cdot 10^8$ & $6.44\cdot 10^8$ & $\mathbf{6.71\cdot 10^8}$ & $6.13\cdot 10^8$ \\ 
 $10^7$ & $7.75\cdot 10^7$ & $1.67\cdot 10^8$ & $3.20\cdot 10^8$ & $4.05\cdot 10^8$ & $4.17\cdot 10^8$ & $\mathbf{4.20\cdot 10^8}$ & $4.07\cdot 10^8$ \\ 
 $10^8$ & $7.73\cdot 10^7$ & $1.72\cdot 10^8$ & $3.22\cdot 10^8$ & $4.07\cdot 10^8$ & $4.19\cdot 10^8$ & $\mathbf{4.24\cdot 10^8}$ & $4.21\cdot 10^8$ \\ 
\end{tabular}
  }
  \end{center}
\end{table}

We ran the stream-like benchmark without particle sorting or any grid on both
Polaris and SuperMUC.
On Polaris, it   
yields almost twice the performance compared to SuperMUC (Table
\ref{table:memory-throughput-polaris}).
This is a direct result of the higher clock frequency on the smaller cluster.

\begin{figure}[!ht]
\centering
  \includegraphics[width=0.44\linewidth]{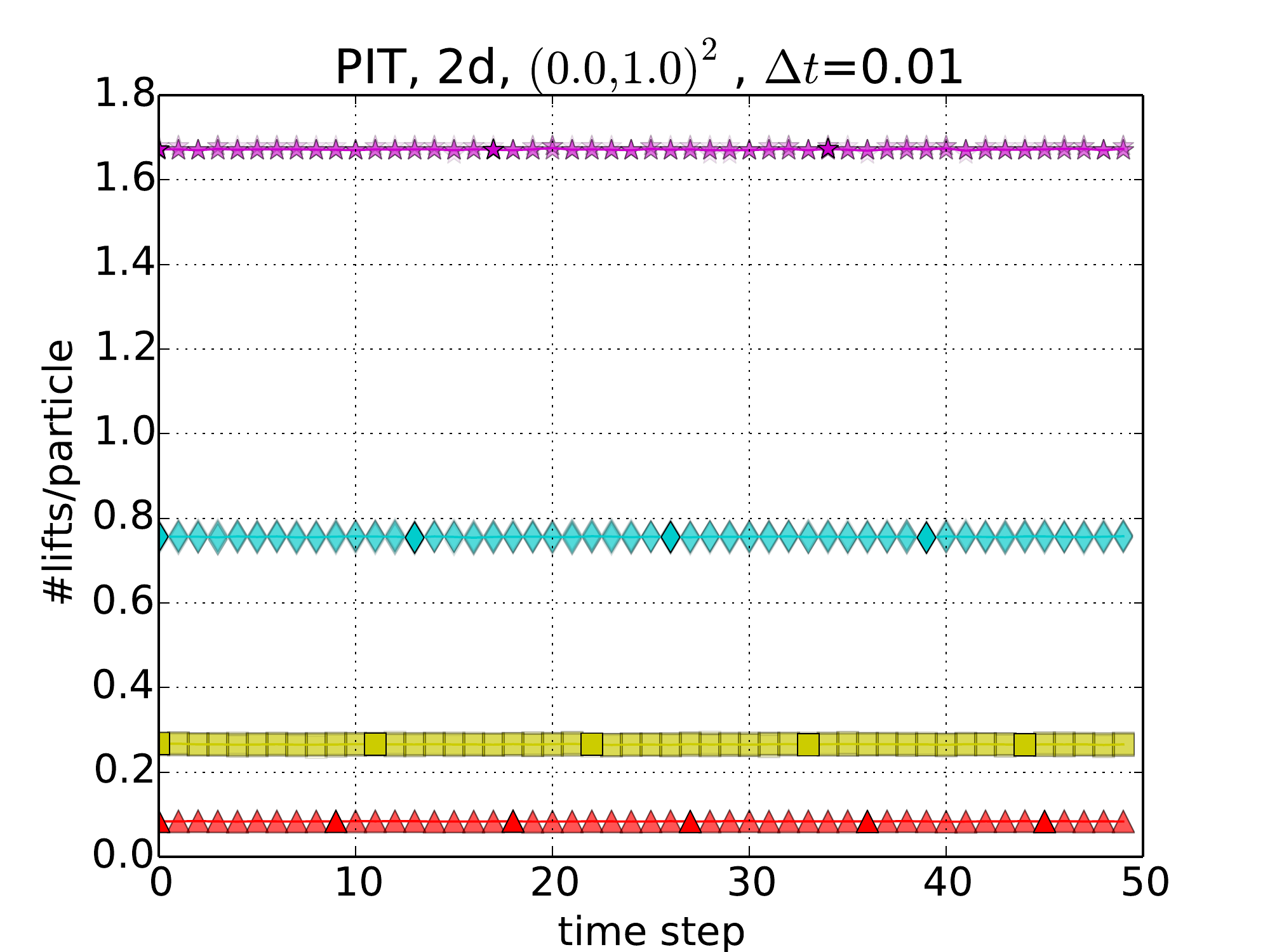}
  \includegraphics[width=0.44\linewidth]{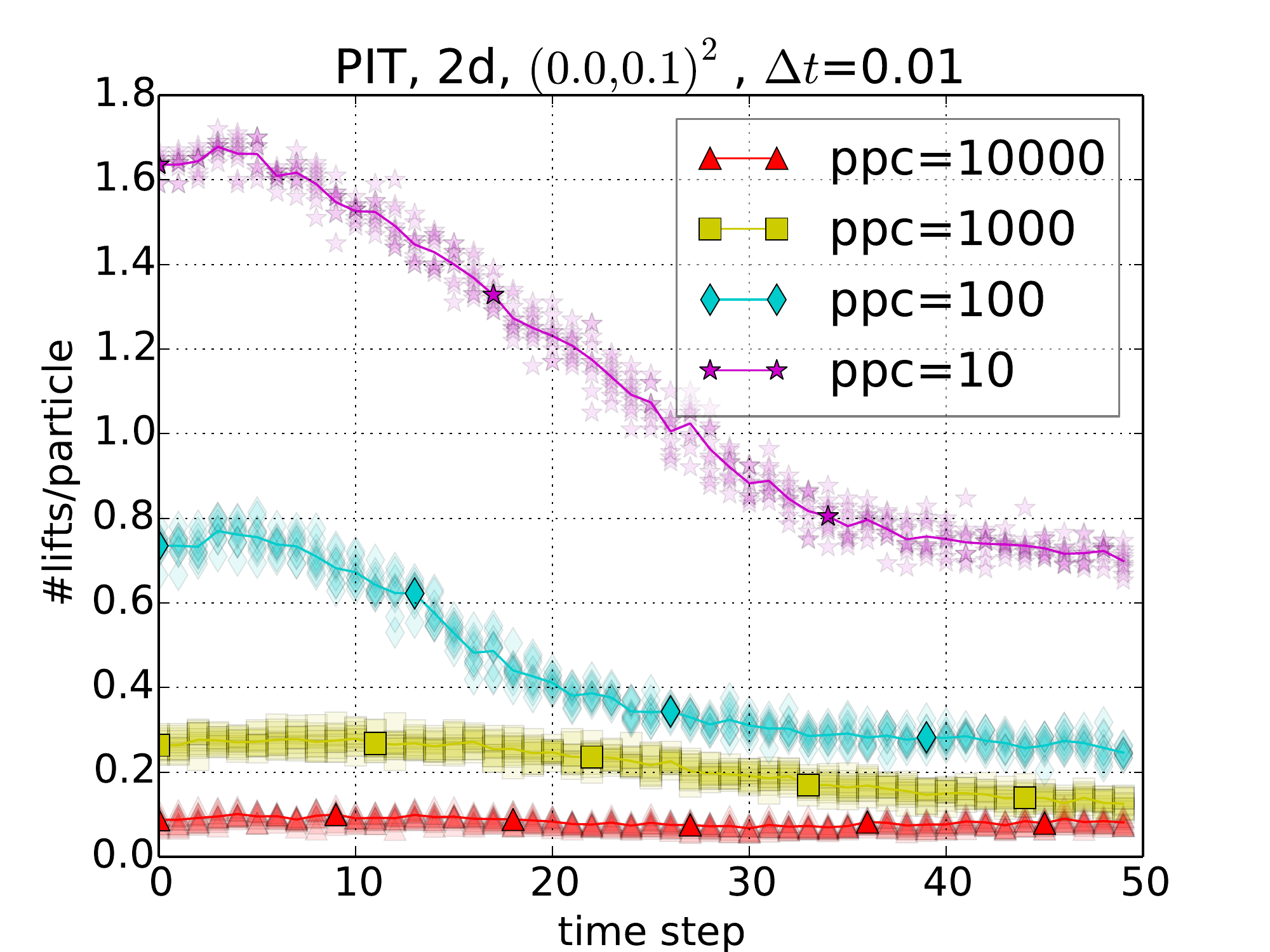}
  \includegraphics[width=0.44\linewidth]{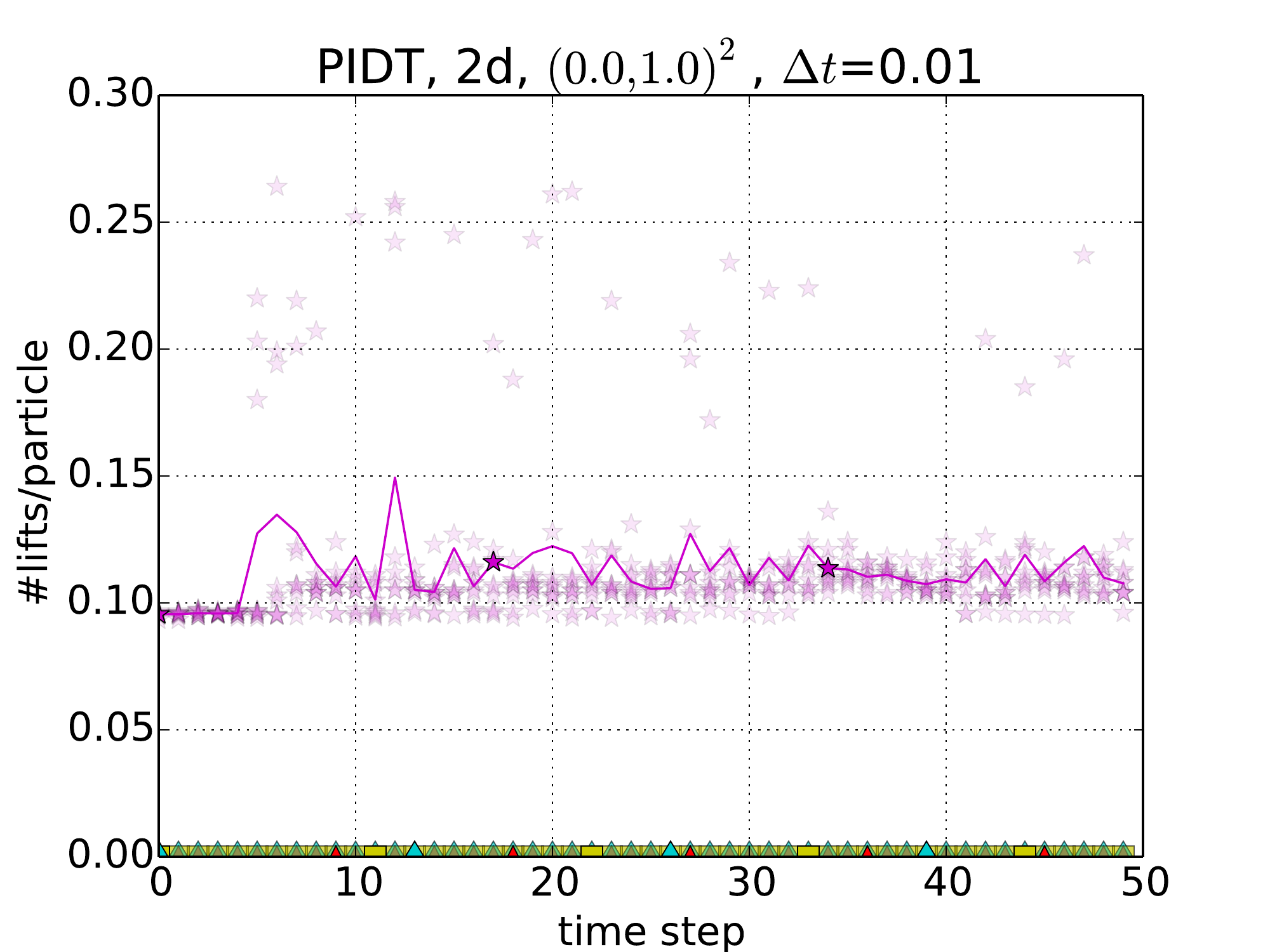}
  \includegraphics[width=0.44\linewidth]{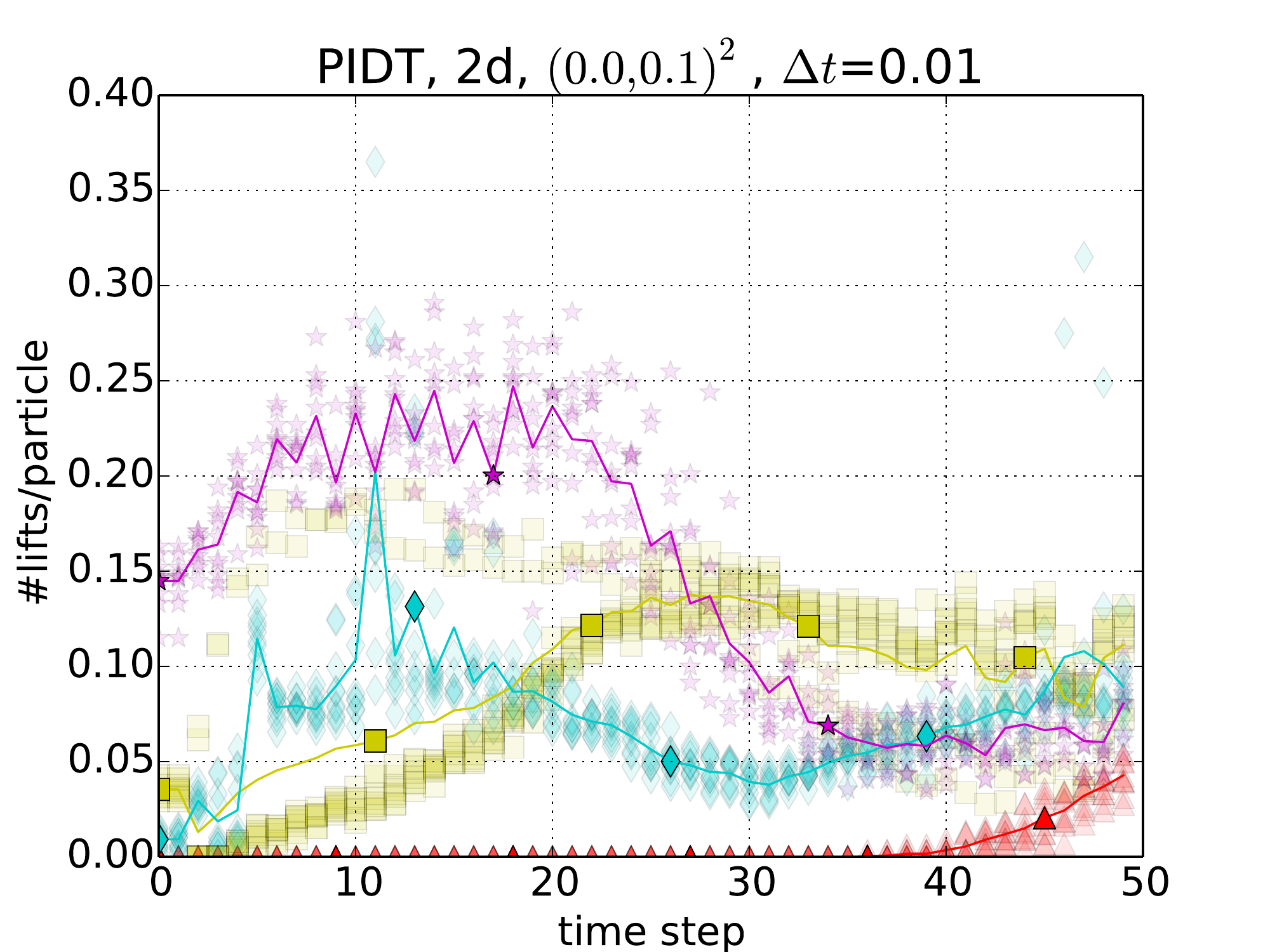}
  \caption{
    $\Delta t=0.01$ for \pit\ (top) and \pidt\ (bottom).
    We restrict to $d=2$ and study homogeneous (left) and
    inhomogeneous (right) initial distributions, i.e.~regular and dynamically
    adaptive grids.
%
%
%
  }
\label{figure:lifts-per-particle-with-different-ppc}
\end{figure}

Finally, some additional plots illustrating the lift behaviour for a fixed time step size can
be found in Figure \ref{figure:lifts-per-particle-with-different-ppc}.
They validate our statements on pros and cons of \pit\ or \pidt\ respectively.

\end{document}